\documentclass[12pt]{article}
\usepackage{amsmath}
\usepackage{graphicx,psfrag,epsf}
\usepackage{enumerate}
\usepackage{natbib}
\usepackage{multicol}
\usepackage{graphics}
\usepackage{amssymb}
\usepackage{graphicx}
\usepackage{bbm}
\usepackage{url}  
\usepackage{hyperref}
\usepackage{xcolor}
\usepackage{amsmath,latexsym,amssymb,amsthm}
\newtheorem{theorem}{Theorem}
\newtheorem{lemma}{Lemma}

\newtheorem{cor}{Corollary}
\hypersetup{colorlinks=true,citecolor=blue,linktocpage=true}

\newtheorem{assumption}{Assumption}
\begin{document}

\def\spacingset#1{\renewcommand{\baselinestretch}%
{#1}\small\normalsize} \spacingset{1}

{
  \title{A semi-parametric model for assessing the effect of temperature on ice accumulation rate from Antarctic ice core data}
  \author{Radhendushka Srivastava\hspace{.2cm}\\
    Department of Mathematics, IIT Bombay India\\
   and \\
    Debasis Sengupta  \\
    Applied Statistics Unit, ISI Kolkata, India
    }
  \maketitle
} 

\begin{abstract}
In this paper, we present a semiparametric model for describing the effect of temperature on Antarctic ice accumulation on a paleoclimatic time scale. The model is motivated by sharp ups and downs in the rate of ice accumulation apparent from ice core data records, which are synchronous with movements of temperature. We prove strong consistency of the estimators under reasonable conditions. We conduct extensive simulations to assess the performance of the estimators and bootstrap based standard errors and confidence limits for the requisite range of sample sizes. Analysis of ice core data from two Antarctic locations over several hundred thousand years shows a reasonable fit. The apparent accumulation rate exhibits a thinning pattern that should facilitate the understanding of ice condensation, transformation and flow over the ages. There is a very strong linear relationship between temperature and the apparent accumulation rate adjusted for thinning. 
\end{abstract}

\noindent%
\vfill

\newpage
\spacingset{1.9} 
\section{Introduction}\label{s1}
An important question in the context of climate change is how the global stocks of frozen water at high elevations and the polar regions would hold up against the warming climate. Reports of melting glaciers from various parts of the planet have raised the stakes for the answer to this question in respect of the largest mass of ice, namely the Antarctic ice sheet (see \cite{IPCC} for a summary of recent findings in this regard). While there is empirical evidence \citep{Allan2008,Algarra2020} as well as theoretical explanations \citep{Wang2013} of greater transportation of moisture in the atmosphere resulting from higher temperatures, there is no clear evidence of the same translating into excess precipitation that can offset the loss of ice from the Antarctic continent. Recent claims of such an offset based on satellite radar and laser altimetry \citep{zwally2015} have since been contested \citep{scambos2016}. A major source of confusion in this regard is the large margin of error in altimetry based estimates. Here we look for an answer from the climatological history of the Earth over several hundred thousand years. 

Ice core data records are the only source of information on the long-term linkage between temperature and ice accumulation. It is important to note that the reconstructed age and temperature corresponding to a specific ice core are based on concentration of isotopic elements extracted in the laboratory. Reconstruction of the temperature is a matter of continued research \citep{Markle}, as errors may occur during the process of extraction of isotopic material as well as reconstruction of age and temperature of the ice cores. There is added complication arising from adjustments of reconstructed ages for alignment with the available domain knowledge. Further, ever since a layer of ice is formed, it goes through a process of thinning under the weight of subsequently deposited layers. A relationship between temperature and ice accumulation over the years needs a careful investigation that takes into account the thinning of ice sheets.       

In this article, we explore this relationship by using ice core data from two locations in Antarctica, Lake Vostok and Dome C of the European Project for Ice Coring in Antarctica (EPICA). The temperature of the different slices of the ice cores is estimated from the concentration of Deuterium isotopes present in the core using dating techniques. The ages of the ice layers at Lake Vostok and EPICA Dome C were originally documented in the GT4 scale \citep{Petit1999} and the EDC3 scale \citep{Parrenin2007}, respectively. The GT4 scale utilizes an ice flow model (aided by isotopic stratigraphic control points at 1534 m (110 kyr) and at 3254 m (390 kyr)) under the assumption that the rate of accumulation of ice varies proportionally to the derivative of the water vapour saturation pressure, which is itself dependent on temperature \citep{Siegert2003}. This procedure may be susceptible to inducing distortions of the role of temperature. The EDC3 scale is also closely related. In order to avoid this possibility, we use for both the data sets the AICC2012 time scale, which is more broadly based \citep{bazin2013, VERES}.  

For ice core data, the temperature anomaly, which is the deviation from long-term average temperature, is reported in degree Celsius, and the age of ice cores is expressed in Kilo Years Before Present (KYrBP, the ``present" being the year 1950). An overview of the two data sets, compiled for uniform intervals of depth, is given in Table~\ref{DATA_info}.
\begin{table}[h]
  \begin{center}
  \caption{Overview of the two ice core data sets \citep{Petit1999,Jouzel2007,Rapp2019}}\label{DATA_info}
  \medskip
    \begin{tabular}{lrrrrrr}
  \hline
  Location & \multicolumn{1}{l}{Age}  & \multicolumn{1}{l}{Depth} & \multicolumn{1}{l}{Base altitude}& \multicolumn{1}{l}{Sample} & \multicolumn{1}{l}{Sampling}\\
  & \multicolumn{1}{l}{limit}  & \multicolumn{1}{l}{limit} & \multicolumn{1}{l}{above MSL}& \multicolumn{1}{l}{size} & \multicolumn{1}{l}{interval}\\
  &\multicolumn{1}{l}{(KYrBP)}    &\multicolumn{1}{l}{(meters)} &\multicolumn{1}{l}{(meters)}& &\multicolumn{1}{l}{(meters)}\\
  \hline
  Lake Vostok  & 403.780 & 3263 &$-500$&3311& 1.00\\
  ($77^\circ 50'$S $106^\circ 00' $E)& & & & &\\
  EPICA Dome C  & 801.588 &  3189.45 &$-40$&  5800  & 0.55 \\
  ($75^\circ 06' $S $123^\circ 20' $E) & & & & &\\
  \hline
   \end{tabular}
  \end{center}
\end{table}

The age difference at successive depths indicates the number of kilo years of net ice accumulation represented by the particular range of depths. We define the apparent accumulation rate (AAR) as the ratio of the height difference of successive ice core samples (measured in meters) and the age difference (in KYrBP) of these samples, as a function of the average age of the two successive layers (in KYrBP). Let $\mathsf{d}_i$, $\mathsf{z}_i$, and $\mathsf{x}_i$ denote the depth, the age and the temperature anomaly, respectively, of the $i$th observation of the ice core for $i=1,2,\ldots,n+1$, where $n+1$ is the total number of observations. The AAR over the age interval $(\mathsf{z}_i,\mathsf{z}_{i+1})$, represented by the average age $z_i=(\mathsf{z}_i+\mathsf{z}_{i+1})/2$ at  the average depth $d_i=(\mathsf{d}_i+\mathsf{d}_{i+1})/2$, is
\begin{equation}
b_i = \frac{\mathsf{d}_{i+1}-\mathsf{d}_i}{\mathsf{z}_{i+1}-\mathsf{z}_i},\quad i=1,2,\ldots,n.    
\end{equation}
This rate corresponds to the age $z_i$ and the average of the temperature anomalies at the end-points of the slice intervals, i.e., $x_i=(\mathsf{x}_i+\mathsf{x}_{i+1})/2$ (henceforth referred to as `temperature').

\begin{figure}[!h]
\centering
\begin{tabular}{cc}
 {\underline{Lake Vostok}}&{\underline{EPICA Dome C}}\\
\includegraphics[width=3in,height=2in]{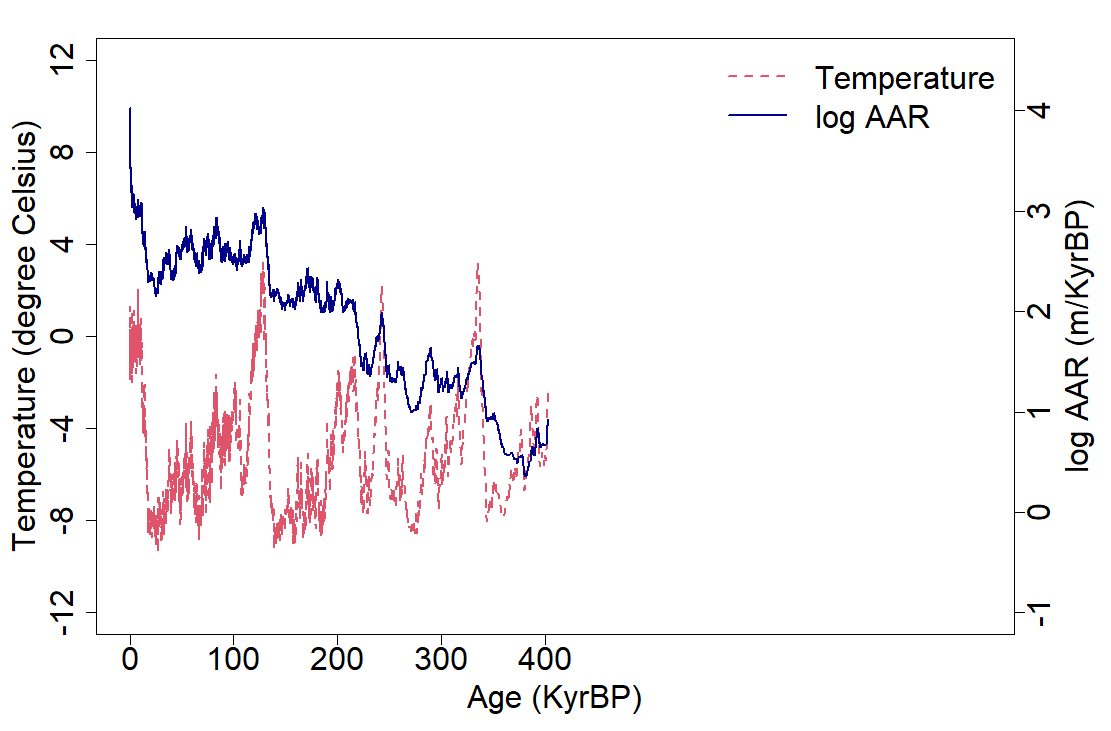}&
\includegraphics[width=3in,height=2in]{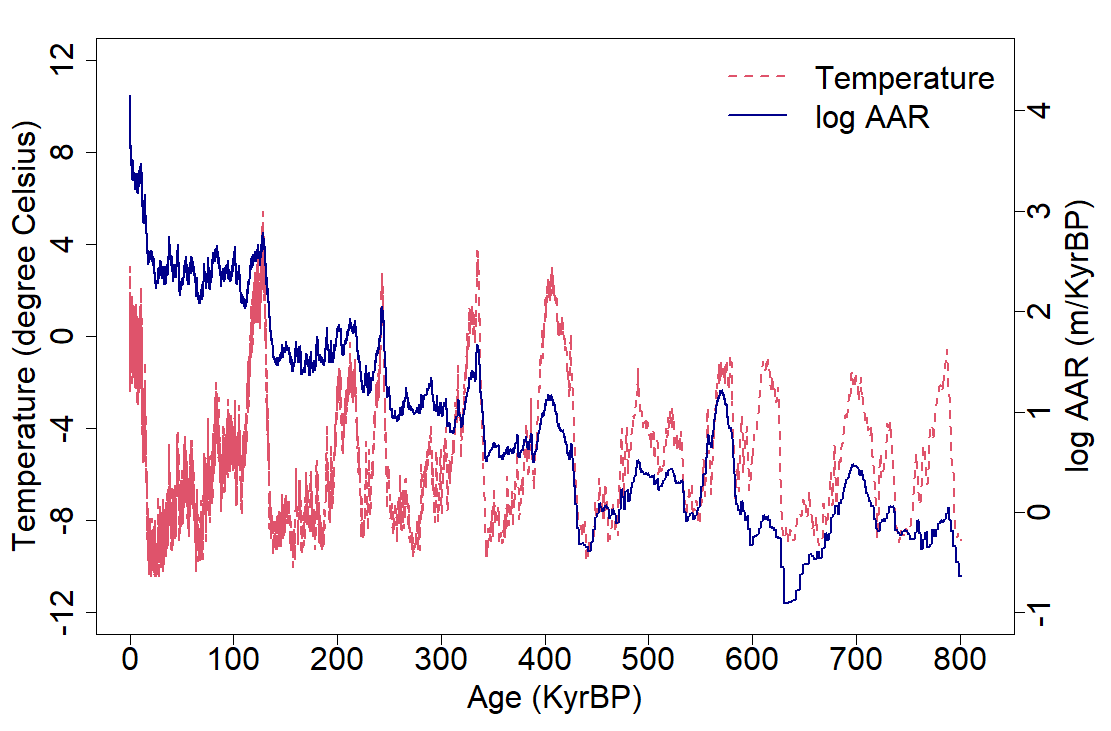}\\
\end{tabular}
\caption{The plots of temperature (in {}$^\circ$C, red curve, labeled along the left vertical axis) and logarithm of AAR (in m/KYr, blue curve and labeled along the right vertical axis) versus age (KYrBP)}
\label{Fig2}
\end{figure}

Figure~\ref{Fig2} shows the plots of temperature (on the left vertical axis) and AAR (in log scale, on right vertical axis)  against age at the two locations. The plots indicate that the fluctuations of AAR with those of the temperature move in the same direction together with a general long-term decreasing trend in AAR for both the data sets.

The reason for our focus on the \textit{apparent} accumulation rate is that the actual rate of ice accumulation cannot be observed over a paleo-climatic scale of time. AAR is the net accumulation observed on site today as the overall effect of several factors like precipitation, thinning under the pressure from the weight of ice lying above (compression factor) through relocation/expulsion of gas and water molecules trapped in ice, and lateral flow of ice. There have been efforts to delineate the actual accumulation rate from thinning, and the results of that exercise are available in the public domain (\url{https://www.ncei.noaa.gov/access/paleo-search/study/15076}). We choose not to separate the two, so as to avoid unnecessary imprecision in assessing the linkage between temperature and accumulation.
 
In Section~\ref{s2}, we present a multiplicative model for AAR that combines a decreasing function of age with a factor that depends linearly on the temperature. We also present an algorithm for estimating the parameters. In Section~\ref{s3}, we establish strong consistency of the estimators under appropriate conditions. In Section~\ref{s4}, we provide a model based bootstrap method to estimate the standard error and the confidence limits of the parameters. In Section~\ref{s5}, we report the findings of a simulation study to validate the applicability of the theoretical results as well as  the bootstrap method to the requisite sample size. In Section~\ref{s6}, we present the analysis of both the ice core data sets and discuss the implications of the findings.  
Proofs of all the theoretical results are given in the supplementary material.

\section{The model and the method}\label{s2} We aim to capture the movement of AAR, observed in Figure~\ref{Fig2}, by the semi-parametric and multiplicative model
\begin{equation}\label{raw_model}
b_i=(1+\gamma_0 x_i) \ g_0(z_i) \  \xi_i, \quad\mbox{for} \quad  i=1,2,\ldots,n,
\end{equation}
where $\gamma_0$ is the parameter associated with the temperature ($x$), $g_0$ is a positive valued, continuous and decreasing function of the age ($z$), and $\xi$ is a sequence of positive valued disturbance term. We presume that $z_i\le  A$, $|x_i|\le x^{(M)}$ and $(1+\gamma_0 x_i)>0$ and  for some positive numbers $A$ and $x^{(M)}$, for $i=1,2,\ldots,n$. 

In the model \eqref{raw_model}, the AAR and temperature are related in a linear manner. This presumption is based on indications from findings of various glacial studies over shorter periods of time \citep{Fudge2016,Yang}. The parameter $\gamma_0$ is the rate of increase in AAR with $1^\circ$C rise in temperature. The function $g_0$ is meant to capture the long-term decreasing trend visible in AAR (see Figure~\ref{Fig2}). We regard the function $g_0$ as a nuisance parameter that represents the nominal rate of accumulation at temperature $x=0$, subject to thinning. Specifically, the thinning factor at a specific age $z$ is $g_0(z)/g_0(0)$. The disturbance term $\xi$ represents the possible error in indirect measurements, and various other atmospheric phenomena that affect AAR but are not included in the model \eqref{raw_model}. Semiparametric models like~\eqref{raw_model} are commonly used in various application areas, though mostly in an additive set-up \cite{ruppert2003semiparametric}.

By allowing $g_0$ to be a general decreasing function, we are able to free the analysis from uncertainties in the determination of age. As mentioned in the foregoing section, there have been multiple attempts to determine the scale of time accurately. The different scales of time may be regarded as monotonically increasing transformations of one another that preserve the order among ages estimated for successive slices of the ice core. In particular, if the reconstructed age $z_i$ is related to the correct age $z_i^*$ through the relationship $z_i^*=g^*(z_i)$ for $i=1,2,\ldots,n$ and some function $g^*$, and the model~\eqref{model} holds with $z_i^*$ in place of $z_i$, then it can be written as $b_i=(1+\gamma_0x_i)g_0(g^*(z_i))\xi_i$. If $g^*$ is an increasing function, then this model has exactly the same form as~\eqref{model} with $g_0\circ g^*$ taking the place of $g_0$. Since this function is allowed to be determined by the data (subject to the constraint of monotonicity), fitting the model to data containing the $z_i$'s is equivalent to fitting it to data containing the $z_i^*$'s. 

We consciously left out an intercept parameter from the linear factor of the model~\eqref{raw_model}. Inclusion of an intercept parameter leads to confounding of the linear factor with $g_0$. An alternative way of ensuring identifiability of the model is to impose a boundary condition like $g_0(0)=1$, which makes $g_0$ the thinning factor. However, that strategy would make all the estimators vulnerable to the well known issues of functional estimation at the boundaries~\citep{Hardle1990}.

The log transformation of model \eqref{raw_model} produces the additive semiparametric model
\begin{equation}\label{model}
y_i=\log(1+\gamma_0 x_i)+\log g_0(z_i)+\varepsilon_i,\quad i=1,2,\ldots,n,
\end{equation}
where $y_i= \log{b_i}$ and $\varepsilon_i=\log{\xi_i}$. The values of AAR at successive slices of the ice core make use of common values of depth and age at the boundary. This construction induces dependence in the AAR sequence. Dependence is also expected to arise from the use of the AICC2012 scale of time, where the ages are subjected to the constraint of monotonicity and are calibrated with other information. In order to address this dependency, we model the noise sequence $\varepsilon$ as a sample from an Ornstein-Uhlenbeck process (see Assumption~\ref{Aerror} below).

The parameters of model \eqref{model} are estimated by minimizing the average squared error loss function
\begin{equation}\label{Square_loss}
 Q_n(\gamma,g)= \frac1{n}\sum_{i=1}^n(y_i-\log(1+\gamma x_i)-\log g(z_i))^2,
\end{equation}
with respect to the scalar parameter $\gamma$ and the functional parameter $g$. The parameter of main interest is $\gamma$, while $g$ is a nuisance parameter. The optimization is made over the space $\Gamma\times \mathcal{G}$, where $\Gamma=[-(1-\lambda)/x^{(M)},(1-\lambda)/x^{(M)}]$ for some positive fraction $\lambda$ (which ensures that $1+\gamma x_i\ge \lambda$ for $i=1,2,\ldots,n$) 
and $\mathcal{G}$ is a set containing positive valued, non-increasing and continuous functions defined over the interval $[0, \ A]$.

We have chosen the simple criterion \eqref{Square_loss}, which does not exploit the presumed dependence in the errors, in order that the method can be shown to be consistent. 

An iterative two stage coordinate descent algorithm is adapted for minimizing $Q_n$ with respect to $\gamma$ and $\log g$. 
The initial iterate for $\gamma$ is $\hat{\gamma}_n^{(0)}=0$. Substitution of this value in~\eqref{Square_loss} leads to the initial iterate of $g$, 
\begin{eqnarray}\label{eq:iter1}
\log\hat{g}^{(0)}_n=\mathop{\arg\min}_{g\in \mathcal{G}} \frac{1}{n}\sum_{i=1}^{n}\left(y_i-\log g(z_i)\right)^2.
\end{eqnarray}
The values of $\log \hat g^{(0)}_n$ at the requisite points may written explicitly as  
\begin{equation}
 \log \hat{g}_n^{(0)}(z_i) = \min_{1\leq j\leq i} \ \max_{i\leq l\leq n}\frac{y_j+\ldots + y_l}{l-j+1},
\end{equation}
and obtained by the pool-adjacent-violators algorithm (PAVA, \cite{RWR,guntuboyina2018}). The $k^{\rm th}$ iterate of $\gamma$ is
\begin{eqnarray}
\hat{\gamma}_n^{(k)}=\mathop{\arg\min}_{\gamma\in \Gamma} \frac{1}{n}\sum_{i=1}^{n}\left(y_i-\log \hat{g}^{(k-1)}(z_i)-\log (1+\gamma x_i)\right)^2\quad k\ge1.\label{bar_Qn}
\end{eqnarray}
The Gauss-Newton algorithm is used to minimize \eqref{bar_Qn}. The $k^{\rm th}$ iterate of $g$ is 
\begin{eqnarray}\label{eq:iter4}
\log\hat{g}_n^{(k)} =\mathop{\arg\min}_{g\in \mathcal{G}} \ \  \frac{1}{n}\sum_{i=1}^{n}\left(y_i-\log \left(1+\hat{\gamma}_n^{(k)} x_i\right)-\log g(z_i)\right)^2,\quad k\ge1,\label{tilde_Qn}
\end{eqnarray}
where $\log\hat{g}^{(k)}$ is obtained by using PAVA on the sequence $y_i-\log \left(1+\hat{\gamma}_n^{(k)} x_i\right)$, $i=1,2,\ldots,n$.

The iterations \eqref{bar_Qn} and \eqref{tilde_Qn} are carried out alternately. We denote the limiting values of ${\hat\gamma}_n^{(k)}$ and $\log\hat{g}_n^{(k)}$ by ${\hat\gamma}_n$ and $\log\hat{g}_n$, respectively, and presume linear interpolation of the latter in between the locations of estimation. The iterations are continued until the relative change in the average squared error loss criterion in successive iterations is less than the desired threshold. 

As $\log g_0$ is presumed to be a smooth function, a smoothed version of the estimator $\log\hat{g}_n$ is desirable. We propose the kernel smoothed estimator
 \begin{eqnarray}\label{tilde_g_est}
        \log\tilde{g}_{n}(z)&=& \frac{\frac1{nh}\sum_{i=1}^n K\left(\frac{z-z_i}{h}\right)\log\hat{g}_n(z_i) }{\frac1{nh}\sum_{i=1}^n K\left(\frac{z-z_i}{h}\right)}\mbox{ for } z>0, 
\end{eqnarray}
where $K(\cdot)$ is the kernel function and $h$ is the bandwidth.
Based on $\log \tilde{g}_{n}$, we update the estimator of $\gamma_0$ as
\begin{eqnarray}\label{tilde_gamma_est}
\tilde{\gamma}_n&=&\mathop{\arg\min}_{\gamma\in \Gamma} \frac{1}{n}\sum_{i=1}^{n}\left(y_i-\log \tilde{g}_n(z_i)-\log (1+\gamma x_i)\right)^2.
\end{eqnarray}
\section{Strong consistency of estimators}\label{s3}
The parameter space of $\gamma$ is $\Gamma$, a compact subset of $\mathbb{R}$. The space of $g$ is $\mathcal{G},$\footnote{A compact subset of $C[0,A]$ is the set of non-negative valued, non-increasing and continuous functions defined over the interval $[0,A]$, with uniformly bounded left and right derivatives (compactness follows from the Arzel\`{a}–Ascoli theorem).} a compact subset of $C[0,A]$ (set of real valued continuous functions defined over the interval $[0,A]$) such that the elements of $\mathcal{G}$ are positive valued, non-increasing and uniformly bounded. We set the parameter space of the model as $\Gamma\times \mathcal{G}$.   
Consequently, the parameter space $\Gamma\times \mathcal{G}$ with the \textit{product distance function}  
\begin{equation*}
d((\gamma_1,g_1),(\gamma_2,g_2))=| \gamma_1-\gamma_2 |+\sup_{z\in[0,A]} |g_1(z)-g_2(z)|  \quad\mbox{for} \ \  \gamma_1,\ \gamma_2\in \Gamma \ \mbox{and} \ g_1, \ g_2\in \mathcal{G}
\end{equation*}
is a complete separable metric space.
The existence of the estimators for such a parameter space follows from arguments similar to those in \cite{jennrich1969}. 

For proving almost sure convergence of the criterion and strong consistency of the estimators, we make some assumptions on the sequences of age ($z$), temperature ($x$), and errors ($\varepsilon$). 

\begin{assumption}\label{Aage}
The numbers $z_1<z_2<\cdots<z_n$ lying within the interval $[0,A]$ are the values of a Lipschitz continuous function  $z(d)$ at uniformly spaced depths $0=d_1<d_2<\cdots<d_n=\Delta$ with $d_i=\frac{i\Delta}{n}$, for $i=1,2,\ldots,n$ and $\Delta$ is a positive number. 
\end{assumption}

\begin{assumption}\label{Atemp}
The numbers $x_1,x_2,\ldots,x_n$ lying within the interval $(-x^{(M)},x^{(M)})$ are the values of a Lipschitz continuous function $x(d)$ at uniformly spaced depths $0=d_1<d_2<\cdots<d_n=\Delta$ with $d_i=\frac{i\Delta}{n}$, for $i=1,2,\ldots,n$ and $\Delta$ is a positive number. 
\end{assumption} 

\begin{assumption}\label{Aerror}
The error sequence $\varepsilon_1,\varepsilon_2,\ldots, \varepsilon_n$ are uniformly spaced samples of  a stationary Ornstein-Uhlenbeck (OU) process
$\{\varepsilon(t), t\in[0,\infty)\}$ of the form 
\begin{eqnarray}
    \varepsilon(t)=e^{-\beta t} \varepsilon(0)+\sigma\int_{0}^t e^{-\beta(t-u)} dW(u),\label{OU_sol}
\end{eqnarray}
where $\beta,\sigma>0$, $\{W(u), u\in[0,\infty]\}$ is a standard Wiener process and $\varepsilon(0)$ is independent of $W(t)$ for all $t$ and has the invariant distribution of the OU process, namely $ N(0,\frac{\sigma^2}{2\beta})$.
Further, the sampling interval $\rho_n$ of the errors $\varepsilon_i=\varepsilon(i\rho_n)$ for $i=1,2,\ldots,n$ is such that $\rho_n\to 0$  
as $n\to \infty$ and $\sum_{n=1}^\infty e^{-cn\rho_n}<\infty$ for all $c>0$. 
\end{assumption}

Assumptions~\ref{Aage} and~\ref{Atemp} are equivalent to the continuity of age and temperature, respectively, over depth. The OU process is a simple model for allowing dependence in the observations of AAR appearing on the left side of the model~\eqref{model}. The condition $\sum_{n=1}^\infty e^{-cn\rho_n}<\infty$ implies that $n\rho_n\to \infty$ as $n\to\infty$, which corresponds to an increasing range of sampling ($n\rho_n$) with increasing sample size ($n$). The pair of conditions on $\rho_n$ is weaker than the condition $\rho_n\propto n^{-\kappa} $ for some positive fraction $\kappa$.  
All the indices for $d, \ x, \ z$ and $\varepsilon$ depend on $n$, but we omit it for notational simplicity. 

The asymptotic regime used here is somewhat different from the regimes used for increasingly frequent samples from an underlying continuous time signal-plus noise process over a fixed time span \citep{robinson1997large} or for continuous time processes sampled over an increasing span of time \citep{kessler1997,lu2019}. Here, the signal is sampled over a fixed span, but the span of the noise samples expands with sample size. We adopt this asymptotic regime to ensure that the correlation between the observations at any two fixed depths go to zero with sample size going to infinity, which we expect to happen from the indications obtained from the data at hand (see the discussion on Figure 6). This set-up requires modification of available results also.

\begin{theorem}\label{thm1}
Under Assumptions~\ref{Aage}-\ref{Aerror}, the average squared error loss function $Q_{n}(\gamma,g)$ defined in~\eqref{Square_loss} converges uniformly to $Q(\gamma,g)+\frac{\sigma^2}{2\beta}$ with probability 1 as $n\to \infty$,
where 
$$
Q(\gamma,g)=\frac{1}{\Delta}\int_{0}^\Delta[\log (1+\gamma_0 x(w))-\log (1 +\gamma x(w))+\log g_0(z(w))-\log g(z(w))]^2dw.
$$
\end{theorem}

Note that $Q(\gamma,g)\ge0$ for every $\gamma\in\Gamma$ and $g\in\mathcal{G}$ and that $Q(\gamma_0,g_0)=0$.

\begin{theorem}\label{thm2}
If $Q(\gamma,g)$ has a unique minimum at $(\gamma_0,g_0)$, and Assumptions~\ref{Aage}-\ref{Aerror} hold, then $(\hat{\gamma}_n,\hat{g}_n)$ converges almost surely to $(\gamma_0,g_0)$ as $n\to\infty$. 
\end{theorem}

In view of the fact that $\log$ is a monotone function, we have the following corollary to Theorem~\ref{thm2}. 
\begin{cor}\label{Cor1}
    If the conditions of Theorem~~\ref{thm2} hold, then $\sup_{z}|\log\hat{g}_n(z)-\log g_0(z)|$ converges to $0$ with probability $1$.
\end{cor}
For establishing almost sure convergence of the smoothed estimator, we make the following assumption on the kernel function.
\begin{assumption}\label{Kern}
The kernel $K>0$ is a differentiable function supported on $(-1,1)$.   
\end{assumption}

\begin{theorem}\label{thm3}
Let $g_0$ be a positive valued and continuously differentiable function bounded away from 0. Then under Assumptions~\ref{Aage}-\ref{Kern}, as $h\to 0$ and $n\to\infty$, we have the following.
    \begin{enumerate}[(i)]
        \item $\mathop{\sup}_{z}|\log(\tilde{g}_{n}(z))-\log(g_0(z))| \to 0$ \ almost surely. 
        \item $\tilde{\gamma}_{n} \to \gamma_0$ \ almost surely.
  \end{enumerate}
\end{theorem}

For the purpose of resampling, described in the next section, we have to estimate the error parameters mentioned in Assumption~\ref{Aerror}. The estimation of parameters of stochastic diffusion processes based on increasingly faster sampling of the processes has been well studied \citep{ait2002,lu2019}. Here, we estimate the parameters of the error process based on residuals of the model, which can be viewed as samples of the OU process with additive noise that diminishes as the sample size increases. By using the set of residuals 
\begin{eqnarray}\label{tilde_eps_esst}
\tilde{\varepsilon}_i&=&y_i-\log (1+\tilde{\gamma}_n x_i)-\log \tilde{g}_n(z_i), \mbox{ for } i=1,2,\ldots,n, 
\end{eqnarray}
we use the method of moment  estimators of $\beta$ and $\sigma^2$ as
\begin{eqnarray}
\tilde{\beta}_n&=&-\frac1{\rho_n}\log\left(\frac{\frac1n\sum_{i=1}^{n-1}\tilde{\varepsilon}_i\tilde{\varepsilon}_{i+1}}{\frac1n\sum_{i=1}^n\tilde{\varepsilon}_i^2}\right),\label{beta_tilde}\\
\tilde{\sigma}_n^2&=&2\tilde{\beta}_n\frac1n \sum_{i=1}^n \tilde{\varepsilon}^2_i.\label{sig_tilde}
\end{eqnarray}
The next theorem establishes the  consistency of the estimators $\tilde{\beta}_{n}$ and $\tilde{\sigma}_{n}^2$ over the parameter space $(0,\infty)$.
\begin{theorem}\label{thm4}
If the assumptions of Theorem~\ref{thm3} hold, and additionally $n\rho_n^2\to \infty$ as $n\to\infty$, then we have the following.
\begin{enumerate}[(i)]
    \item $\tilde{\beta}_{n} \to \beta$    in probability.
    \item $\tilde{\sigma}_{n}^2 \to \sigma^2$ in probability.
\end{enumerate}

\end{theorem}

\section{Bootstrap standard error and confidence limits}\label{s4}
Nonparametric and model based bootstrap methods provide a computational alternative to the calculation of sampling distribution of estimators, for obtaining standard errors \citep{Efron1993} and confidence intervals \citep{Efron1996}. While the bootstrap was originally proposed for independent and identically distributed data, model based methods for dependent data, pertinent for the present work, have been studied in subsequent literature (see \cite{Lahiri2003}). An obstacle to model based bootstrap in the present case is the PAVA estimator for isotonic regression, which is known to lead to overfitting  \citep{Nicules2005,Luss2012}, and consequently to systematically underestimated magnitudes of residuals. 
For the analogous problem of monotonic density estimation, \cite{Bodhi2010} showed that the bootstrap confidence interval of the Grenander estimator is inconsistent, and provided a consistent procedure obtained by smoothing the Grenander estimator, which averts the pitfalls of resampling from excessively small residuals. It is for this reason that the smoothed nonparametric estimator $\log\tilde{g}_n$ was chosen for specification of variance by bootstrap.

The following are the suggested steps for generating model-based bootstrap replicates of $\tilde{\gamma}_n$ and $\log\tilde{g}_n$.
\begin{enumerate}[(i)]
    \item Estimate the error innovation sequence as $$\tilde{\zeta}_i=\tilde{\varepsilon}_i-e^{-\tilde{\beta}_n\rho_n} \tilde{\varepsilon}_{i-1}, \mbox{ for } i=2,3,\ldots,n,$$ where $\tilde{\varepsilon}_i$ and $\tilde{\beta}_n$ are as defined in \eqref{tilde_eps_esst} and \eqref{beta_tilde}, respectively.
    \item Draw an innovation sample of size $n+1000$ by using simple random sampling with replacement from $\{\tilde{\zeta}_i, \ i =1,3,\ldots,n\}$ (see Chapter~2 of \cite{Lahiri2003}),
    and denote it by $\{\zeta_{-999}^*, \zeta_{-998}^*,\ldots, \zeta_n^*\}$.
    \item Compute a bootstrap replicate of AAR from model \eqref{model} as $$y_i^*=\log(1+\tilde{\gamma}_n x_i) +\log\tilde{g}_n(z_i)+\tilde{\varepsilon}_i^*,\qquad i=2,\ldots,n,$$ 
    where $\log\tilde{g}_n$ and $\tilde{\gamma}_n$ are as defined in \eqref{tilde_g_est} and \eqref{tilde_gamma_est}, respectively, and $\tilde{\varepsilon}_i^*$ is computed from the recursion $\tilde{\varepsilon}_i^* =e^{-\tilde{\beta}_n\rho_n}\tilde{\varepsilon}_{i-1}^*+\zeta_i^*$, after starting with an initial value of zero and discarding the first 1000 values of the recursion.
    \item Estimate $\log\tilde{g}_n^*$ and $\tilde{\gamma}_n^*$ by substituting $y_i^*$ in place of $y_i$ in \eqref{tilde_g_est} and \eqref{tilde_gamma_est} respectively. 
\end{enumerate}

The above steps are repeated to produce the bootstrap replicates $\log\tilde{g}_n^*$ and $\tilde{\gamma}_n^*$ independently a large number of times. The bootstrap estimates of the standard errors of $\tilde{\gamma}_n$, and $\log\tilde{g}_n(z)$ are the sample standard deviations of the corresponding bootstrap replicates. For any given value of $\alpha$, the middle $(1-\alpha)\%$ of the relevant bootstrap replicates would constitute bootstrap confidence intervals of $\tilde{\gamma}_n$ and  $\log\tilde{g}_n(z)$ for any $z$.  

\section{Simulation Study}\label{s5}
In order to understand the finite sample performance of the methods presented in Sections~\ref{s3} and~\ref{s4}, we run some simulations under conditions chosen to mimic the circumstances of the ice core data. The sequence of depths is presumed to be a uniform grid over the interval $0$ meters to $\Delta$ meters, for sample sizes $n=750, \ 1500, \ 3000$ and $6000$. The choices $n=3000$ and $n=6000$ are similar to the sample sizes of the ice core data from Lake Vostok and EPICA Dome C, respectively, and the other two choices extend the explored range of sample sizes. The grid spacing is $\frac{\Delta}{n}$. We choose the age $z_i$ (in KYrBP) and temperature $x_i$ (in degree Celsius) at depth $d_i=\frac{i\Delta}{n}$ by linear interpolation of the EPICA Dome C series over the range of depths 40 to 3000 meters. We choose the effect of temperature as $\gamma_0=0.06$ and the nominal accumulation rate at temperature $x=0$ as the exponentially decaying function $g_0(z)=25 \ e^{-z/150}+1$. We choose the parameters of the underlying error process $\varepsilon$ as $\beta=1.5$ and $\sigma=0.15$. The sampling interval $\rho_n$ of the error process is chosen as $n^{-0.35}$, which is similar to the data driven heuristic choices for the two data sets, mentioned in Section~\ref{s6}. For the kernel smoothed estimator $\log\tilde{g}_n$, we choose the Epanechnikov kernel
\begin{eqnarray*}
    K(w)&=& \begin{cases}
     \frac34 \left(1-w^2\right)  & \mbox{ if } |w|\le 1 \\
     0&\mbox{ Otherwise.}
    \end{cases}
\end{eqnarray*}

All the empirical measures of performance are computed on the basis of 1000 simulation runs. The bootstrap estimates reported in Sections~\ref{s5.1} and~\ref{s5.3} are based on 1000 sets of re-samples for each simulation run.

\subsection{Choice of bandwidth}\label{s5.1}
There are generic considerations for choosing the bandwidth $h$ of the kernel function, which is the smoothing parameter~\citep{Hardle1990}. However, generic considerations do not always work well in some specific situations~\citep{GhoshChau2004}. The focus of this article is on modeling and estimation of the effect of temperature on AAR for the ice core data, along with a reliable measure of its uncertainty. Therefore, we studied the coverage proportion of bootstrap confidence intervals of $\gamma$ and $\log g$ at various choices of bandwidth $h$, and sought a bandwidth that can produce reasonable coverage under simulated conditions designed to match the circumstances of the data sets.

For this purpose, we simulated data sets as described above, fitted the model~\eqref{model} and drew 1000 bootstrap samples from the residuals in the manner described in Section~\ref{s4}. As mentioned in that section, the range of the middle 95\%\ of the bootstrap replicates of $\tilde\gamma_n$ is the bootstrap confidence interval of $\gamma_0$ with nominal coverage probability $1-\alpha=0.95$. We noted for each simulation run the actual inclusion status of the correct value of the parameter ($\gamma_0=0.06$), and repeated the process 1000 times. We observe the inclusion status of $\log g_0(z_i)$ in the middle 95\% of the bootstrap replicates of $\log \tilde{g}_n(z_i)$ in the same manner, for $i=2,\ldots,n$. In initial simulations, a pattern of under-coverage was observed, which was traced to underestimation of the error variance. (In other words, the corrective measure of \cite{Bodhi2010} mentioned in Section~\ref{s4} does not completely address the well-known issue of bias in this case.) In order to improve coverage, the following heuristic step was inserted after step (ii) of the bootstrap procedure described in Section~\ref{s4}: If the sample variance of $(\zeta_i^*,  \ i = 2,\ldots,n)$ is smaller than that of $\{\tilde{\zeta}_i, \ i =2,3,\ldots,n\}$, repeat step (ii); else, proceed to the next step. The left panel of Figure~\ref{h_search} shows the proportion of times, out of the 1000 runs, that the bootstrap confidence interval included the correct value of the parameter $\gamma_0$, for a range of values of bandwidths $h$ and sample sizes $n=750, \ 1500,\ 3000$ and $6000$. The right panel shows the 5\%\ trimmed mean of the coverage proportions of $\log g_0(z_2)$, $\log g_0(z_3)$, $\ldots,$ $\log g_0(z_n)$ for the same values of $h$ and $n$. The values of $h$ were varied over a grid of size 2 over the interval $[10,30]$.  
\begin{figure}
\centering
\begin{tabular}{cc}
\includegraphics[width=2.5in,height=1.8in]{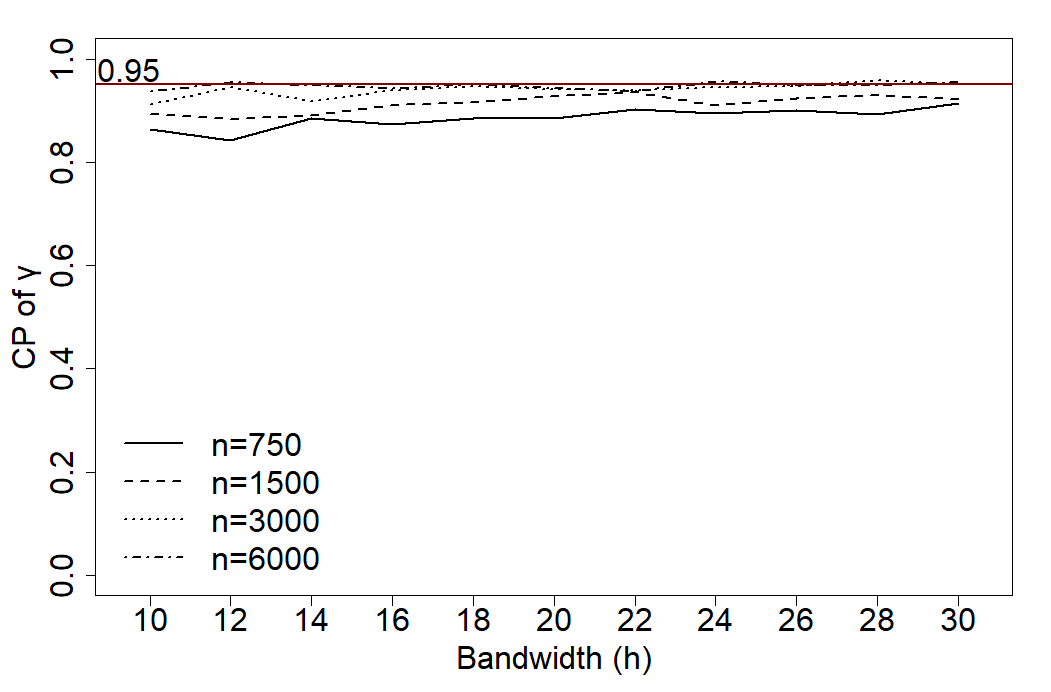}&
\includegraphics[width=2.5in,height=1.8in]{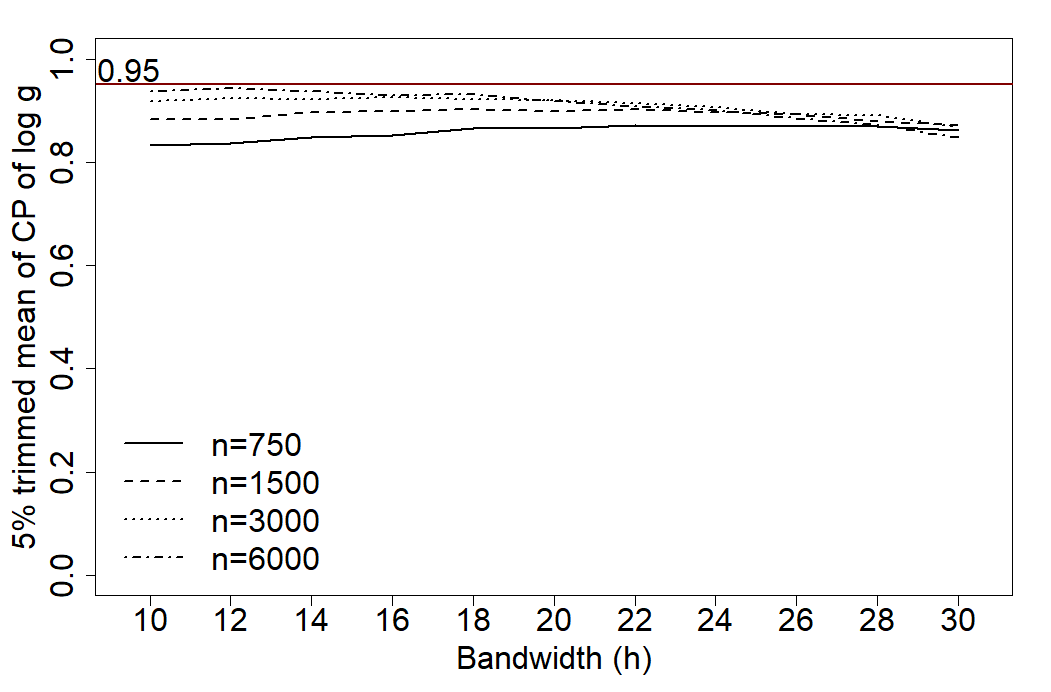}
\end{tabular}
\caption{Plots of coverage proportion (CP) for $\gamma_0$ (left panel) and 5\% trimmed mean of CP of $\log g_0$ (right panel) for different values of the bandwidth $h$. }
\label{h_search}
\end{figure}

It is seen from Figure~\ref{h_search} that the coverage proportions of $\gamma_0$ and $\log g_0$ are lower than the nominal coverage probability for smaller sample sizes and, in the case of $\log g_0$, the coverage proportion reduces for larger $h$. The graphs of the coverage proportions for $\gamma_0$ are somewhat flat. The peak of the graph of the 95\%\ trimmed mean of coverage proportions for $\log g_0$ occurs at different values of $h$ (ranging from 14 to 22) for different sample sizes. In the next two subsections we report simulation results for $h=14$ for all the four choices of $n$.

\subsection{Finite sample performance of estimators across runs}\label{s5.2}

Table~\ref{tab2} shows the empirical bias and standard deviation (sd) of the estimators $\tilde{\gamma}_n$, $\tilde{\beta}_n$, and $\tilde{\sigma}_n$. Evidently, all the biases and standard deviations reduce with increase in the sample size, as one expects for consistent estimators. The bias generally tends to be smaller in magnitude than the standard deviation, except in the case of $\tilde{\beta}_n$. The issue of small sample bias of $\tilde{\beta}_n$ has been noted in the literature, and a bootstrap based remedial measure has also been proposed \citep{Tang2009}. However, in trial simulation runs, the solution (previously reported for much smaller values of $\beta_0$ than the value used here) did not produce useful reduction of bias, or noticeable impact on the coverage proportions of the confidence intervals of the parameters of interest. 

\begin{table}[!h]
  \caption{Empirical bias and standard deviation (sd) of $\tilde{\gamma}_n$, $\tilde{\beta}_n$ and $\hat{\sigma}_n$ based on $1000$ simulation runs}
  \vskip0.1in
  \centering
  \begin{tabular}{|l|l|rcr|}
    \hline
     $n$ & Attribute & \multicolumn{1}{c}{$\tilde{\gamma}_n$} &
     \multicolumn{1}{c}{$\tilde{\beta}_n$} &
     \multicolumn{1}{c|}{$\tilde{\sigma}_n$}\\
     && \multicolumn{1}{c}{(${\gamma}_0=0.06$)} &
     \multicolumn{1}{c}{(${\beta}_0=1.5$)} &
     \multicolumn{1}{c|}{(${\sigma}_0=0.15$)}\\
    \hline
    750 & bias & $-0.310\times10^{-3}$ & 1.044 & $9.954\times10^{-2}$\\
     & sd &$1.556\times10^{-3}$ & 0.358 & $0.438\times10^{-2}$\\
     \hline
     1500 & bias & $-0.280\times10^{-3}$ & 0.646 & $0.193\times10^{-2}$\\
     & sd & $1.469\times10^{-3}$ & 0.240 & $0.302\times10^{-2}$\\
     \hline
     3000 & bias & $-0.320\times10^{-3}$ & 0.406 & $0.125\times10^{-2}$\\
     & sd & $1.336\times10^{-3}$ & 0.164 & $0.199\times10^{-2}$\\
     \hline 
     6000 & bias & $-0.220\times10^{-3}$ & 0.252 & $0.060\times10^{-2}$\\
     & sd & $1.155\times10^{-3}$ & 0.119 & $0.141\times10^{-2}$\\
     \hline
  \end{tabular}\label{tab2}
\end{table}

The left and right panels of Figure~\ref{f2} show the plots of the estimates of pointwise bias and standard deviation, respectively, of the estimator $\log(\tilde{g}_n(z))$. Apart from the excessive errors at the boundaries, which is a well-known feature of nonparametric estimators of functions including isotonic regression estimators~\citep{dai2020bias}, the notable features of these plots are (a)~smaller magnitude of bias than standard error, (b)~reduction of both of these measures with increase in the sample size and (c)~increase of the standard error with age ($z$). 
\begin{figure}[!h]
\centering
\begin{tabular}{cc}
 bias: ($ \log \tilde g_n$)&
 sd: ($\log \tilde g_n$) \\
\includegraphics[width=3in,height=2in]{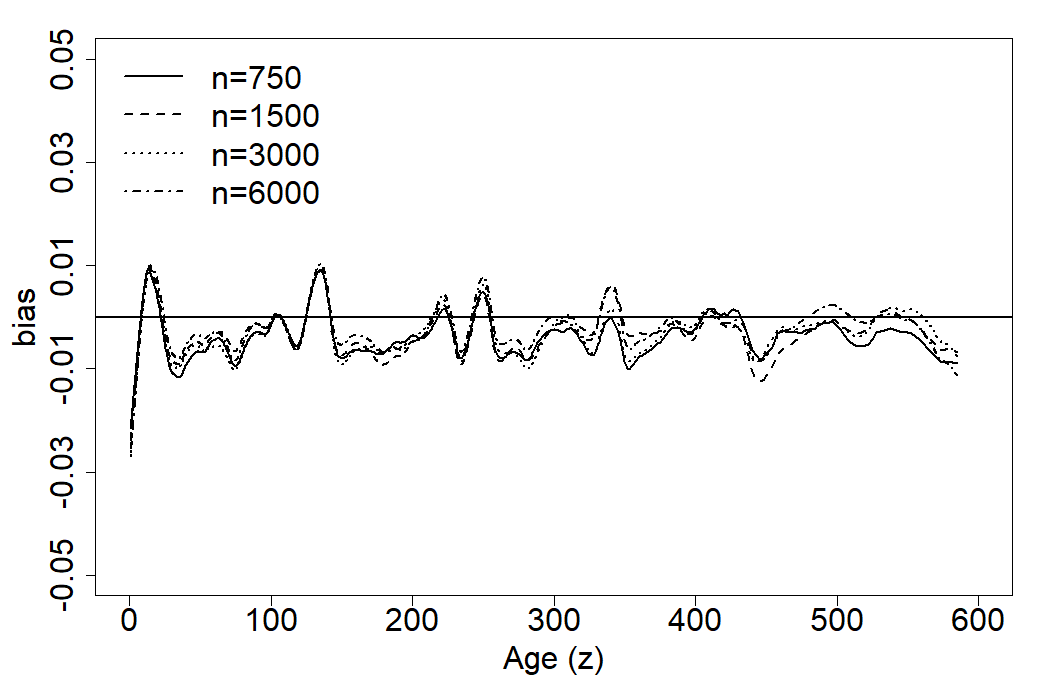}&
\includegraphics[width=3in,height=2in]{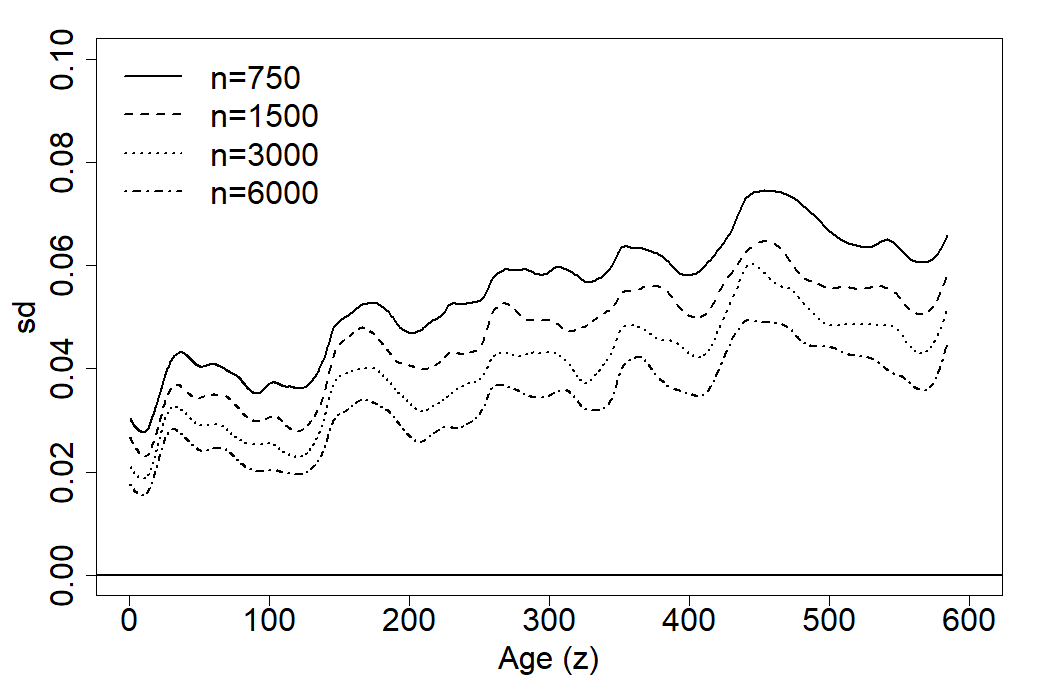} 
\end{tabular}
\caption{Plots of estimates of pointwise bias (left column) and sd (right column) of $\log(\tilde{g}_n(z))$.}
\label{f2}
\end{figure}

\subsection{Finite sample performance of bootstrap method }\label{s5.3}
Figure~\ref{f3} shows the histograms of bootstrap estimates of the standard error of $\tilde\gamma_n$ for sample sizes $n=750$ (top left) and $n=1500$ (top right), $n=3000$ (bottom left) and $n=6000$ (bottom right), respectively. The estimate of standard error computed directly from the multiple simulation runs (without bootstrap) is shown as a maroon vertical line on each plot. The histograms reflect a general pattern of underestimation and a greater concentration of the bootstrap estimates around the maroon line for larger sample sizes. The coverage proportions of the bootstrap confidence interval (intended for 95\% coverage), reported on the left side of each of the histograms, is close to the nominal coverage probability for larger sample sizes. For further insight on the coverage probability, the reader is referred to the supplementary material, where the behavior of $0.025$ and $0.975$ bootstrap quantiles of $\tilde{\gamma}$ is discussed.

\begin{figure}[!h]
\centering
\begin{tabular}{cc}
 $n=750$ & $n=1500$\\
\includegraphics[width=2.5in,height=1.8in]{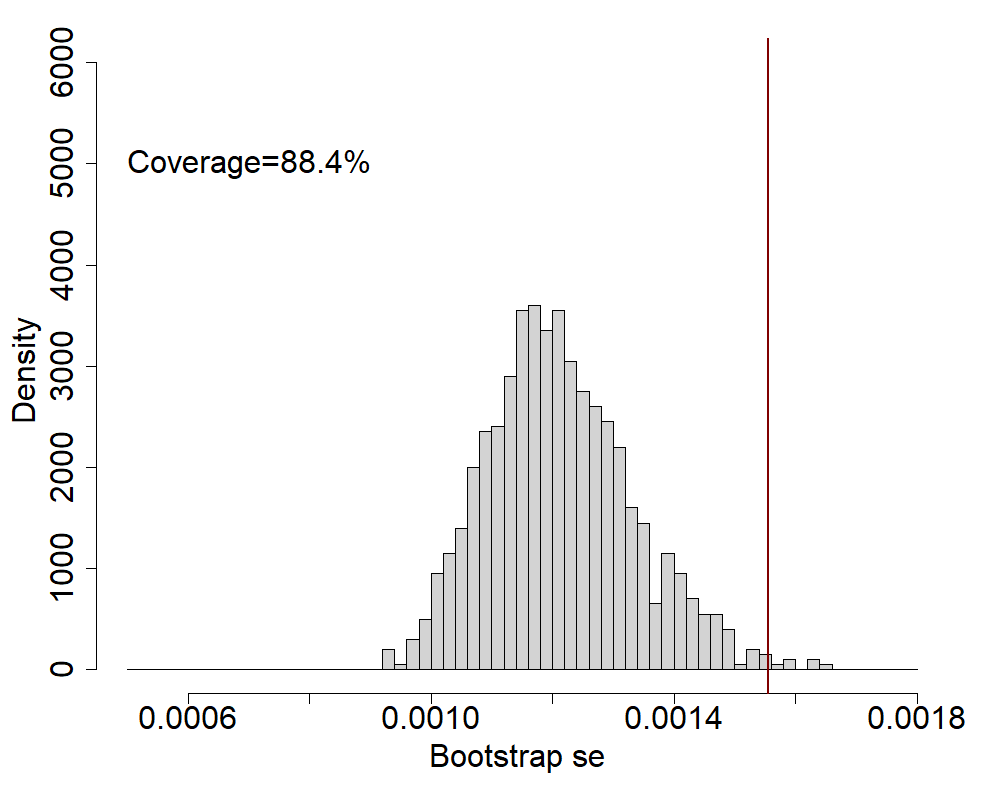}&
\includegraphics[width=2.5in,height=1.8in]{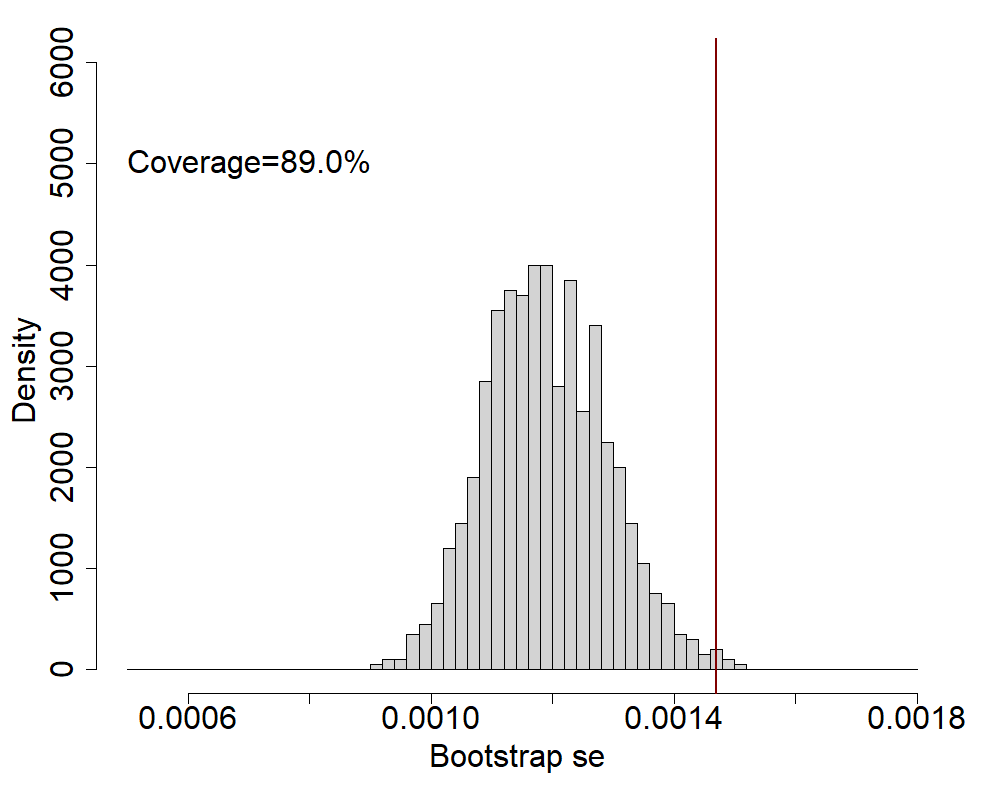}\\
 $n=3000$ & $n=6000$\\
 \includegraphics[width=2.5in,height=1.8in]{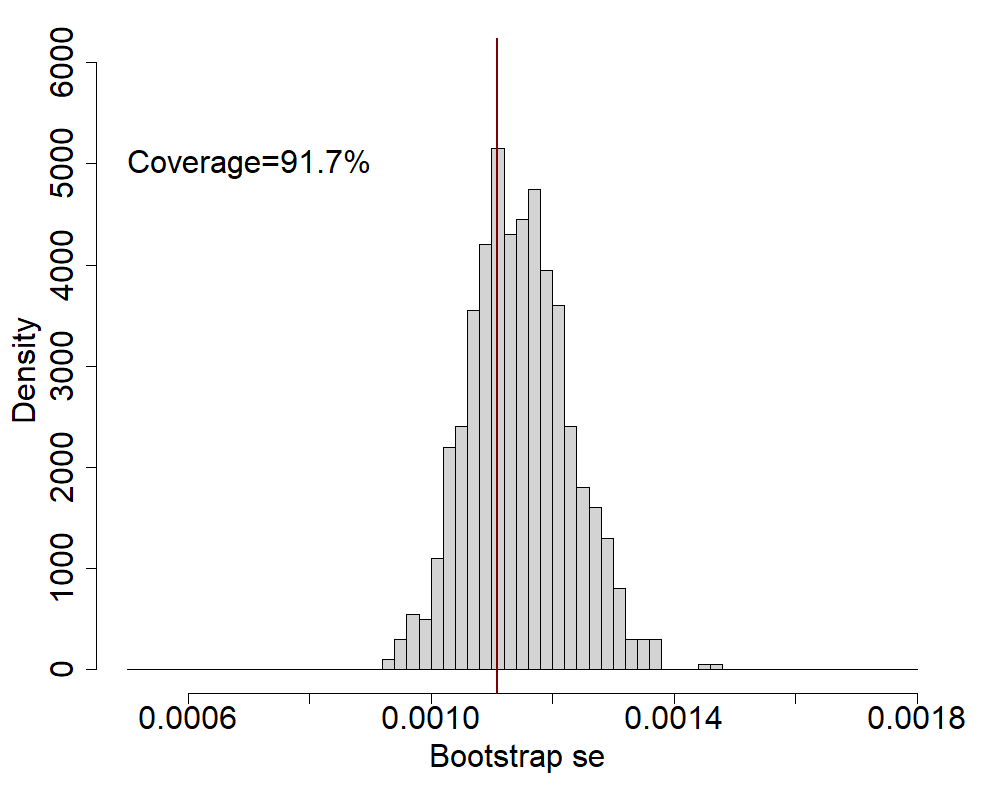}&
\includegraphics[width=2.5in,height=1.8in]{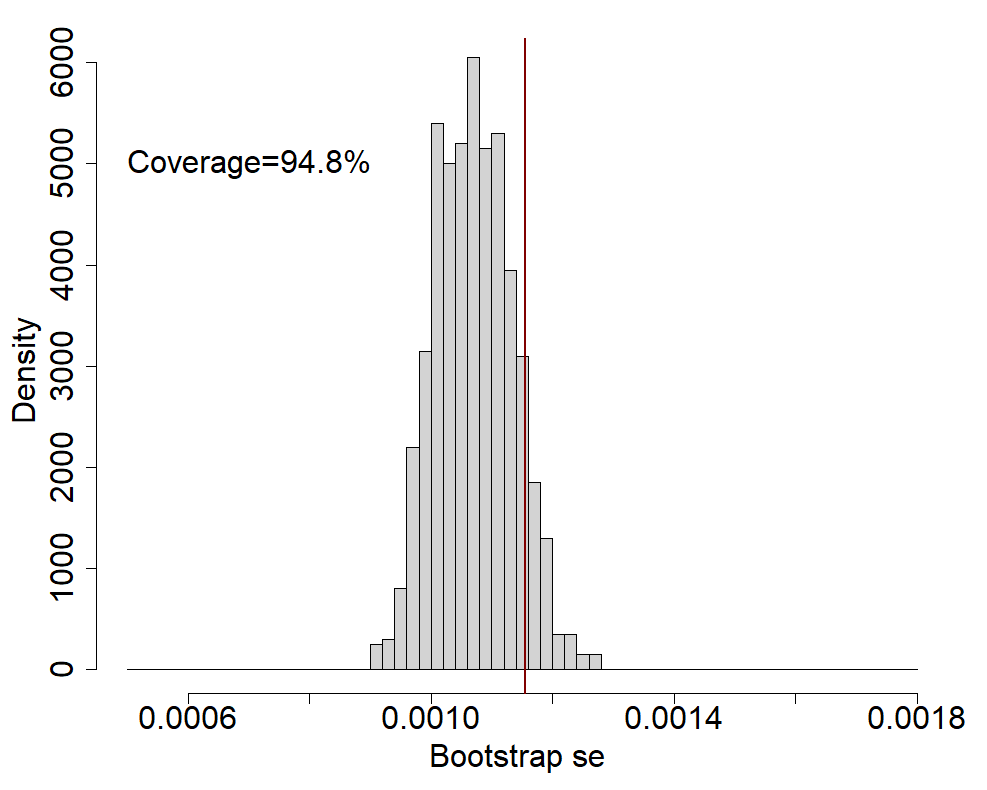}
\end{tabular}
\caption{Histograms of bootstrap estimates of the standard error of $\hat\gamma_n$ and coverage proportion of nominally 95\% coverage confidence intervals}
\label{f3}
\end{figure}

Figure~\ref{covergae_log_g} shows the pointwise coverage proportion of the bootstrap confidence intervals (with nominal coverage probability 0.95) of $\log{g}$. As the sample size increases, the pointwise coverage proportions approach the nominal level. The coverage proportions are lower than the nominal level around the boundaries even for larger sample sizes. This may have been due to excessive bias in the estimates in these areas, as observed in the section~\ref{s5.2}. The behavior of average pointwise confidence limits for $\log{g}$ and its standard errors are reported in the supplementary material.

\begin{figure}[!h]
\centering
\includegraphics[width=3in,height=2in]{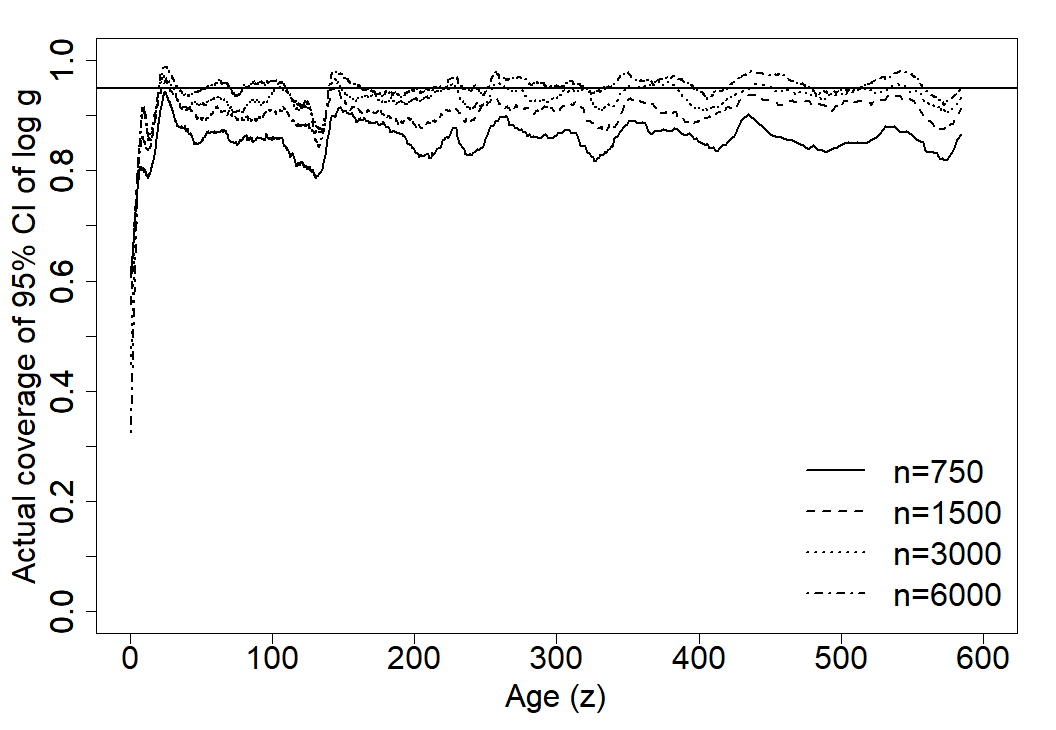}
\caption{Plots of coverage proportions of the pointwise bootstrap confidence intervals of $\log g$.}
\label{covergae_log_g}
\end{figure}

\newpage
\section{Analysis of Lake Vostok and EPICA Dome C Data}\label{s6}
For the analyses of the ice core data sets, we discard the data from depths less than 40 meters from the surface, in order to avoid excessive variation in ice density \citep{greve2009}. We choose the Epanechnikov kernel with bandwidth $h=14$, as in Subsection~\ref{s5.1}. 
The standard error of the estimators and the confidence limits are computed from $1000$ bootstrap samples. 

In order to select a working value of the sampling interval $\rho_n$ of the presumed error process, we use a number of sub-samples of the data at hand and compute $\tilde{\beta}_n\rho_n$ from the right hand side of \eqref{beta_tilde} for each subsample. It was shown in Section~\ref{s3} that $\tilde{\beta}_n$ is a consistent estimator of $\beta$. If $\rho_n$ is proportional to $n^{-\kappa}$ for some positive constant $\kappa$, then  $\log(\tilde{\beta}_n\rho_n)$ should have a linear relationship with $\log n$ with slope $-\kappa$. A scatter plot of $\log(\tilde{\beta}_n\rho_n)$ vs. $\log n$ for the different subsamples should have points around a straight line with slope $-\kappa$. We choose the subsamples (of approximate size $n/k$) by selecting every $k$th sample of the original series with $k$ different starting points (1, 2, \ldots, $k$), for $k=1,2,\ldots,8$. Figure~\ref{Fig_rho} shows the scatter plots for the data sets of Lake Vostok and EPICA Dome C, along with the respective least squares fitted lines. The scatters in both cases do show a downward pattern. The least squares estimates of $\kappa$ turn out to be 0.45 and 0.26 for Lake Vostok and EPICA Dome C, respectively. We use these values for the subsequent data analysis.

\begin{figure}
\centering
\begin{tabular}{cc}
 {\underline{Lake Vostok}} &{\underline{EPICA Dome C}}\\
\includegraphics[width=2.5in,height=1.8in]{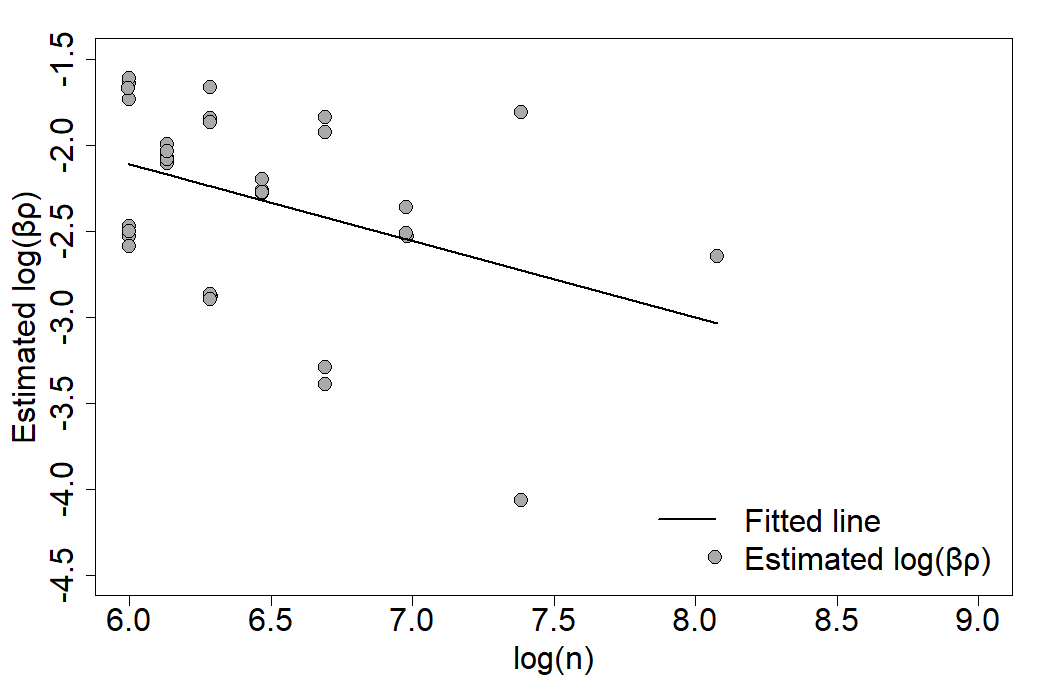}&
\includegraphics[width=2.5in,height=1.8in]{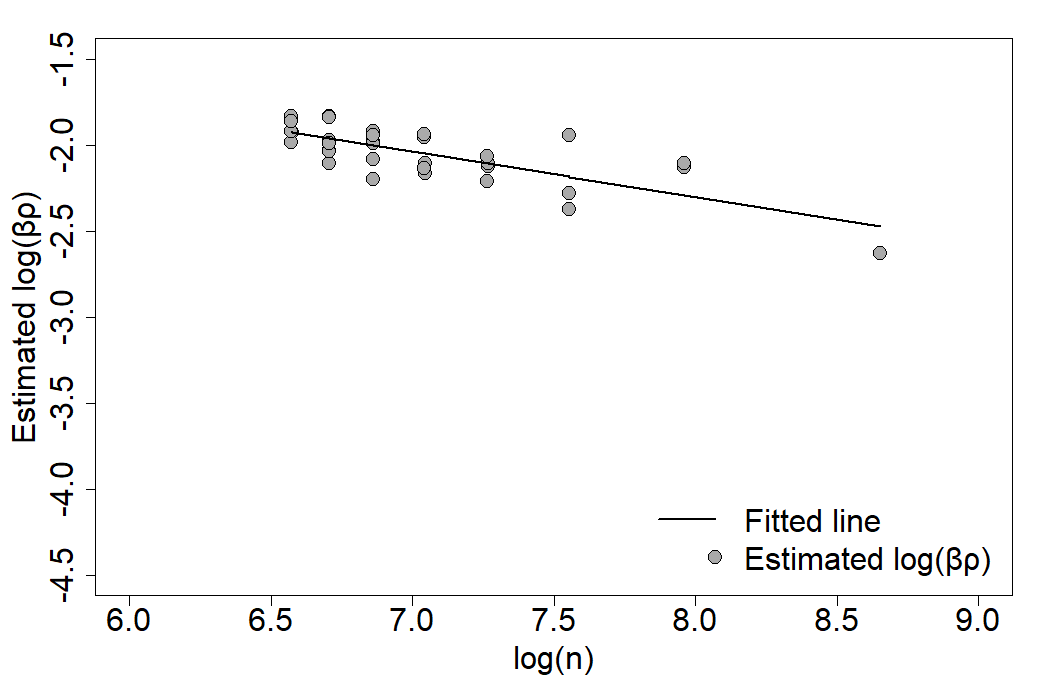}
\end{tabular}
\caption{The scatter plot of $\log(\tilde{\beta} \rho_n)$ against $n$ along with the fitted line, for subsamples from the data of Lake Vostok (left panel) and EPICA Dome C (right panel).}
\label{Fig_rho}
\end{figure}

Table~\ref{Tab4} shows the estimated values of the parameters $\gamma$, $\beta$ and $\sigma$ along with their respective bootstrap standard errors (s.e.) and the 95\%\ coverage bootstrap confidence intervals (c.i.), for both locations. The table indicates that the estimate of $\gamma$ for the two locations are similar. The confidence intervals are also of comparable width. 

\begin{table}[h]
  \caption{Estimates of the $\gamma$, $\beta$ and $\sigma$ with bootstrap s.e. and confidence interval.}
  \centering
  \begin{tabular}{|lll|}
    \hline
     & Lake Vostok & EPICA Dome C\\
     &($\rho_n=n^{-0.45}$) &($\rho_n=n^{-0.26}$)\\
     \hline
    $\tilde\gamma_n$ & 0.05747 & 0.06162\\
    Bootstrap s.e. ($\tilde\gamma_n$) & 0.00159& 0.00141\\
    95\% Bootstrap c.i. of $\gamma$ & (0.05431, 0.06042)& (0.05777, 0.06337)\\
        \hline
    $\tilde\beta_n$ & 2.68898  & 0.68419\\
    Bootstrap s.e. ($\tilde\beta_n$) & 0.30278 & 0.05545\\
    95\% Bootstrap c.i. of $\beta$ & (2.63187, 3.80233)& (0.68331, 0.90202)\\
     
        \hline
        $\tilde\sigma_n$ & 0.16847 & 0.12650\\
    Bootstrap s.e. ($\tilde\sigma_n$) & 0.00180 & 0.00109\\
    95\% Bootstrap c.i. of $\sigma_n$ & (0.16434,  0.17111)&(0.12629, 0.13063)\\
     
        \hline
  \end{tabular}\label{Tab4}
\end{table}

The plots of the estimated function $\log(\tilde{g})$ against age for the two locations, along with $95\%$ pointwise confidence intervals, are shown in the top panel of Figure~\ref{Fig4}. The plots suggest somewhat different thinning patterns at the two sites. This may have been caused not only by the differences in topography, but also by the fact that there is a lake underneath the ice sheet at one site and solid bedrock at the other. The confidence intervals get wider at older ages.

The plots of the observed and the fitted values of log(AAR) against Age for the two locations are shown in the bottom panel of Figure~\ref{Fig4}. The  plots indicate that the proposed model~\eqref{model} fits both the data sets reasonably well, though there is some misfit at older ages. 

\begin{figure}
\centering
\begin{tabular}{cc}
 {\underline{Lake Vostok}} &{\underline{EPICA Dome C}}\\
\includegraphics[width=2.5in,height=1.8in]{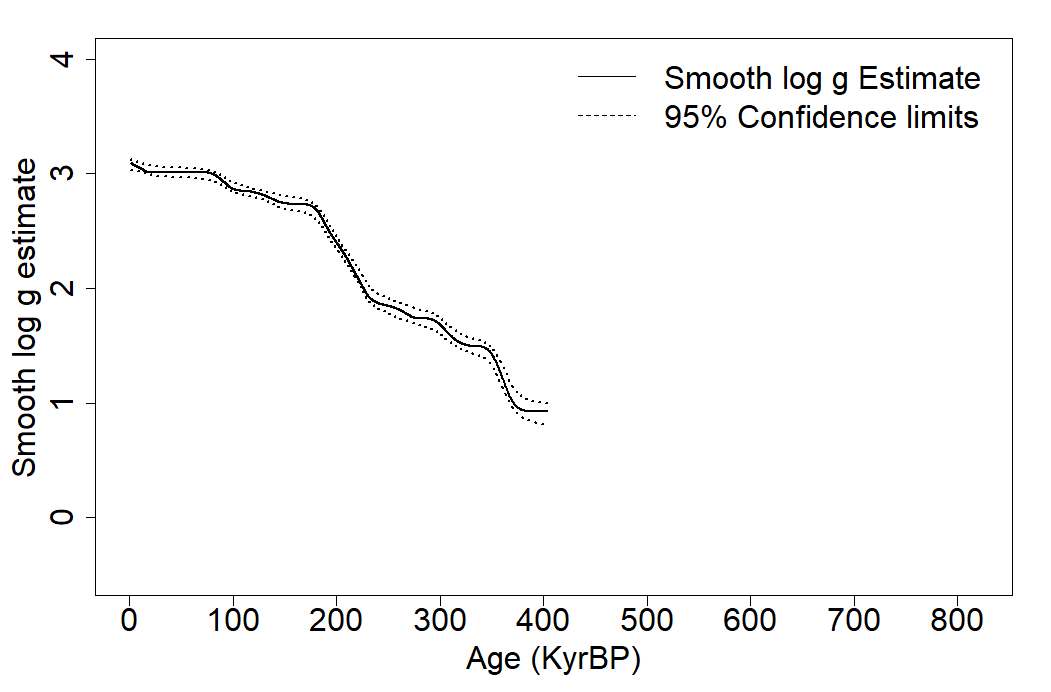}&
\includegraphics[width=2.5in,height=1.8in]{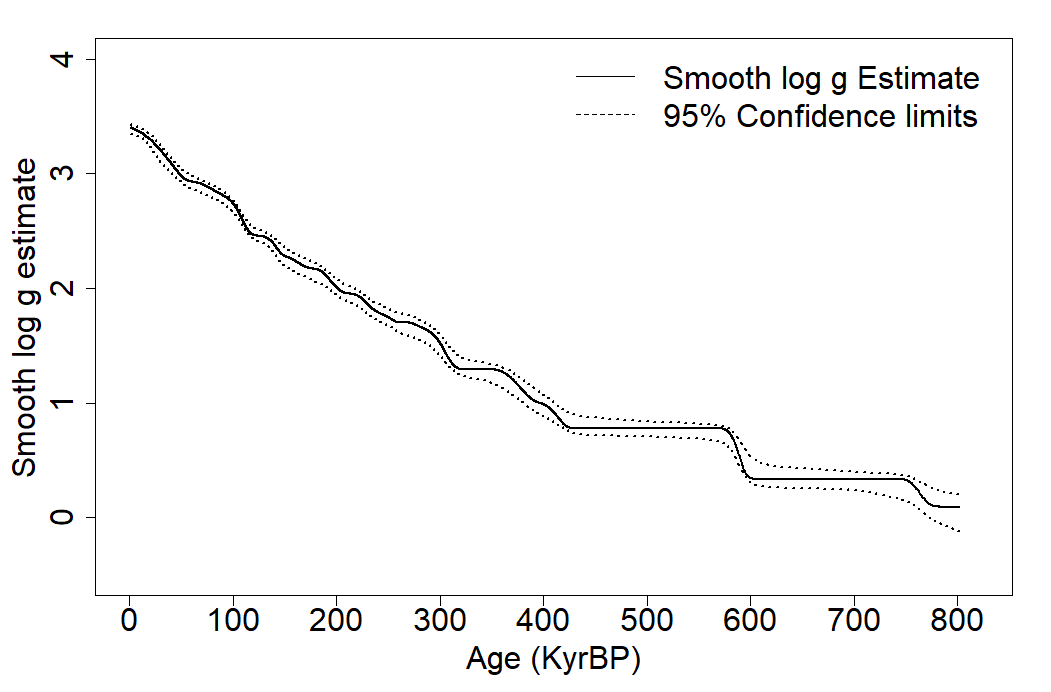}
\\
\includegraphics[width=2.5in,height=1.8in]{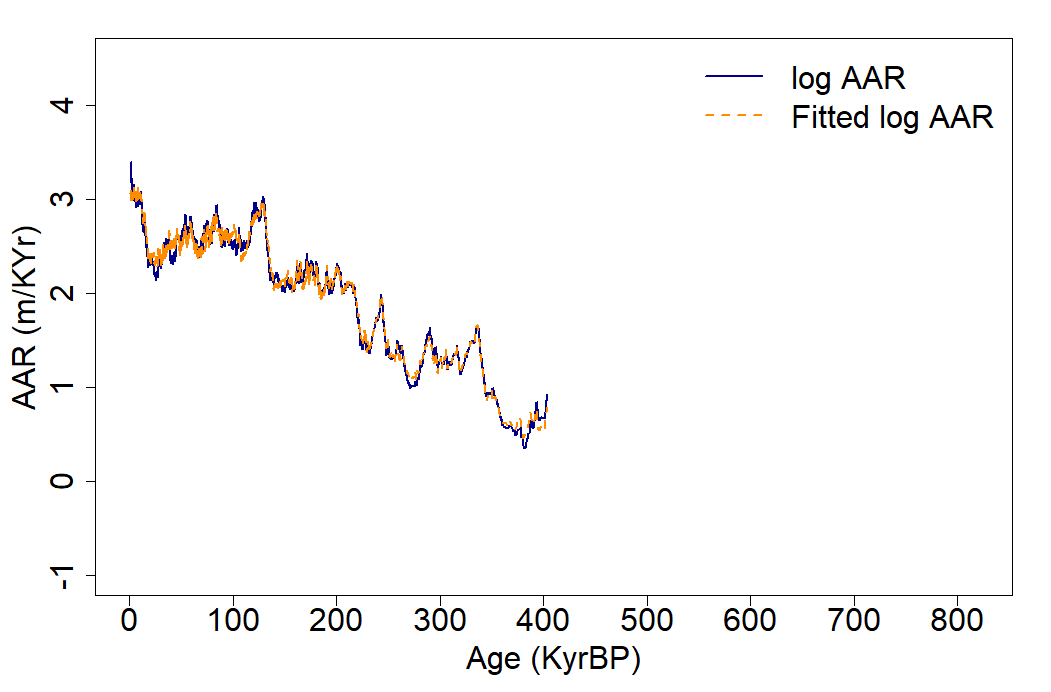}&
\includegraphics[width=2.5in,height=1.8in]{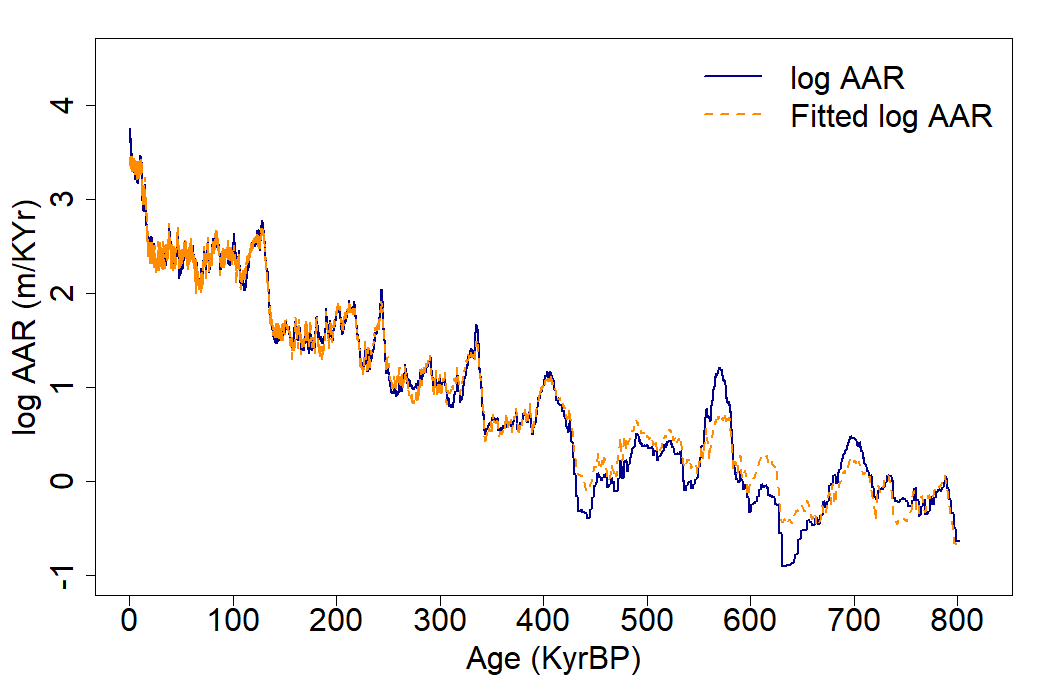}
\end{tabular}
\caption{Plots of $\log\tilde{g}_n$ with pointwise bootstrap confidence intervals against age (top panel) and the plots of $\log$AAR together with fitted $\log$AAR (bottom panel) at the two locations.}
\label{Fig4}
\end{figure}

\begin{figure}
\centering
\begin{tabular}{cc}
 {\underline{Lake Vostok}} &{\underline{EPICA Dome C}}\\
\includegraphics[width=2.5in,height=1.8in]{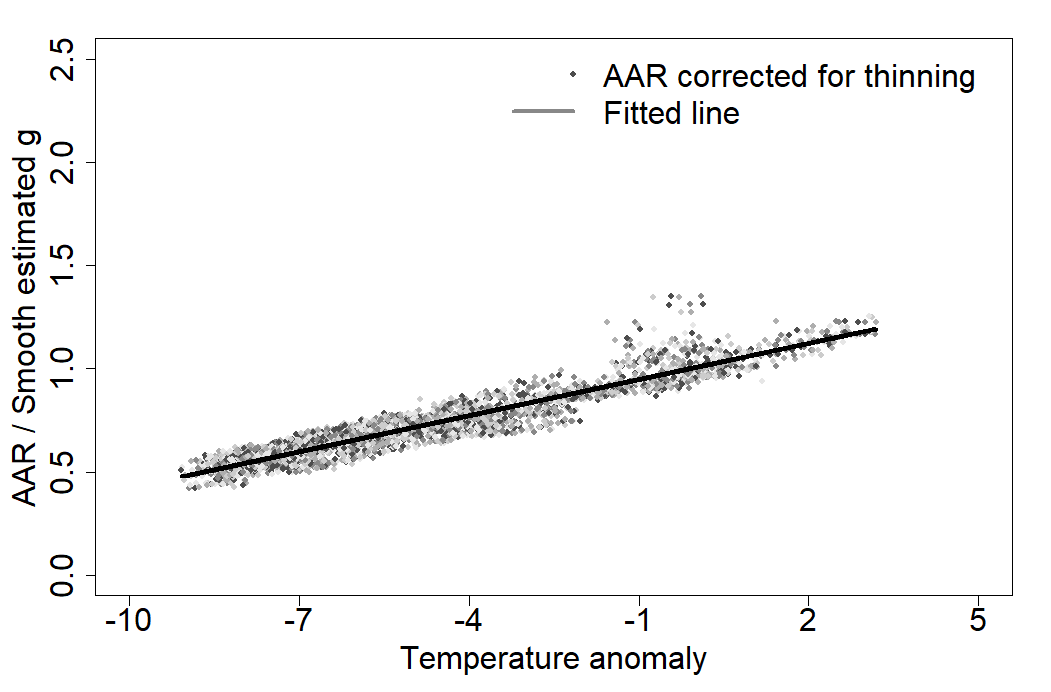}&
\includegraphics[width=2.5in,height=1.8in]{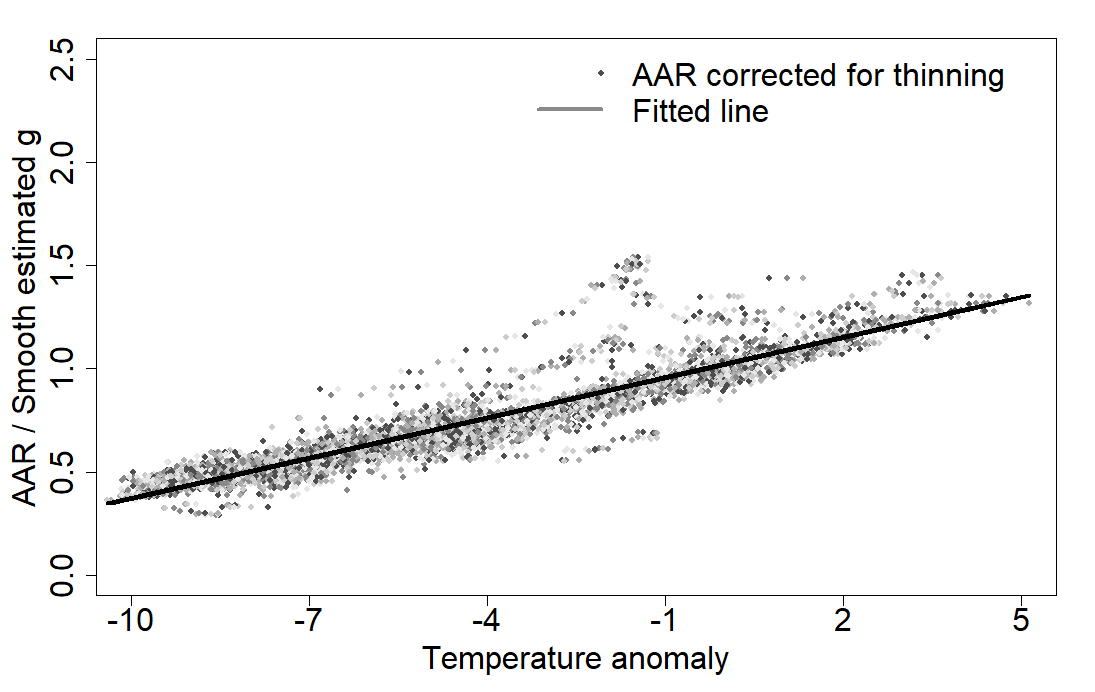}\\
\\
\includegraphics[width=2.5in,height=1.8in]{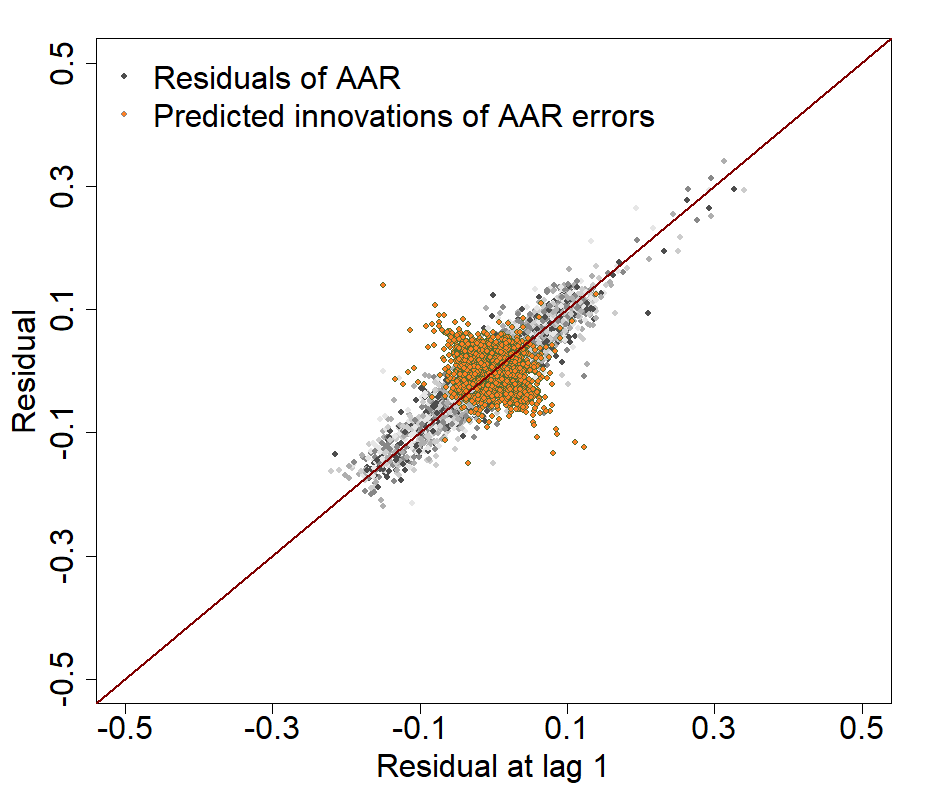}&
\includegraphics[width=2.5in,height=1.8in]{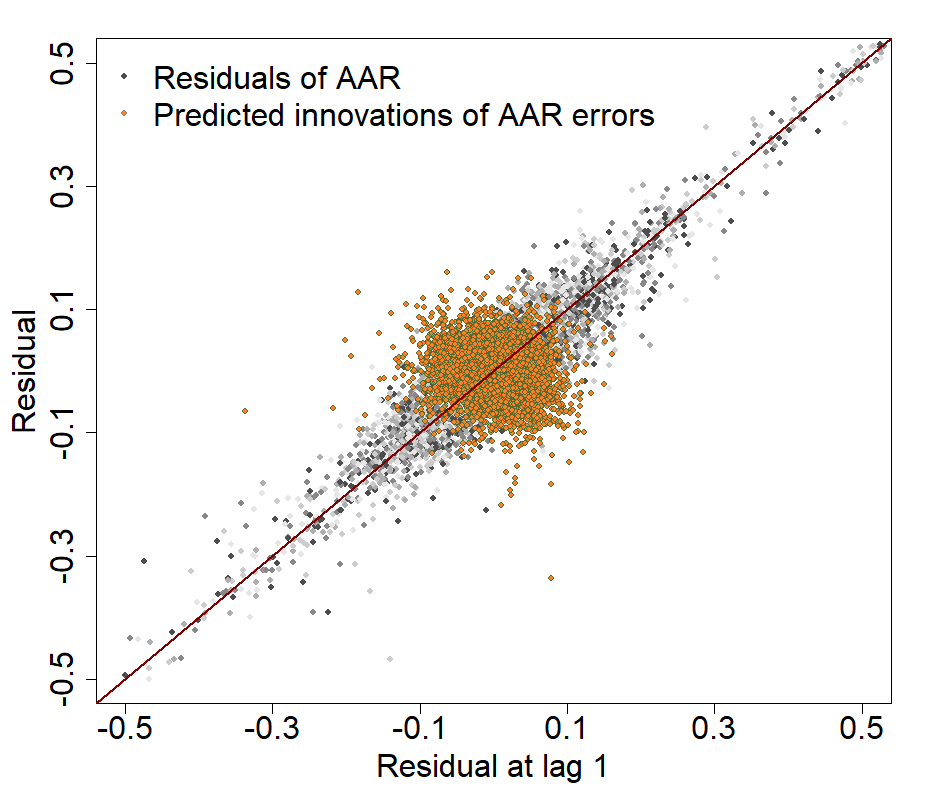}
\\
\includegraphics[width=2.5in,height=1.8in]{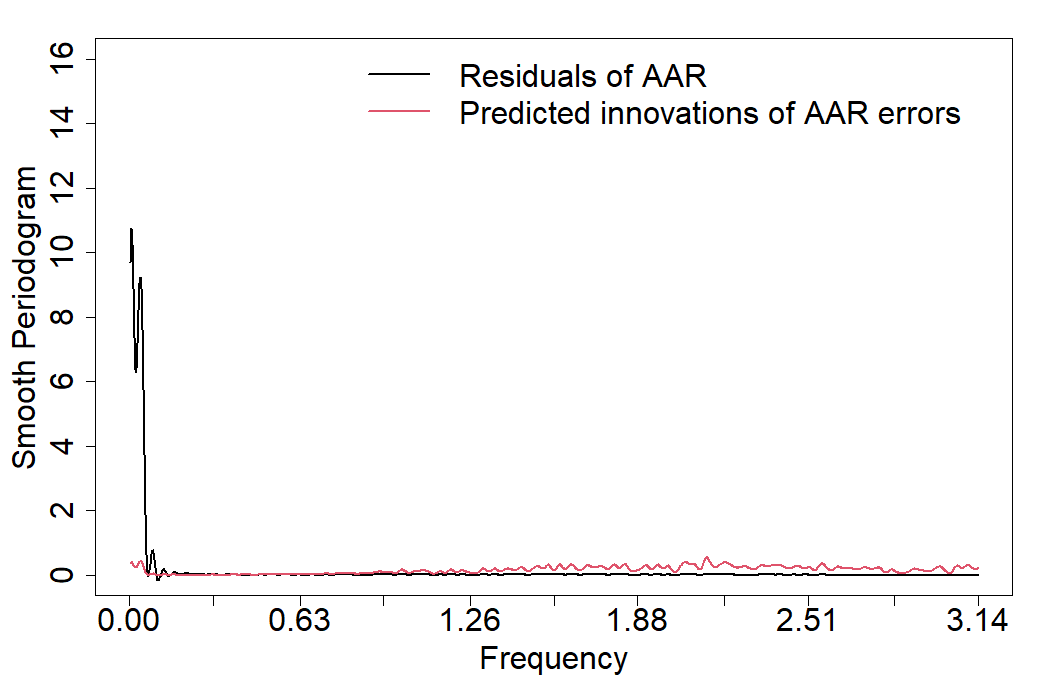}&
\includegraphics[width=2.5in,height=1.8in]{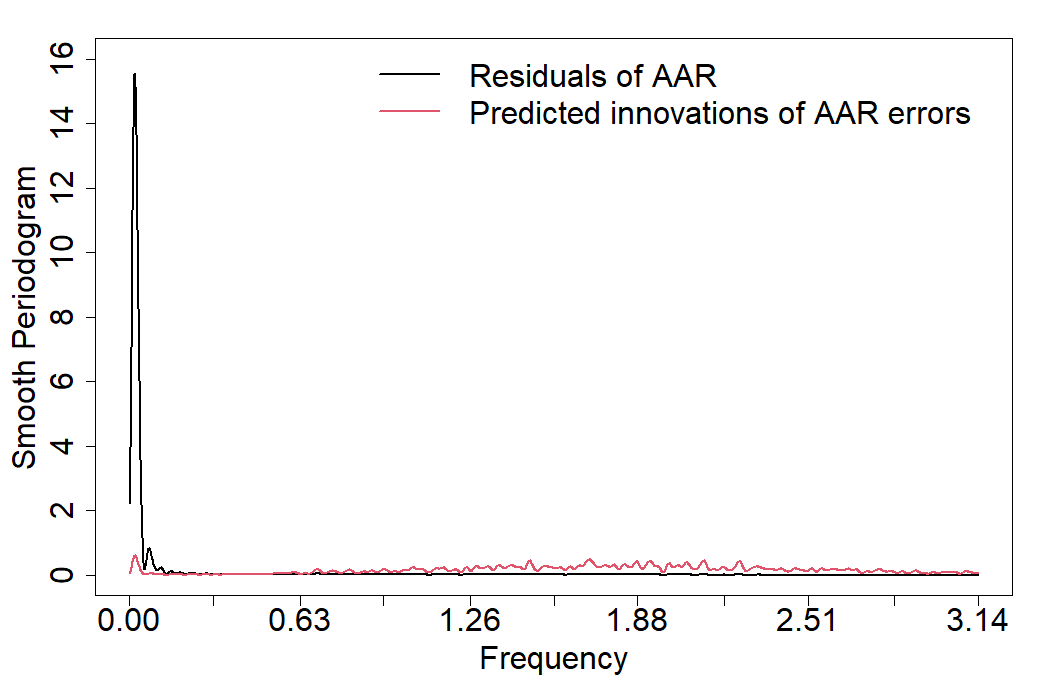}
\end{tabular}
\caption{Scatter plots of AAR/$\tilde{g}$ against temperature anomaly along with the fitted line (top panel), scatter plots of $\tilde{\varepsilon}_{i}$ against $\tilde{\varepsilon}_{i-1}$ in gray points and $\tilde{\zeta}_{i}$ against  $\tilde{\zeta}_{i-1}$ in brown points (middle panel) and the plots of smooth periodogram of $\tilde{\varepsilon}_{i}$ and $\tilde{\zeta}_{i}$ (bottom panel).}
\label{Fig4_diagnostic}
\end{figure}
The top panel of Figure~\ref{Fig4_diagnostic} shows the plots of the fitted AAR (adjusted for thinning) against temperature ($x$), for the two locations. The adjustment for thinning, based on model~\eqref{model}, is made by dividing the observed AAR by the estimated $g$. The line $1+\tilde\gamma_nx$ is overlaid on the scatter as the ``fitted line''. There are occasional departures from the fitted line in the plots for both the locations. The larger departures for EPICA Dome C occur in cycles at the distant past (see bottom panel of Figure~\ref{Fig4}). Apart from these systematic departures, which may be attributed to factors not included in the model, the points in the scatter plots stay very close to the fitted line. This indicates a persistently linear relationship with temperature disturbed at times by unaccounted factors. Please see Section~\ref{s7} for a discussion on the unaccounted factors.  

The middle panel of Figure~\ref{Fig4_diagnostic} shows the lag plots of the predicted innovations $\tilde{\zeta}_i$ of the AAR errors, overlaid on the lag plots of the residuals $\tilde{\varepsilon}_i$ of AAR, at both the locations. The conspicuous stretch of the lag plot of the $\tilde{\varepsilon}_i$s indicate the strong correlation of the errors that had to be incorporated in the model. The lag plot of the $\tilde{\zeta}_i$s is mostly free from this anomaly. The bottom panel of Figure~\ref{Fig4_diagnostic} shows smoothed periodogram of the predicted innovations $\tilde{\zeta}_i$ of the AAR errors, overlaid on the periodogram of the residuals $\tilde{\varepsilon}_i$ of AAR. At both the locations, the sharp peak of the periodogram near zero in the case of residuals is not visible in the case of the predicted innovations.

The findings from the lag plots and smoothed periodograms are complemented by the fact that the estimates of $\beta$ reported in Table~\ref{Tab4} are not too large, in spite of the pattern of over-estimation reported in Section~\ref{s5}. Smaller values of $\beta$ correspond to stronger dependence in the OU process. These findings justify the use of the OU process model for estimating standard errors. 

\section{Discussion}\label{s7}

The strong correlation in the residuals of the model for both Lake Vostok and EPICA Dome C deserves an explanation. The model~\eqref{model} does not take into account the effect of lateral flow of ice. Distortions also arise from spatial redistribution of mass through short-term cycles of mass exchange between moisture and precipitate. The extent of these physical processes vary with time. Some of the slowly varying unaccounted factors induce slowly varying changes in AAR, which would naturally be attributed to the model error. Another issue is the dependence of the AICC2012 scale of time on a number of glaciological factors and stratigraphic markers of absolute ice age \citep{bazin2013}, which in turn are physical processes. If age gets locally distorted in a systematic way, and remain unaccounted in the systematic part of the model, that would also induce correlation in the error part. A third probable factor relates to how AAR depends on temperature. There may be changes in the pattern of dependence in different interglacial cycles. Since the long-term model only accounts for a single linear effect of time, any short-term effect is likely to be expressed as slowly varying distortions that induce correlation among the model errors. In summary, the dependence of AAR on time, indicated by the strong serial correlation among the residuals, may merely be a proxy for its dependence on unaccounted influencing factors that vary slowly with time. As mentioned in Section~\ref{s1}, a part of the correlation also arises from the manner we compute the average ages and temperatures at successive depth intervals. 

The criterion~\eqref{Square_loss} minimized for estimation of the model parameters does not make use of the dependence of the errors. It is known that in the case of linear regression, serial correlation leads to under-estimation of error variance, but no effect on the bias and not much effect on the variance of the least squares estimators \citep{Kramer83}. It is for this reason why we chose the least squares criterion but drew bootstrap resamples from the innovations of residuals modeled as samples from a fitted OU process, rather than resampling from the residuals themselves. One may also explore whether any improvement of efficiency is possible through a provably consistent method.

There may be a question about the reliability of ice core data records as a source of information on the relationship between temperature and accumulation. It is known that ice core data records are affected by post-deposition processes in low-accumulation areas \citep{Casado}. There are issues of coherent climate variations being confounded with local proxy variability \citep{Munch}. Another complication arises from the difference between gas ages needed for temperature reconstructions and ice ages needed for accumulation rate reconstructions \citep{bradley}. Nevertheless, the strong relationship between temperature and AAR adjusted for thinning observed over a vast scale of time (see Figure~\ref{Fig4_diagnostic}, top panel) is hard to dismiss. Further confidence on the that finding may be derived from the fact that the general decreasing function $g_0$ used in the model~\eqref{model} permits the analysis to be resilient to any minor mis-specification in age.

The estimated rise of AAR by 5.7\%\ at Lake Vostok and by 6.2\%\ at EPICA Dome C per degree Celsius rise in temperature, as reported in Section~\ref{s6} is an average phenomenon over several hundred millennia. The model may be applied to data over a shorter time scale for a more nuanced understanding of the issue. Such an empirical study may help us understand whether variability in the shorter and longer time scales are influenced by different dynamic and thermodynamic processes. There may be a methodological challenge owing to shortage of data, which may be overcome by an extension of the proposed model that permits the coefficient $\gamma$ (effect of temperature on precipitation) to change slowly with time.

The model and the method presented here may be extended to sampling at non-uniform depths, enabling analysis of ice core data from a wider range of sites for a better understanding of spatial variations. This exercise will also be useful for understanding the mass balance of the entire Antarctic ice sheet.

It is well known that reconstruction of paleoclimatic temperature from isotopes is a challenging task \citep{Markle, Casado}. Since isotopic composition is influenced by factors other than local temperatures also, there are bound to be uncertainties in the temperature reconstructions. This uncertainty may be incorporated in the proposed model by bringing in errors in the temperature variable $x_i$. This can be another direction of future work.

While the analysis of Section~\ref{s6} shows a positive effect of temperature on AAR over several hundred millennia, it is not clear how fluctuations in temperature have affected \textit{precipitation} during the same period. The rate of accumulation differs from the rate of precipitation due to several factors, such as diurnal sublimation–condensation cycles, melting, lateral flow and so on. The relation between temperature and precipitation is a matter of direct relevance to the issue of ice mass balance. An assessment of precipitation variability from the accumulation rates will be needed to fill the gap.




\section*{Supplementary material}

\renewcommand{\theequation}{S.\arabic{equation}}

In this supplementary material, we provide the proofs of all theorems of the main manuscript in section~\ref{s.s1} and a few additional finite sample studies in section~\ref{s.s2}.

\spacingset{1.9}
\section{Proof of theorems of the main manuscript}\label{s.s1}
\begin{proof}[Proof of Theorem~1] By using equation (4) of the main manuscript, we can write
\begin{equation}\label{Qn_decompose}
Q_n(\gamma,g)
=\frac{1}{n}\sum_{i=1}^{n}[v_i+v_i']^2 +\frac{1}{n}\sum_{i=1}^{n}\varepsilon_i^2 \quad\mbox{} +\frac{2}{n}\sum_{i=1}^{n}v_i\varepsilon_i+\frac{2}{n}\sum_{i=1}^{n}v_i'\varepsilon_i,
\end{equation}
where $v_i=\left(\log (1 +\gamma_0 x_{i})-\log (1 +\gamma x_{i})\right)$ and $v_i'=\left(\log g_0(z_i)-\log g(z_i)\right)$, for $i=1,2,\ldots,n$. 

From Assumptions~1 and~2, the first term on the right hand side (RHS) of \eqref{Qn_decompose} can be viewed as the Riemann sum,
\begin{eqnarray*}
\frac{1}{\Delta}\sum_{i=1}^{n}[\log (1+\gamma_0 x(d_i))-\log (1 +\gamma x(d_i))+\log g_0(z(d_i))-\log g(z(d_i))]^2\frac{\Delta}{n},    
\end{eqnarray*}
which converges to $Q(\gamma,g)$. We now turn to the second term on RHS of \eqref{Qn_decompose}. Note that
\begin{eqnarray}
    \frac1n\sum_{i=1}^n \varepsilon_i^2&\!\!=\!\!&\frac{1}{n\rho_n}\sum_{i=1}^n \varepsilon^2(i\rho_n) \rho_n\nonumber\\
    &\!\!=\!\!&\frac{1}{n\rho_n}\int_{0}^{n\rho_n}\varepsilon^2(u) du \!+\! 
    \frac{1}{n\rho_n}\int_{0}^{n\rho_n}\left(\sum_{i=1}^n \varepsilon^2(i\rho_n) I_{((i-1)\rho_n,i\rho_n]}(u)-\varepsilon^2(u)\right) du,\qquad\mbox{}\label{2term_1}     
\end{eqnarray}
where  
\begin{equation*}
    I_{((i-1)\rho_n, i\rho_n]}(u)=\begin{cases}
    1 &\mbox{if } u\in ((i-1)\rho_n, i\rho_n]\\
    0& \mbox{otherwise}
\end{cases}
\end{equation*}
for $i=1,2,\ldots,n$.\\
Now for any $u\in[0, \ n\rho_n]$, note that   
$\sum_{i=1}^n \varepsilon^2(i\rho_n) I_{((i-1)\rho_n,i\rho_n]}(u)=\varepsilon^2\left (\lceil u/\rho_n\rceil \rho_n \right)$ where $\lceil\cdot\rceil$ is the ceiling function.
Therefore, by using \eqref{2term_1}, we have 
\begin{eqnarray}
 \left|\frac1n\sum_{i=1}^n \varepsilon_i^2-\frac{\sigma^2}{2\beta}\right|\le \left|\frac{1}{n\rho_n}\int_{0}^{n\rho_n}\varepsilon^2(u) du-\frac{\sigma^2}{2\beta}\right|\!+\!  \left|\frac{1}{n\rho_n}\int_{0}^{n\rho_n}\left(\varepsilon^2\left (\lceil u/\rho_n\rceil \rho_n \right)-\varepsilon^2(u)\right) du\right|. \label{2term_2}  
\end{eqnarray}
The first term on the RHS of \eqref{2term_2} converges to $0$ almost surely by using Lemma~\ref{lem:OU} given below.
We now consider the second term on the RHS of \eqref{2term_2}. 
Note that, by using equation (11) of the main manuscript and $u\le\lceil u/\rho_n\rceil \rho_n$,  we have
\begin{eqnarray}
    \varepsilon(u)&=& e^{-\beta u}  \ \varepsilon(0)+ \sigma e^{-\beta u }\int_{0}^u e^{\beta s} dW(s)\label{OU_u}\\
    \varepsilon(\lceil u/\rho_n\rceil \rho_n)&=&e^{-\beta \lceil u/\rho_n\rceil \rho_n} \  \varepsilon(0)+ \sigma e^{-\beta \lceil u/\rho_n\rceil \rho_n } \int_{0}^u e^{\beta s} dW(s) \nonumber\\&&+ \sigma e^{-\beta \lceil u/\rho_n\rceil \rho_n } \int_{u}^{\lceil u/\rho_n\rceil \rho_n} e^{\beta s} dW(s). \label{OU_urhon}
\end{eqnarray}
 Now, by using \eqref{OU_u} and \eqref{OU_urhon}, the integrand of the second term on the RHS of \eqref{2term_2} is decomposed as follows.
 \begin{eqnarray}
    &&\hskip-20pt\varepsilon^2(\lceil u/\rho_n\rceil \rho_n) -\varepsilon^2(u)\nonumber\\&=& \ \ \ (e^{-2\beta \lceil u/\rho_n\rceil \rho_n }-e^{-2\beta u}) \ \varepsilon^2(0)\nonumber\\&&+
    \left((e^{-2\beta \lceil u/\rho_n\rceil \rho_n }-e^{-2\beta u}) (e^{2\beta u}-1)\right)\left(\frac{\sigma}{\sqrt{e^{2\beta u}-1}}\int_{0}^u e^{\beta s} dW(s)\right)^2\nonumber\\
    &&+  \left(e^{-2\beta \lceil u/\rho_n\rceil \rho_n }(e^{2\beta \lceil u/\rho_n\rceil \rho_n}-e^{2\beta u}) \right)\left(\frac{\sigma}{\sqrt{e^{2\beta \lceil u/\rho_n\rceil \rho_n}-e^{2\beta u}}}\int_{u}^{\lceil u/\rho_n\rceil \rho_n} e^{\beta s} dW(s)\right)^2\nonumber\\
    &&+2\left((e^{-2\beta \lceil u/\rho_n\rceil \rho_n }-e^{-2\beta u})\sqrt{e^{2\beta u}-1}\right)\left(\varepsilon(0)\frac{\sigma}{\sqrt{e^{2\beta u}-1}}\int_{0}^u e^{\beta s} dW(s)\right)\nonumber\\
    &&+ 2 \left(e^{-2\beta \lceil u/\rho_n\rceil \rho_n } \sqrt{e^{2\beta \lceil u/\rho_n\rceil \rho_n}-e^{2\beta u}}\right)\left( \varepsilon(0) \frac{\sigma}{\sqrt{e^{2\beta \lceil u/\rho_n\rceil \rho_n}-e^{2\beta u}}} \int_{u}^{\lceil u/\rho_n\rceil \rho_n} e^{\beta s} dW(s)\right)\nonumber\\
    &&+2 \left(e^{-2\beta \lceil u/\rho_n\rceil \rho_n } \sqrt{e^{2\beta u}-1} \sqrt{e^{2\beta \lceil u/\rho_n\rceil \rho_n}-e^{2\beta u}}\right)\nonumber\\&&\hskip20pt\left(\frac{\sigma}{\sqrt{e^{2\beta u}-1}}\int_{0}^u e^{\beta s} dW(s) \frac{\sigma}{\sqrt{e^{2\beta \lceil u/\rho_n\rceil \rho_n}-e^{2\beta u}}}\int_{u}^{\lceil u/\rho_n\rceil \rho_n} e^{\beta s} dW(s)\right)\nonumber\\
    &=&T_{21}+T_{22}+T_{23}+T_{24}+T_{25}+T_{26}, \ (\mbox{say}).\label{Integrand_decompose}
 \end{eqnarray}
We now show by using Lemma~\ref{lem:gen} that each term on the RHS of \eqref{Integrand_decompose} converges uniformly (over all u) to $0$ almost surely. We begin with the term $T_{21}$ on the RHS of \eqref{Integrand_decompose}. Note that
\begin{equation}
    \left|e^{-2\beta \lceil u/\rho_n\rceil \rho_n }-e^{-2\beta u}\right|\le 2\beta |\lceil u/\rho_n\rceil \rho_n-u|\le 2\beta \rho_n\label{T21_fn}.
\end{equation} 
Further, from Assumption~3,  $\varepsilon^2(0)\stackrel{\mathcal{L}}{\equiv}\frac{\sigma^2}{2\beta} \chi^2_1$ (where the notation $\stackrel{\mathcal{L}}{\equiv}$ indicates distributional equivalence) and it 
satisfies \eqref{le2:eq1}. Thus, by 
Lemma~\ref{lem:gen}, $T_{21}$ converges uniformly to $0$ almost surely. 

We consider the term $T_{22}$ on the RHS of \eqref{Integrand_decompose}.
Note that 
\begin{eqnarray}
   \left|(e^{-2\beta \lceil u/\rho_n\rceil \rho_n }-e^{-2\beta u}) (e^{2\beta u}-1)\right|\le |e^{-2\beta (\lceil u/\rho_n\rceil \rho_n-u) }-1|\le   2\beta|\lceil u/\rho_n\rceil \rho_n-u|\le 2\beta\rho_n. \label{T22_fn}
\end{eqnarray}
Further, by using Wiener integration, we have
\begin{equation}
    \left(\frac{\sigma}{\sqrt{e^{2\beta u}-1}}\int_{0}^u e^{\beta s} dW(s)\right)^2\stackrel{\mathcal{L}}{\equiv}\frac{\sigma^2}{2\beta} \chi^2_1\label{T22_vn}, 
\end{equation}
and the LHS of \eqref{T22_vn} satisfies \eqref{le2:eq1}. Thus, it follows from Assumption~3, \eqref{T22_fn}, \eqref{T22_vn} and Lemma~\ref{lem:gen} that $T_{22}$ converges uniformly to $0$ almost surely.

We now turn to the term $T_{23}$ on the RHS of \eqref{Integrand_decompose}. 
Note that 
\begin{eqnarray}
   \left|e^{-2\beta \lceil u/\rho_n\rceil \rho_n }(e^{2\beta \lceil u/\rho_n\rceil \rho_n}-e^{2\beta u}) \right|\le |e^{-2\beta (\lceil u/\rho_n\rceil \rho_n-u) }-1|\le   2\beta|\lceil u/\rho_n\rceil \rho_n-u|\le 2\beta\rho_n. \label{T23_fn}
\end{eqnarray}
Further, by using Wiener integration, we have 
\begin{equation}
\left(\frac{\sigma}{\sqrt{e^{2\beta \lceil u/\rho_n\rceil \rho_n}-e^{2\beta u}}}\int_{u}^{\lceil u/\rho_n\rceil \rho_n} e^{\beta s} dW(s)\right)^2\stackrel{\mathcal{L}}{\equiv}\frac{\sigma^2}{2\beta} \chi^2_1, \label{T23_vn}   
\end{equation}
 and the LHS of \eqref{T23_vn} satisfies \eqref{le2:eq1}. Thus, it follows from Assumption~3, \eqref{T23_fn}, \eqref{T23_vn}, and Lemma~\ref{lem:gen} that $T_{23}$ converges uniformly to $0$ almost surely.

Similar arguments establish 
the uniform convergence of the terms $T_{24}$, $T_{25}$ and $T_{26}$ on the RHS of \eqref{Integrand_decompose} to $0$ uniformly almost surely. This shows that the second term on the RHS of 
\eqref{2term_2} converges to $0$ almost surely. Therefore, we have
\begin{equation}
    \lim_{n\to\infty}\frac1n \sum_{i=1}^n \varepsilon_i^2=\frac{\sigma^2}{2\beta} \quad \mbox{ almost surely.}\label{quad_err_conv}
\end{equation}
We now turn to the third term on the RHS of \eqref{Qn_decompose}. Note that
\begin{equation}
    \frac{1}{n}\sum_{i=1}^n v_i\varepsilon_i=  \frac{1}{n\rho_n}\sum_{i=1}^n  \mathcal{V}\left(i\rho_n \frac{\Delta}{n\rho_n}\right)\varepsilon(i\rho_n) \rho_n=\frac{1}{n\rho_n}\int_0^{n\rho_n}\mathcal{V}\left(\left\lceil u/\rho_n\right\rceil\rho_n  \frac{\Delta}{n\rho_n}\right) \varepsilon(\lceil u/\rho_n\rceil \rho_n) du,
\label{Qn_third_term}
\end{equation}
where
$\mathcal{V}(s)=\log \left(1 +\gamma_0 x(s\right))-\log \left(1 +\gamma x(s)\right)$ for $s\in[0,\Delta]$. Now, by using \eqref{Qn_third_term}, we have
\begin{eqnarray}
    \left|\frac{1}{n}\sum_{i=1}^n v_i\varepsilon_i\right|&\le& \left|\frac{1}{n\rho_n}\int_0^{n\rho_n}\left(\mathcal{V}\left(\left\lceil u/\rho_n\right\rceil \rho_n \frac{\Delta}{n\rho_n}\right) \varepsilon(\lceil u/\rho_n\rceil \rho_n) -\mathcal{V}\left(u\frac{\Delta}{n\rho_n}\right)\varepsilon(u)\right) du\right|\nonumber\\&&+ \left|\frac{1}{n\rho_n}\int_0^{n\rho_n}\mathcal{V}\left(u\frac{\Delta}{n\rho_n}\right) \varepsilon(u) du\right|.\label{Qn_third_decompose}
\end{eqnarray}
We intend to show that both the terms on the RHS of \eqref{Qn_third_decompose} converge to $0$ almost surely. By using \eqref{OU_u} and \eqref{OU_urhon}, we decompose the integrand of the first term on the RHS of \eqref{Qn_third_decompose} as follows. 
\begin{eqnarray}
 && \hskip-0.3in\mathcal{V}\left(\left\lceil u/\rho_n\right\rceil\rho_n  \frac{\Delta}{n\rho_n}\right) \varepsilon(\lceil u/\rho_n\rceil \rho_n) -\mathcal{V}\left(u\frac{\Delta}{n\rho_n}\right)\varepsilon(u)\nonumber\\&=&
 \left(\mathcal{V}\left(\left\lceil u/\rho_n\right\rceil\rho_n  \frac{\Delta}{n\rho_n}\right) e^{-\beta \lceil u/\rho_n\rceil \rho_n} -\mathcal{V}\left(u\frac{\Delta}{n\rho_n}\right)e^{-\beta u}\right) \  \varepsilon(0)\nonumber\\
 &&\hskip-0.15in+ \left(\left(\mathcal{V}\left(\left\lceil u/\rho_n\right\rceil\rho_n  \frac{\Delta}{n\rho_n}\right) e^{-\beta \lceil u/\rho_n\rceil \rho_n} -\mathcal{V}\left(u\frac{\Delta}{n\rho_n}\right)e^{-\beta u}\right) \sqrt{e^{2\beta}-1} \right)\left(\frac{\sigma}{\sqrt{e^{2\beta u}-1}}\int_0^u e^{\beta s} dW(s) \right)\nonumber\\&&\hskip-0.15in+   \mathcal{V}\left(\left\lceil u/\rho_n\right\rceil\rho_n  \frac{\Delta}{n\rho_n}\right) e^{-\beta \lceil u/\rho_n\rceil \rho_n} \sqrt{e^{2\beta \lceil u/\rho_n\rceil \rho_n}-e^{2\beta u}} \left(\frac{\sigma}{\sqrt{e^{2\beta \lceil u/\rho_n\rceil \rho_n}-e^{2\beta u}}}\int_{u}^{\lceil u/\rho_n\rceil \rho_n} e^{\beta s} dW(s)\right) \nonumber\\
 &=& T_{31} + T_{32} + T_{33}, (\mbox{ say}).\label{eq:T3_first_decompose}
\end{eqnarray}
We begin with the term $T_{31}$ on the RHS of \eqref{eq:T3_first_decompose}. Note that 
\begin{eqnarray}
 &&\mbox{}\hskip-30pt\left|\mathcal{V}\left(\left\lceil u/\rho_n\right\rceil\rho_n  \frac{\Delta}{n\rho_n}\right) e^{-\beta \lceil u/\rho_n\rceil \rho_n} -\mathcal{V}\left(u\frac{\Delta}{n\rho_n}\right)e^{-\beta u}\right|\nonumber\\&\le& \left|\mathcal{V}\left(\left\lceil u/\rho_n\right\rceil\rho_n  \frac{\Delta}{n\rho_n}\right) e^{-\beta( \lceil u/\rho_n\rceil \rho_n-u)} -\mathcal{V}\left(u\frac{\Delta}{n\rho_n}\right)\right|\nonumber\\&\le& 
    \left|\mathcal{V}\left(\left\lceil u/\rho_n\right\rceil\rho_n  \frac{\Delta}{n\rho_n}\right) -\mathcal{V}\left(u\frac{\Delta}{n\rho_n}\right)\right| +
    \left|\mathcal{V}\left(\left\lceil u/\rho_n\right\rceil\rho_n  \frac{\Delta}{n\rho_n}\right)\right| \beta| \lceil u/\rho_n\rceil \rho_n-u|.\qquad\ \label{T31_fn}
\end{eqnarray}
It follows from the definition of $\Gamma$ and Assumption~3 that $\mathcal{V}$ is a Lipschitz continuous function. Now by using \eqref{T31_fn}, we have
\begin{eqnarray}
    \left|\mathcal{V}\left(\left\lceil u/\rho_n\right\rceil\rho_n  \frac{\Delta}{n\rho_n}\right) e^{-\beta \lceil u/\rho_n\rceil \rho_n} -\mathcal{V}\left(u\frac{\Delta}{n\rho_n}\right)e^{-\beta u}\right|\le k_0 \frac{\Delta}{n}+ \sup_{s\in[0,\Delta]}\mathcal{V}(s) \beta\rho_n,\label{T31_fn1}
\end{eqnarray}
where the constant $k_0>0$ does not depend on $u$. Now, by using \eqref{T31_fn1}, Assumption~3 and Lemma~\ref{lem:gen}, the term $T_{31}$ converges to $0$ uniformly almost surely.

We now consider the term $T_{32}$ on the RHS of \eqref{eq:T3_first_decompose}.
Note that
\begin{eqnarray}
    &&\hskip-30pt\left|\left(\mathcal{V}\left(\left\lceil u/\rho_n\right\rceil\rho_n  \frac{\Delta}{n\rho_n}\right) e^{-\beta \lceil u/\rho_n\rceil \rho_n} -\mathcal{V}\left(u\frac{\Delta}{n\rho_n}\right)e^{-\beta u}\right) \sqrt{e^{2\beta u}-1} \right|\nonumber\\&\le&\left|\mathcal{V}\left(\left\lceil u/\rho_n\right\rceil\rho_n  \frac{\Delta}{n\rho_n}\right) e^{-\beta( \lceil u/\rho_n\rceil \rho_n-u)} -\mathcal{V}\left(u\frac{\Delta}{n\rho_n}\right)\right|. \label{T32_fn1}
\end{eqnarray}
Therefore, by using \eqref{T31_fn}, \eqref{T31_fn1} and \eqref{T32_fn1}, we have
\begin{equation}
  \left|\left(\mathcal{V}\left(\left\lceil u/\rho_n\right\rceil\rho_n  \frac{\Delta}{n\rho_n}\right) e^{-\beta \lceil u/\rho_n\rceil \rho_n} -\mathcal{V}\left(u\frac{\Delta}{n\rho_n}\right)e^{-\beta u}\right) \sqrt{e^{2\beta u}-1} \right|\to 0 \mbox{ uniformly }.\label{T32_fn2}  
\end{equation}
Since the Wiener integral $\frac{\sigma}{\sqrt{e^{2\beta}-1}}\int_0^u e^{\beta s} dW(s)$ for any $u$ is a normally distributed random variable with mean $0$ and variance $\frac{\sigma^2}{2\beta}$, it satisfies condition  \eqref{le2:eq1} of Lemma~\ref{lem:gen}. It follows from that lemma that $T_{32}$ converges uniformly to $0$ almost surely.

It follows from a similar line of argument involving \eqref{T23_fn} and Lemma~\ref{lem:gen} that the term $T_{33}$ on the RHS of \eqref{eq:T3_first_decompose} converges uniformly to $0$ almost surely. This shows that the first term on the RHS of \eqref{Qn_third_decompose} converges to $0$ almost surely. 

We now consider the second term on the RHS of \eqref{Qn_third_decompose}. By using equation (11) of the main manuscript, we have
\begin{eqnarray}
 \mathsf{T}_n&=&   \frac{1}{n\rho_n} \int_{0}^{n\rho_n} \mathcal{V}\left( u\frac{\Delta}{n\rho_n}\right)\varepsilon(u)du\nonumber\\ &=&\frac{1}{n\rho_n} \int_{0}^{n\rho_n} \mathcal{V}\left( u\frac{\Delta}{n\rho_n}\right)e^{-\beta u}du \ \varepsilon(0) +
    \frac{1}{n\rho_n} \int_{0}^{n\rho_n} \mathcal{V}\left( u\frac{\Delta}{n\rho_n}\right)\sigma \int_{0}^{u} e^{-\beta(u-s)}dW(s)du\nonumber\\
    &=& \frac{1}{n\rho_n} \int_{0}^{n\rho_n} \mathcal{V}\left( u\frac{\Delta}{n\rho_n}\right)e^{-\beta u}du \ \varepsilon(0) +
    \frac{\sigma}{n\rho_n} \int_{0}^{n\rho_n} \left(\int_{s}^{n\rho_n} \mathcal{V}\left( u\frac{\Delta}{n\rho_n}\right) e^{-\beta(u-s)}du \right)dW(s).\hskip-20pt 
    \mbox{ } \nonumber\\
    &&\label{qn_third_second_part}
\end{eqnarray}
By using \eqref{qn_third_second_part} and Assumption~3, we observe that $\mathsf{T}_n$ is a $0$ mean Gaussian random variable with
\begin{eqnarray}
    Var(\mathsf{T}_n)&=&\frac{\sigma^2}{2\beta}\frac{1}{(n\rho_n)^2}\left(\int_{0}^{n\rho_n} \mathcal{V}\left( u\frac{\Delta}{n\rho_n}\right)e^{-\beta u}du\right)^2\nonumber\\
    &&+ \frac{\sigma^2}{(n\rho_n)^2} \int_{0}^{n\rho_n} \left(\int_{s}^{n\rho_n} \mathcal{V}\left( u\frac{\Delta}{n\rho_n}\right) e^{-\beta(u-s)}du \right)^2ds\nonumber\\
    &\le& \frac{\sigma^2(\sup_{u\in[0,\Delta]}\mathcal{V}(u))^2}{2\beta^3 (n\rho_n)^2}+ \frac{\sigma^2(\sup_{u\in[0,\Delta]}\mathcal{V}(u))^2}{\beta^2(n\rho_n)}\le \frac{\tau}{n\rho_n},\label{var_Tn_bound} 
\end{eqnarray}
where $\tau$ is a positive constant. By using the tail bounds of the Gaussian random variable and \eqref{var_Tn_bound}, for any $\delta>0$, we have
\begin{eqnarray}
    P(|\mathsf{T}_n|\ge \delta)&\le & 2e^{-\frac{\delta^2}{Var(\mathsf{T}_n)}}\le 2e^{-\frac{\delta^2} {\tau} \  n\rho_n}.\label{tailbound}
\end{eqnarray}
Since $\sum_{n=1}^\infty e^{-\frac{\delta^2}{\tau} \ n\rho_n}<\infty$ from Assumption~3, the above inequality and the Borel-Cantelli lemma imply that $\mathsf{T}_n$ converges to $0$ almost surely. Thus, we have 
\begin{eqnarray}
    \lim_{n\to \infty }\frac 1n\sum_{i=1}^n v_i \varepsilon_i=0 \quad \mbox{almost surely}.\label{weight_avg_convergence}
\end{eqnarray}

We can show by using Assumption~1, Assumption~3 and an argument similar to that used in establishing the convergence of \eqref{weight_avg_convergence}, that the fourth term on the RHS of \eqref{Qn_decompose} converges to $0$ almost surely. This completes the proof. 
\end{proof}

\begin{lemma}\label{lem:OU}
Let $f$ be a real valued function such that $\int_{-\infty}^\infty  f(s) \frac{1}{\sqrt{2\pi \sigma^2/2\beta}}e^{-\frac12\frac{s^2}{\sigma^2/2\beta}} \ ds<\infty$. Then, under Assumption~3,
\begin{eqnarray*}
\lim_{n\to \infty}    \frac1{n\rho_n} \int_{0}^{n\rho_n} f(\varepsilon(s))ds&=& \int_{-\infty}^\infty  f(s) \frac{1}{\sqrt{2\pi \sigma^2/2\beta}}e^{-\frac12\frac{s^2}{\sigma^2/2\beta}} \ ds \  \mbox { almost surely}. 
\end{eqnarray*}
\end{lemma}

\begin{proof}[Proof of Lemma~\ref{lem:OU}] 
The proof follows from Theorem~5.1 of \citep{Hasminskii1980} that establishes the strong law of large number of homogeneous Markov process with invariant distribution arising from stochastic differential equations. Note that, from Assumption~3, the OU process $\{\varepsilon(s), s\in[0,\infty) \}$ is given by the stochastic differential equation 
\begin{eqnarray}
    d\varepsilon(s)=-\beta \varepsilon(s) ds +\sigma dW(s),
\end{eqnarray}
and is a homogeneous recurrent Markov process with invariant distribution $N(0,\frac{\sigma^2}{2\beta})$.  The OU process satisfies the conditions of Theorem~5.1 of \citep{Hasminskii1980}. This completes the proof.   
\end{proof}

\begin{lemma}\label{lem:gen}
Let $f_n$ be a sequence of real valued functions defined over $[0,\infty)$  such that $f_n\to 0$ uniformly as $n\to \infty $. Let $\{\vartheta_n(u), \ u\in[0,\infty)\}$ be a sequence of stochastic processes so that, for any $\eta>0$, there exists $M>0$ such that 
\begin{equation}
    P\left(|\vartheta_n(u)|<M \ \ \forall \ u\in[0,\infty) \right)\ge 1-\eta.\label{le2:eq1}
\end{equation}
Then, for any $\delta, \ \eta>0$, there exists an integer $N(\delta)$ such that
\begin{equation}
    P\left(|f_n(u) \vartheta_n(u)|<\delta \ \forall \ n\ge N(\delta) \mbox{ and } u\in [0,\infty)\right)\ge 1-\eta.\label{le2:eq2}
\end{equation}
\end{lemma}

\begin{proof}[Proof of Lemma~\ref{lem:gen}]
Since $f_n(u) \to 0$ uniformly as $n\to\infty$, for any given $\delta>0$ there exists an integer $N(\delta)$ such that
\begin{equation}
    |f_n(u)|\le \frac{\delta}{M} \ \ \forall n\ge N(\delta) \mbox{ and } u\in [0,\infty).\label{le2:eq3}
\end{equation}
Since 
\begin{eqnarray*}
    \left\{|f_n(u) \vartheta_n(u)|\le\delta \right\}\supseteq \left\{|f_n(u) |<\frac{\delta}{M},  |\vartheta_n(u)|\le M\right\},
\end{eqnarray*}
the assertion of the lemma follows from \eqref{le2:eq1} and \eqref{le2:eq3}. This completes the proof.
\end{proof}

\begin{proof}[Proof of Theorem~2]
The functional $Q(\gamma,g)$ defined on the parameter space $\Gamma\times\mathcal{G}$ is continuous with respect to the distance metric defined at the beginning of Section~3 of the main manuscript. 
Let $(\hat{\gamma}_{n_k},\log \hat{g}_{n_k})$ be a subsequence of $(\hat{\gamma}_n,\log\hat{g}_n)$ that converges to the limit point ($\gamma',\log g'$). Then by using Theorem~1, we have
\begin{equation*}
Q_{n_k}(\hat{\gamma}_{n_k},\hat{g}_{n_k})\overset{a.s.}\to Q(\gamma',g')+\frac{\sigma^2}{2\beta} \ \mbox{as}\ n_k\to \infty
\end{equation*}
Since $\hat{\gamma}_{n_k}, \hat{g}_{n_k}$ minimizes $Q_{n_{k}}$, we have $Q_{n_k}(\hat{\gamma}_{n_k},\hat{g}_{n_k})\leq Q_{n_k}(\gamma,g)$ for all $\gamma, g$ in the parameter space. Therefore,
\begin{equation*}
\lim_{n_k\rightarrow\infty}Q_{n_k}(\gamma_{n_k},g_{n_k})\leq \lim_{n_k\rightarrow\infty}Q_{n_k}(\gamma_{0},g_{0})\overset{a.s.}=\frac{\sigma^2}{2\beta}.
\end{equation*}
This implies that $Q(\gamma',g')=0$. Since $Q$ has a unique minimum at $(\gamma_0,g_0)$, $(\gamma',g')$ must coincide with $(\gamma_0,g_0)$. This completes the proof.  
\end{proof}

\begin{proof}[Proof of Theorem~3] {\it Part (i)}: Note that
\begin{eqnarray}
    \hskip-20pt|\log\tilde{g}_{n}(z)-\log g_0(z)|&=&\left|\frac{\frac1{nh}\sum_{i=1}^n K\left(\frac{z-z_i}{h}\right)[\log\hat{g}_n(z_i)-\log g_0(z)] }{\frac1{nh}\sum_{i=1}^n K\left(\frac{z-z_i}{h}\right)}\right|\nonumber\\
    &\le&\frac{\frac1{nh}\sum_{i=1}^n K\left(\frac{z-z_i}{h}\right)|\log\hat{g}_n(z_i)-\log g_0(z_i)| }{\frac1{nh}\sum_{i=1}^n K\left(\frac{z-z_i}{h}\right)}
\nonumber\\&&+\left|\frac{\frac1{nh}\sum_{i=1}^n K\left(\frac{z-z_i}{h}\right)\log g_0(z_i) }{\frac1{nh}\sum_{i=1}^n K\left(\frac{z-z_i}{h}\right)}-\log g_0(z)\right|.\label{thm3_eq1}
\end{eqnarray}
From Theorem~2, we have 
\begin{eqnarray}
    \sup_{z} |\log \hat{g}_n(z)-\log g_0(z)| \to 0 \mbox{ almost surely},
\end{eqnarray}
i.e., for an arbitrary $\eta>0$ there exists a (possibly random) positive integer $N(\eta)$ such that
\begin{eqnarray}
   \sup_{z} |\log(\hat{g}_n(z))-\log (g_0(z))|<\eta \ \ \mbox{for all $n\ge N(\eta)$ with probability 1}.  
\end{eqnarray}
Therefore, the first term on the RHS of \eqref{thm3_eq1} converges uniformly over $z$ almost surely.
The second term is non-random. By using first-order Taylor series expansion of $\log g_0(z_i)$ around $z$, we have
\begin{eqnarray}\label{Taylor}
    \log g_0(z_i)=\log g_0(z)+ \frac{g_0'(\xi_i)}{g_0(\xi_i)} (z_i-z),
\end{eqnarray}
for some $\xi_i\in (z_i,z)$. 
It follows from the assumption on $g_0$ made in the theorem and the definition of $\mathcal{G}$ that
\begin{eqnarray}\label{g_deri_bound}
   \left |\frac{g_0'(z)}{g_0(z)}\right|\le \mathcal{C}_{g_0}
\end{eqnarray}
for some constant $\mathcal{C}_{g_0}$. Thus, the second term on the RHS of \eqref{thm3_eq1} is bounded as
\begin{eqnarray*}
\left|\frac{\displaystyle\mathop{\sum}_{i\in\{1,2,\ldots,n\}: |z-z_i|<h}  K\left(\frac{z-z_i}{h}\right)\frac{g_0'(\xi_i)}{g_0(\xi_i)} (z_i-z) }{\displaystyle\mathop{\sum}_{i\in\{1,2,\ldots,n\}: |z-z_i|<h}  K\left(\frac{z-z_i}{h}\right)}\right|&\le& \mathcal{C}_{g_0}h \mbox{ uniformly over } z,  
\end{eqnarray*}
which goes to zero as $h\to 0$. This completes the proof of part (i).\\
{\it Part (ii)}: We can expand $\tilde{Q}_n(\gamma)=\frac1n\sum_{i=1}^n
\left(y_i-\log \tilde{g}_n(z_i)-\log(1+\gamma x_i)\right)^2$ as follows.
\begin{eqnarray}
\tilde{Q}_n(\gamma)&=&\frac1n\sum_{i=1}^n\varepsilon_i^2+\frac{1}n \sum_{i=1}^n[\log g_0(z_i)-\log \tilde{g}_n(z_i)]^2
+\frac1n\sum_{i=1}^n[\log(1+\gamma_0 x_i)-\log(1+\gamma x_i)]^2\nonumber\\
&&+\frac{2}n \sum_{i=1}^n\varepsilon_i[\log g_0(z_i)-\log \tilde{g}_n(z_i)]
+\frac2n\sum_{i=1}^n\varepsilon_i[\log(1+\gamma_0 x_i)-\log(1+\gamma x_i)]\nonumber\\
&&+\frac{2}n \sum_{i=1}^n[\log g_0(z_i)-\log \tilde{g}_n(z_i)]
[\log(1+\gamma_0 x_i)-\log(1+\gamma x_i)].\label{Qn_tilde}
\end{eqnarray}
We will establish the convergence of each term on the RHS of \eqref{Qn_tilde}. The almost sure convergence of the first term to $\frac{\sigma^2}{2\beta}$ follows directly from \eqref{quad_err_conv}. 
The second term converges almost surely to~0 by virtue of Part~(i), since 
\begin{eqnarray}
    \left|\frac{1}n \sum_{i=1}^n[\log g_0(z_i)-\log \tilde{g}_n(z_i)]^2\right|&\le& \sup_{z} \left|\log g_0(z)-\log \tilde{g}_n(z)\right|^2.
\end{eqnarray}
From Assumption~2, the third term on the RHS of \eqref{Qn_tilde} can be viewed as the Riemann sum which converges as follows.
\begin{eqnarray*}
  &&\hskip-1in \frac1{\Delta}
\lim_{n\to\infty}\sum_{i=1}^n\left[\log\left(1+\gamma_0 x\left(\frac{i\Delta}{n}\right)\right)-\log\left(1+\gamma x\left(\frac{i\Delta}{n}\right)\right)\right]^2\frac{\Delta}{n}\\&&=
   \frac1{\Delta}
\int_{0}^{\Delta}[\log(1+\gamma_0 x(w)-\log(1+\gamma x(w))]^2dw.
\end{eqnarray*}
As for the fourth term, note that 
\begin{eqnarray}\label{Qn_tilde_4term}
    \left|\frac{2}n \sum_{i=1}^n\varepsilon_i[\log g_0(z_i)-\log \tilde{g}_n(z_i)]\right|&\le&2 \sup_{z} \left|\log g_0(z)-\log \tilde{g}_n(z)\right|\times \frac1n\sum_{i=1}^n|\varepsilon(i\rho_n)|.
\end{eqnarray}
The first factor on the RHS of \eqref{Qn_tilde_4term} converges almost surely to $0$, by virtue of Part (i). Note that
\begin{eqnarray}
 \left|\frac1n\sum_{i=1}^n|\varepsilon(i\rho_n)|-  \sqrt{\frac{\sigma^2 }{2\beta }}\sqrt{\frac{2}{\pi}}\right|&\le& \left|\frac1{n\rho_n}\int_{0}^{n\rho_n}(|\varepsilon(\lceil u/\rho_n \rceil\rho_n)|- |\varepsilon(u)|) du \right|\nonumber\\&&+ \left|\frac1{n\rho_n}\int_{0}^{n\rho_n} |\varepsilon(u)| du-  \sqrt{\frac{\sigma^2 }{2\beta }}\sqrt{\frac{2}{\pi}}\right|  
\end{eqnarray}
A similar argument as the one used to establish the convergence of \eqref{2term_2} shows that the second factor on the RHS of \eqref{Qn_tilde_4term} converges to $\sqrt{\frac{\sigma^2 }{2\beta }}\sqrt{\frac{2}{\pi}}$ almost surely. Thus, the product of the two factors converges to 0 almost surely. 
The fifth term on the RHS of \eqref{Qn_tilde} is identical to the third term on the RHS of~\eqref{Qn_decompose}, which has already been shown to converge almost surely to~$0$. Finally, the sixth term reduces, due to the boundedness of $\log(1+\gamma x_i )$, to 
\begin{eqnarray}
  &&\hskip-1in\left|\frac{2}n \sum_{i=1}^n[\log g_0(z_i)-\log \tilde{g}_n(z_i)]
[\log(1+\gamma_0 x_i)-\log(1+\gamma x_i)]\right|\nonumber\\&\le&
2 \sup_{z} \left|\log g_0(z)-\log \tilde{g}_n(z)\right|\times \sup_{x,\gamma} \left|\log(1+\gamma_0 x_i)-\log(1+\gamma x_i)\right|,\label{sixth_term}
\end{eqnarray}
which converges to 0 almost surely by Part (i). Thus, we have
\begin{equation}
\lim_{n\to\infty}\tilde{Q}_n(\gamma)= \frac{\sigma^2}{2\beta}+\frac1{\Delta}
\int_{0}^{\Delta}[\log(1+\gamma_0 x(w)-\log(1+\gamma x(w))]^2dw,\ \mbox{almost surely}.
\end{equation}
The almost sure convergence of $\tilde{\gamma}_n=\arg\min_{\gamma\in\Gamma}\tilde{Q}_n(\gamma)$ to $\gamma_0$ follows along the lines of the proof of Theorem~2. This completes the proof of Part (ii).
\end{proof}

\begin{proof}[Proof of Theorem~4] {\it Part (i)}: 
By using equation (12) of the main manuscript, we have
\begin{eqnarray}
\tilde{\varepsilon}_i
&=&\varepsilon(i\rho_n)+\mathsf{T}_n(d_i) \mbox{ for } i=1,2,\ldots,n,\label{tilde_eps_sys}
\end{eqnarray}
where 
\begin{equation}
 \mathsf{T}_n(d)= \log (1+\gamma_0 x(d))+\log g_0(z(d))-\log (1+\tilde{\gamma}_n x(d))-\log \tilde{g}_n(z(d)).\label{def_Tn_sys}
\end{equation}
Now, by using \eqref{tilde_eps_sys}, we have
\begin{eqnarray}
\frac1n\sum_{i=1}^{n-1}\tilde{\varepsilon}_i\tilde{\varepsilon}_{i+1}&=& \frac1n\sum_{i=1}^{n-1}\varepsilon(i\rho_n)\varepsilon(i\rho_n+\rho_n)+\frac1n\sum_{i=1}^{n-1}\mathsf{T}_n(d_i)\mathsf{T}_n(d_{i+1})+\frac1n\sum_{i=1}^{n-1}\varepsilon(i\rho_n)\mathsf{T}_n(d_{i+1})\nonumber\\
&&+\frac1n\sum_{i=1}^{n-1}\varepsilon(i\rho_n+\rho_n)\mathsf{T}_n(d_i)\label{tilde_beta_num}\\
\frac1n\sum_{i=1}^{n}\tilde{\varepsilon}^2_i&=& \frac1n\sum_{i=1}^{n}\varepsilon^2(i\rho_n)+\frac1n\sum_{i=1}^{n}\mathsf{T}_n^2(d_i)+\frac2n\sum_{i=1}^{n}\varepsilon(i\rho_n)\mathsf{T}_n(d_i)\label{tilde_beta_den}.
\end{eqnarray}
By using equation (11) of the main manuscript, we have 
\begin{eqnarray}
    \varepsilon(i\rho_n+\rho_n)&=& e^{-\beta \rho_n}\left(\varepsilon(i\rho_n)+ \sigma \int_{i\rho_n}^{i\rho_n+\rho_n}e^{\beta (s-i\rho_n)} dW(s)\right).\label{eps_i+1}
\end{eqnarray}
Therefore, by using \eqref{tilde_beta_num}, \eqref{eps_i+1},  we have
\begin{eqnarray}
  \frac1n\sum_{i=1}^{n-1}\tilde{\varepsilon}_i\tilde{\varepsilon}_{i+1}&=& e^{-\beta \rho_n}\left(\frac1n\sum_{i=1}^{n-1}\varepsilon^2(i\rho_n)+\frac1n\sum_{i=1}^{n-1}\varepsilon(i\rho_n)\sigma \int_{i\rho_n}^{i\rho_n+\rho_n}e^{\beta (s-i\rho_n)} dW(s)\right.\nonumber\\&&\left.\hskip40pt+\frac{e^{\beta\rho_n}}{n}\sum_{i=1}^{n-1}\mathsf{T}_n(d_i)\mathsf{T}_n(d_{i+1})+  \frac{1}{n}\sum_{i=1}^{n-1}\varepsilon(i\rho_n)(e^{\beta\rho_n}\mathsf{T}_n(d_{i+1})+\mathsf{T}_n(d_{i}))\right.\nonumber\\
  &&\left.\hskip40pt+\frac{1}{n}\sum_{i=1}^{n-1}\sigma \int_{i\rho_n}^{i\rho_n+\rho_n}e^{\beta (s-i\rho_n)} dW(s)\mathsf{T}_n(d_{i})
  \right).\mbox{}\quad\label{tilde_beta_num1}
\end{eqnarray}
Now, by using equation (13) of the main manuscript, \eqref{tilde_beta_den} and  \eqref{tilde_beta_num1}, we have
\begin{eqnarray}
   \tilde{\beta}_n-\beta&=& -\frac1{\rho_n}\log\left(1+\frac{\mathsf{U}_{1n}}{\mathsf{D}_{n}}+\frac{\mathsf{U}_{2n}}{\mathsf{D}_{n}}+\frac{\mathsf{U}_{3n}}{\mathsf{D}_{n}}+\frac{\mathsf{U}_{4n}}{\mathsf{D}_{n}}+\frac{\mathsf{U}_{5n}}{\mathsf{D}_{n}}\right)\label{tilde_beta_gap}
,\end{eqnarray}
where
\begin{eqnarray}
  \mathsf{U}_{1n}&=&\frac1n\sum_{i=1}^{n-1}\varepsilon(i\rho_n)\sigma \int_{i\rho_n}^{i\rho_n+\rho_n}e^{\beta (s-i\rho_n)} dW(s)\label{def_U1n}\\
  \mathsf{U}_{2n}&=& \frac{1}{n}\sum_{i=1}^{n-1}(e^{\beta\rho_n}\mathsf{T}_n(d_i)\mathsf{T}_n(d_{i+1})-\mathsf{T}_n^2(d_i))\label{def_U2n}\\
  \mathsf{U}_{3n}&=& \frac{1}{n}\sum_{i=1}^{n-1}\varepsilon(i\rho_n)(e^{\beta\rho_n}\mathsf{T}_n(d_{i+1})-\mathsf{T}_n(d_{i}))\label{def_U3n}\\
  \mathsf{U}_{4n}&=& \frac{1}{n}\sum_{i=1}^{n-1}\sigma \int_{i\rho_n}^{i\rho_n+\rho_n}e^{\beta (s-i\rho_n)} dW(s)\mathsf{T}_n(d_{i})\label{def_U4n}\\
  \mathsf{U}_{5n}&=& -\frac{1}{n}\varepsilon^2(n\rho_n)-\frac1n\mathsf{T}_n^2(d_n)-\frac2n\varepsilon(n\rho_n)\mathsf{T}_n(d_n)\label{def_U5n}\\
  \mathsf{D}_{n}&=& \frac1n\sum_{i=1}^{n}\varepsilon^2(i\rho_n)+\frac1n\sum_{i=1}^{n}\mathsf{T}_n^2(d_i)+\frac1n\sum_{i=1}^n \varepsilon(i\rho_n)\mathsf{T}_n(d_i).\label{def_Dn}
\end{eqnarray}
Now by using \eqref{tilde_beta_gap}, we have
\begin{eqnarray}
    \left|\tilde{\beta}_n-\beta\right|&\le& \frac{\frac1{\rho_n}|\mathsf{U}_{1n}|}{\mathsf{D}_{n}}+\frac{\frac1{\rho_n}|\mathsf{U}_{2n}|}{\mathsf{D}_{n}}+\frac{\frac1{\rho_n}|\mathsf{U}_{3n}|}{\mathsf{D}_{n}}+\frac{\frac1{\rho_n}|\mathsf{U}_{4n}|}{\mathsf{D}_{n}}+\frac{\frac1{\rho_n}|\mathsf{U}_{5n}|}{\mathsf{D}_{n}}.   \label{tilde_beta_gap_bound}
\end{eqnarray}
To complete the proof, we show that each term on the RHS of \eqref{tilde_beta_gap_bound} converges to 0. 
Note that $\mathsf{T}_n(d)$ as given in  \eqref{def_Tn_sys} is a Lipschitz continuous function by using Assumptions~1 and~2. Therefore, by using Theorem~3, for any $\delta>0$, there exists $N_0$ (does not depend on $d_i$) such that
\begin{eqnarray}
 |\mathsf{T}_n(d_i)|\le \delta \mbox{ for all } n\ge N_0 \mbox{ 
 almost surely uniformly over } d_i. \label{Tn_bound}    
\end{eqnarray} 
Therefore, by using \eqref{quad_err_conv}, \eqref{Tn_bound} and Lemma~\ref{lem:OU}, we have
\begin{eqnarray}
\lim_{n\to\infty}\mathsf{D}_n = \frac{\sigma^2}{2\beta} \qquad \mbox{ almost surely.}    
\end{eqnarray}
Now by using Assumption~3, $E\left(\frac1{\rho_n} \mathsf{U}_{1n}\right)=0$. Now, note that 
\begin{eqnarray}
   Var\left( \frac{1}{\rho_n}\mathsf{U}_{1n}\right)&=& \frac1{n^2\rho_n^2} \sum_{i=1}^{n-1} Var\left(\varepsilon(i\rho_n)\sigma \int_{i\rho_n}^{i\rho_n+\rho_n}e^{\beta (s-i\rho_n)} dW(s)\right)\nonumber\\ &&+\frac2{n^2\rho_n^2} \sum_{i=1}^{n-1}\sum_{j=1}^{i-1} Cov\left(\varepsilon(i\rho_n)\sigma \int_{i\rho_n}^{i\rho_n+\rho_n}e^{\beta (s-i\rho_n)} dW(s), \right.\nonumber\\&&\hskip1.5in\left.\ \varepsilon(j\rho_n)\sigma \int_{j\rho_n}^{j\rho_n+\rho_n}e^{\beta (s-j\rho_n)} dW(s)\right).\label{var_u1n}
\end{eqnarray}
For $j<i$, by using the independence of Wiener integrals, the second term on the RHS of \eqref{var_u1n} vanishes. Therefore, by using Assumption~3 and Wiener integration, \eqref{var_u1n} reduces as follows.
\begin{eqnarray}
   Var\left( \frac{1}{\rho_n}\mathsf{U}_{1n}\right)&=& \frac{(n-1)}{n^2\rho_n^2}  \frac{\sigma^2}{2\beta} \left(\frac{e^{2\beta\rho_n}-1}{2\beta}\right)=o\left(\frac{1}{n\rho_n^2}\right).
\end{eqnarray}
Therefore, we have 
\begin{eqnarray}
    \frac1{\rho_n} \mathsf{U}_{1n}\to 0 \mbox{ in probability}. 
\end{eqnarray}
Since $d_{i+1}-d_i=\frac{\Delta}{n}$, by using Lipschitz continuity of $\mathsf{T}_n$ and \eqref{Tn_bound}, we have
\begin{eqnarray}
    \left|\frac{T_n(d_i)-e^{\beta\rho_n}T_n(d_{i+1})}{\rho_n}\right|= o(1) \mbox{ almost surely uniformly over } d_i\mbox{'s}.\label{Tn_variation}
\end{eqnarray}
Now, by using \eqref{Tn_variation}, we have
\begin{eqnarray}
  \frac{1}{\rho_n}|\mathsf{U}_{2n}|&=&o(1) \mbox{ almost surely}. 
\end{eqnarray}
Using a similar line of argument and Lemma~\ref{lem:OU}, we have 
\begin{eqnarray}
  \frac{1}{\rho_n}|\mathsf{U}_{3n}|&=&o(1) \mbox{ almost surely}.
\end{eqnarray}   
Note that we have
\begin{eqnarray}
    \frac1{\rho_n}\mathsf{U}_{4n}&=& \frac1{\sqrt{n\rho_n}} \frac{1}{\sqrt{n}}\sum_{i=1}^{n-1}  \sqrt{\frac{e^{2\beta \rho_n}-1}{2\beta\rho_n}}\left( \frac{\sigma}{\sqrt{\frac{e^{2\beta \rho_n}-1}{2\beta\rho_n}}}\int_{i\rho_n}^{i\rho_n+\rho_n}e^{\beta (s-i\rho_n)} dW(s)\right)\mathsf{T}_n(d_{i}).\qquad\mbox{}
\end{eqnarray}
Since the Wiener integrals $ \frac{\sigma}{\sqrt{\frac{e^{2\beta \rho_n}-1}{2\beta\rho_n}}}\int_{i\rho_n}^{i\rho_n+\rho_n}e^{\beta (s-i\rho_n)} dW(s)$  are independent zero mean normal random variables with variance $\sigma^2$  and $\sqrt{\frac{e^{2\beta \rho_n}-1}{2\beta\rho_n}} \mathsf{T}_n(d_i)=o(1)$ almost surely uniformly over $d_i$s, we have
\begin{eqnarray}
    \frac1{\rho_n} \mathsf{U}_{4n}\to 0 \mbox{ in probability}. 
\end{eqnarray}
Now, by using equation (11) of the main manuscript, \eqref{Tn_bound} and Lemma~\ref{lem:OU}, we have
\begin{eqnarray}
  \frac{1}{\rho_n}|\mathsf{U}_{5n}|&=&O\left(\frac{1}{n\rho_n}\right) \mbox{ almost surely}.
\end{eqnarray}   
This completes the proof.

\noindent
{\it Part (ii)}: The proof follows using Part (i) and \eqref{quad_err_conv}.
\end{proof}

\section{Additional bootstrap simulation results}\label{s.s2}
In this section, we report the results of additional simulation study as indicated in section~5.3 of the main manuscript. The simulation set up is the same as described in section~5.1 of the main manuscript.  

We first report the distribution of $0.025$ and $0.975$ bootstrap quantiles of $\tilde\gamma$. 
The left and right columns of Figure~\ref{quantile_hist} show the histograms of these quantiles along with  a vertical line indicating the position of $\gamma_0$. 
The rows of Figure~\ref{quantile_hist} correspond to sample sizes $n=750, \ 1500, \ 3000$ and $6000$, respectively.
The empirical percentages of the occasions where the lower qunatile lies above $\gamma_0$ or the upper quantile lies below $\gamma_0$ are indicated on the plots. With increasing sample size, these empirical percentages reduce and the histograms are more concentrated around values away from $\gamma_0$.   

\begin{figure}[!h]\label{quantile_hist}
\centering
\begin{tabular}{cc}
2.5\% bootstrap quantile of $\gamma$ & 97.5\% bootstrap quantile of $\gamma$ \\
\includegraphics[width=2.5in,height=1.5in]{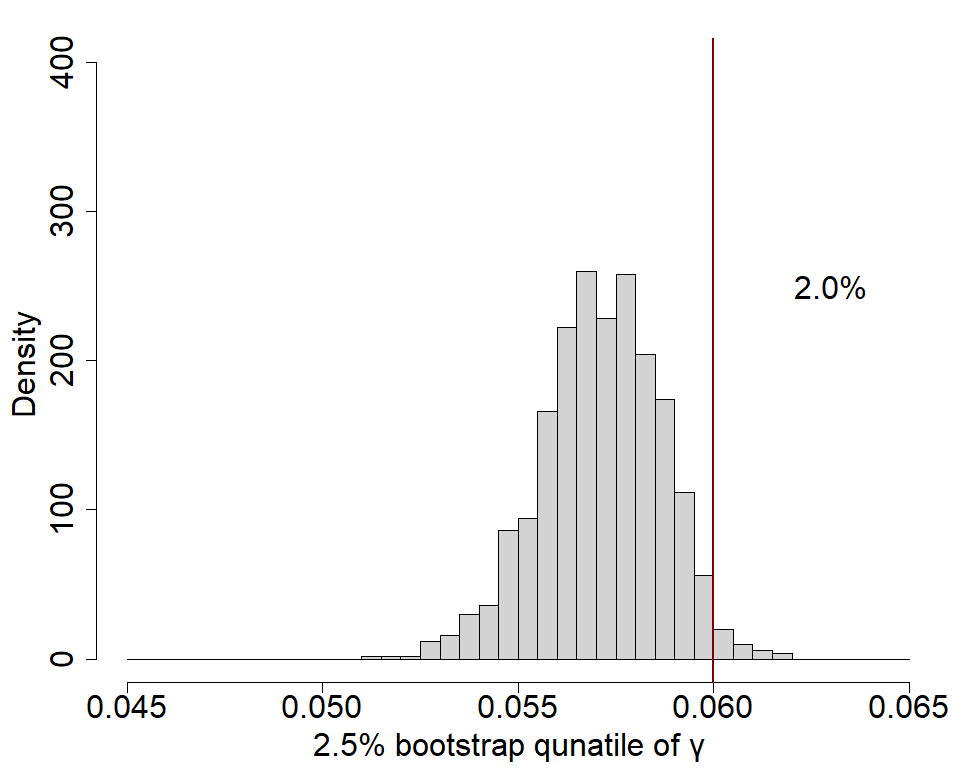}
&
\includegraphics[width=2.5in,height=1.5in]
{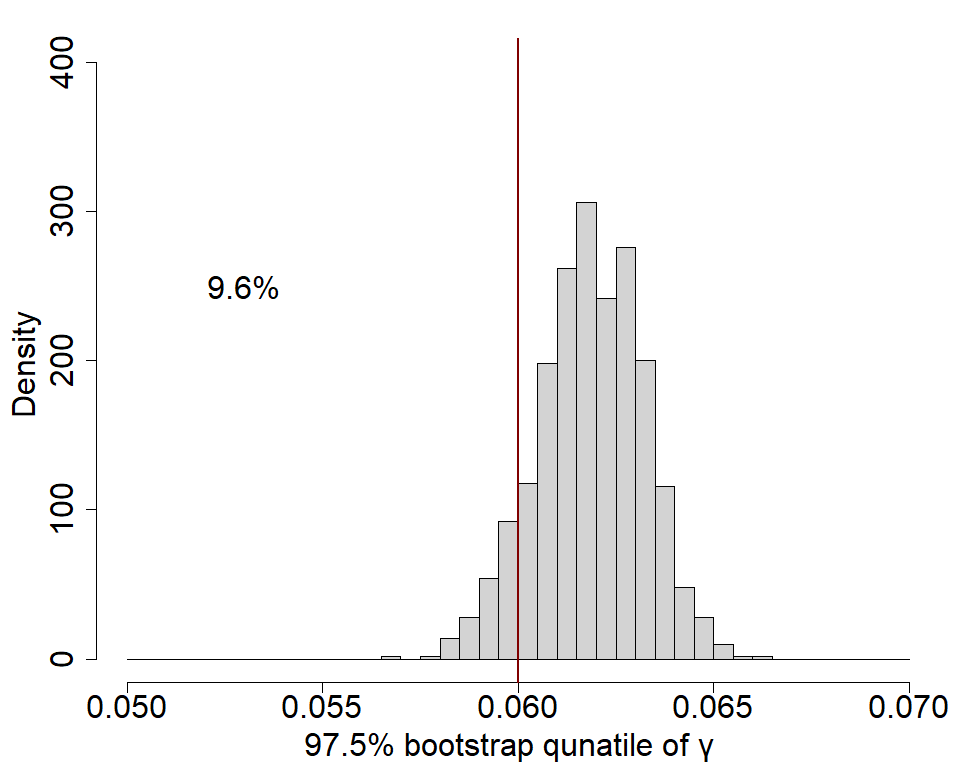}\\
\includegraphics[width=2.5in,height=1.5in]{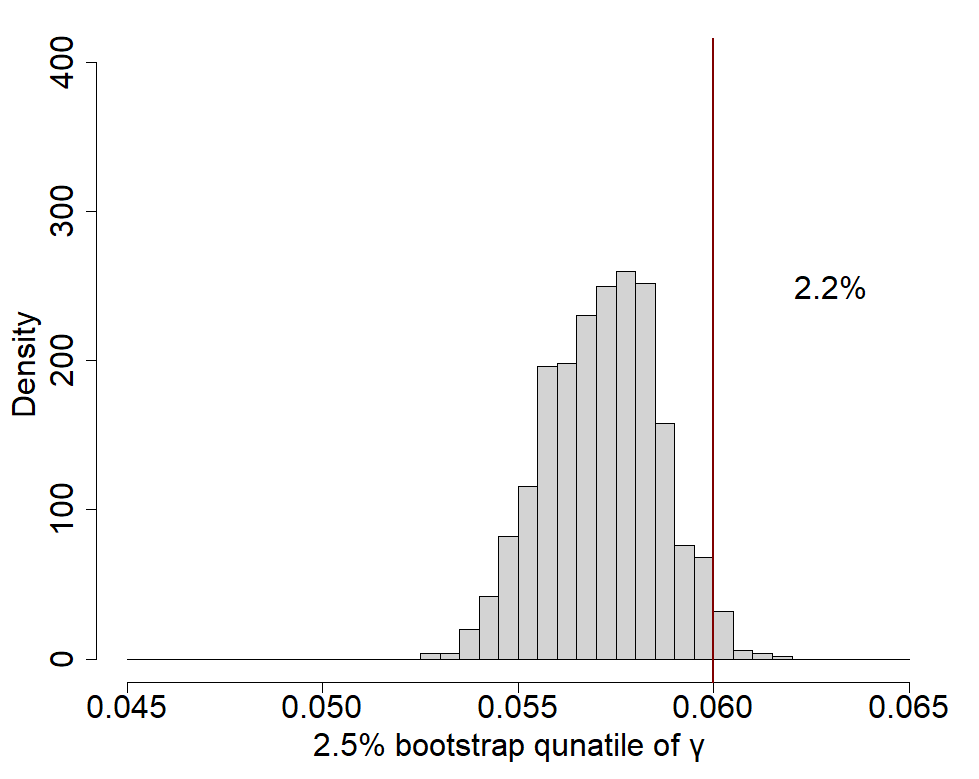} &
\includegraphics[width=2.5in,height=1.5in]
{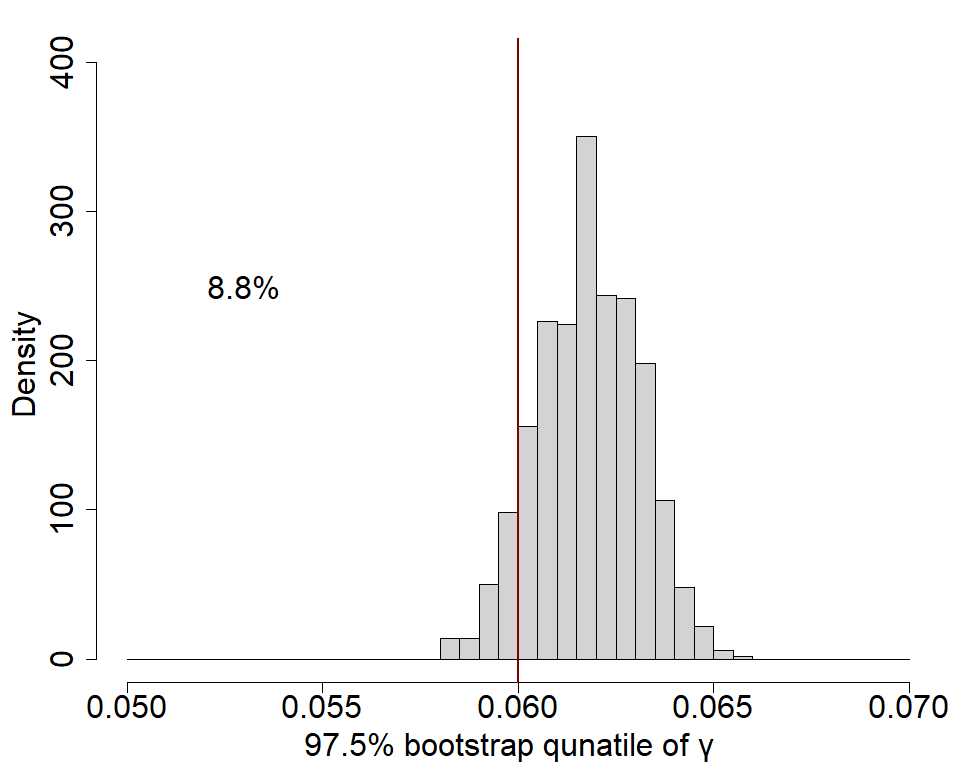}\\
\includegraphics[width=2.5in,height=1.5in]{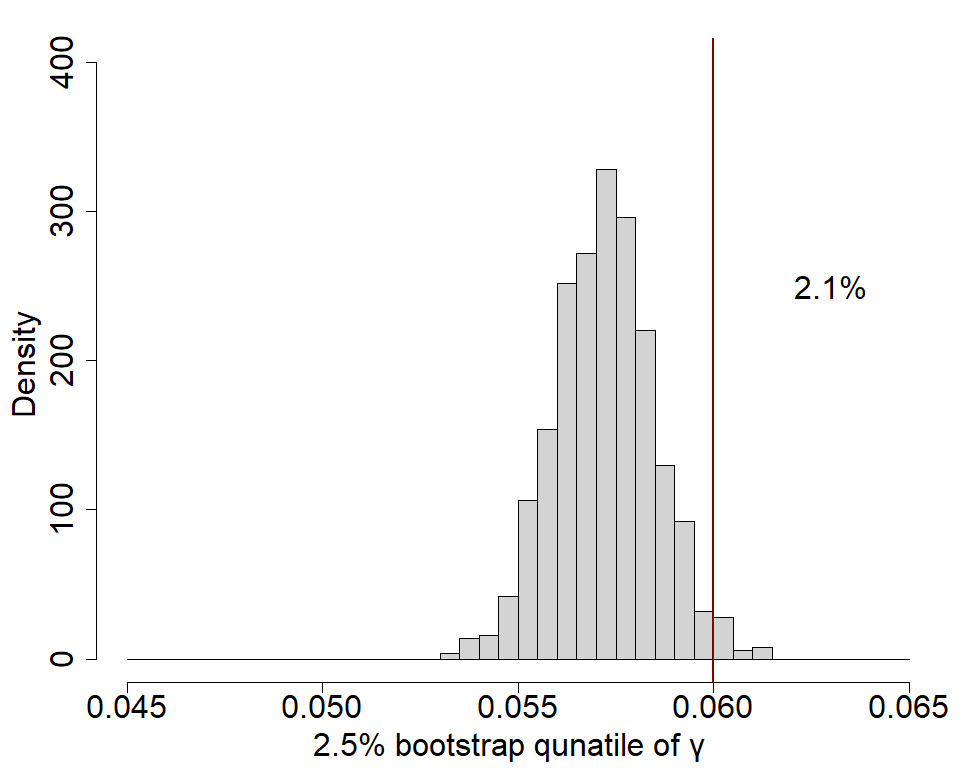} &
\includegraphics[width=2.5in,height=1.5in]
{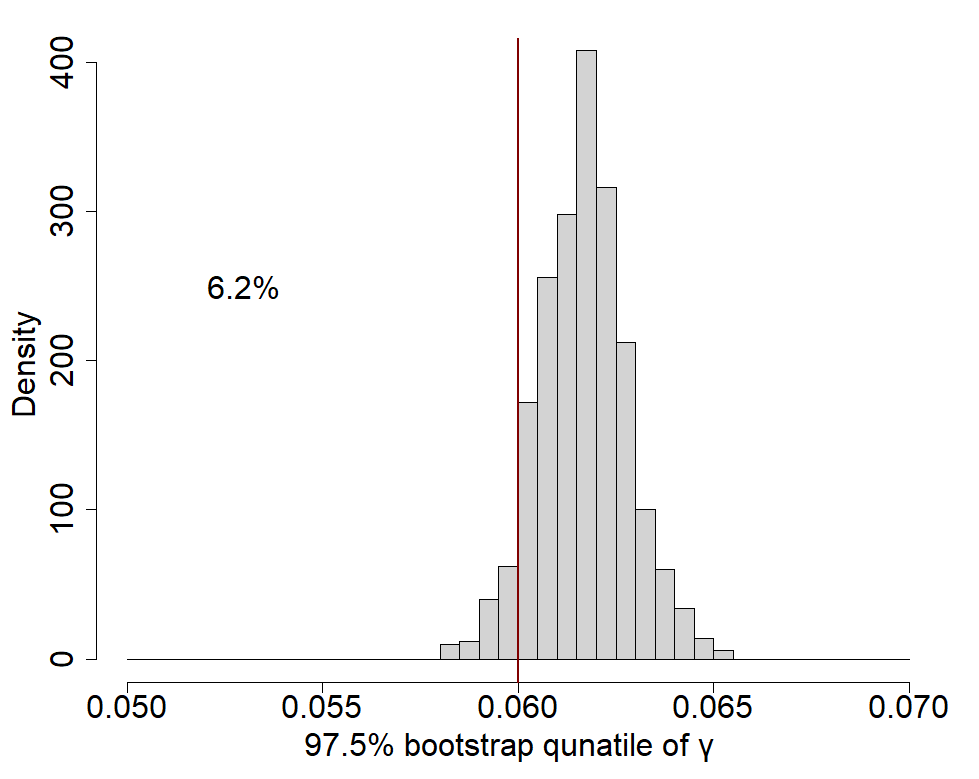}\\
\includegraphics[width=2.5in,height=1.5in]{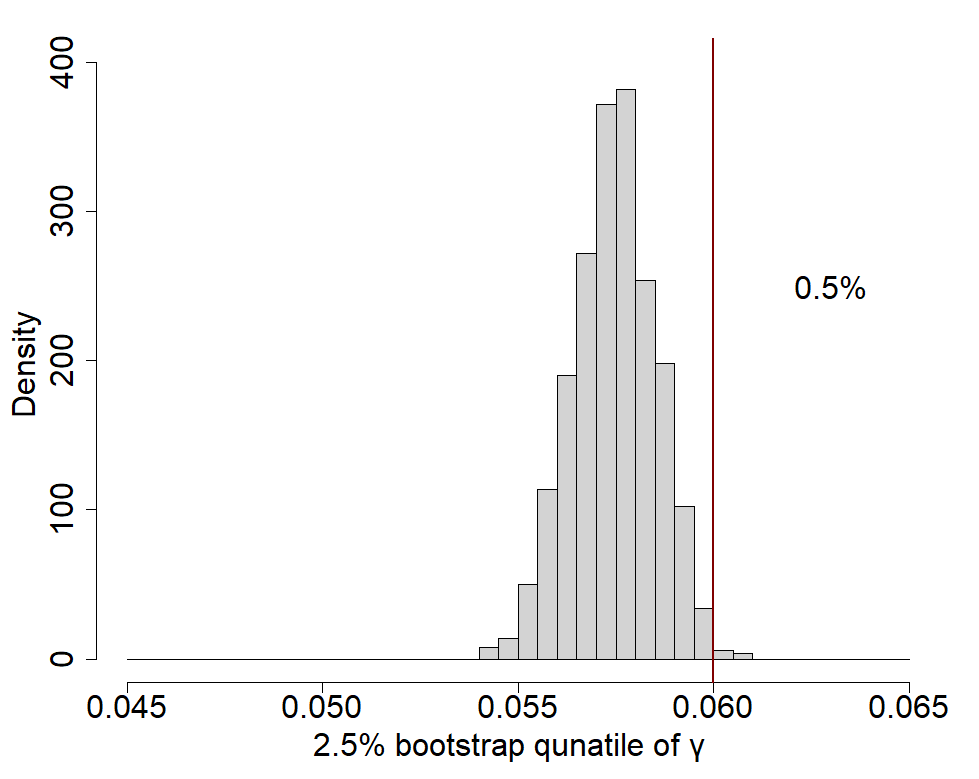} &
\includegraphics[width=2.5in,height=1.5in]
{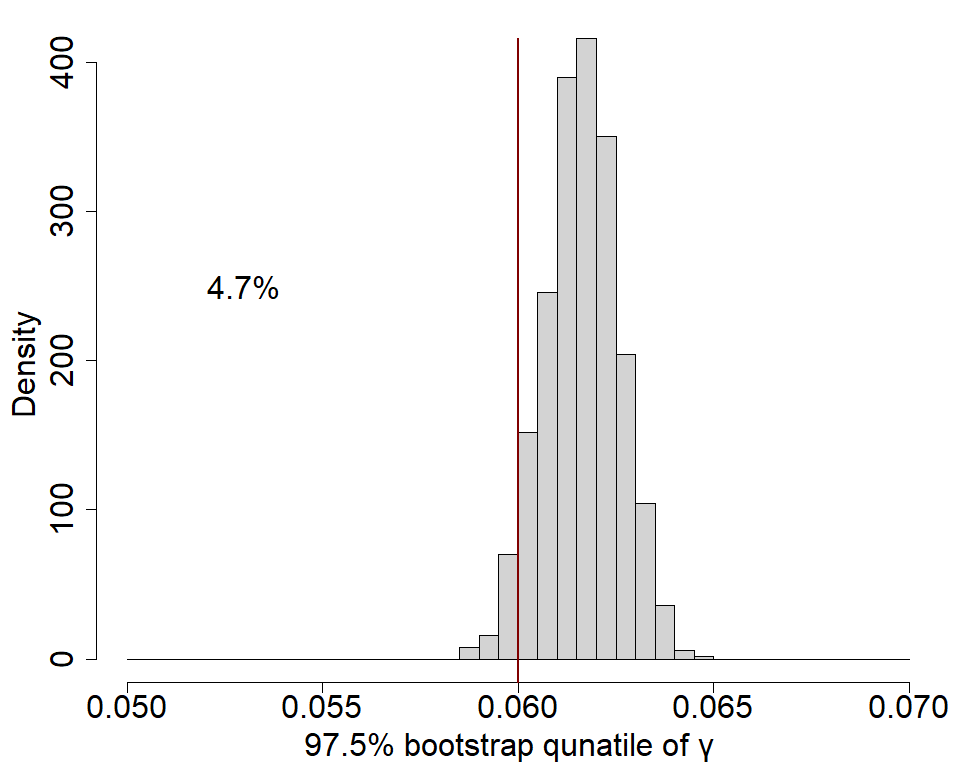}
\end{tabular}
\caption{Histogram of 0.025 bootstrap qunatile  (left) and histogram of 0.975 bootstrap quantile (right) of $\tilde\gamma$,  along with $\gamma_0$ (maroon vertical line). The rows from top to bottom correspond to sample sizes, $n= 750, 
\ 1500, \ 3000$ and $6000$, respectively.}
\label{f4}
\end{figure}

\begin{figure}[ht!]
\centering
\begin{tabular}{cc}
Average  CI of log(g) & se of log(g) \\
\includegraphics[width=2.5in,height=1.5in]{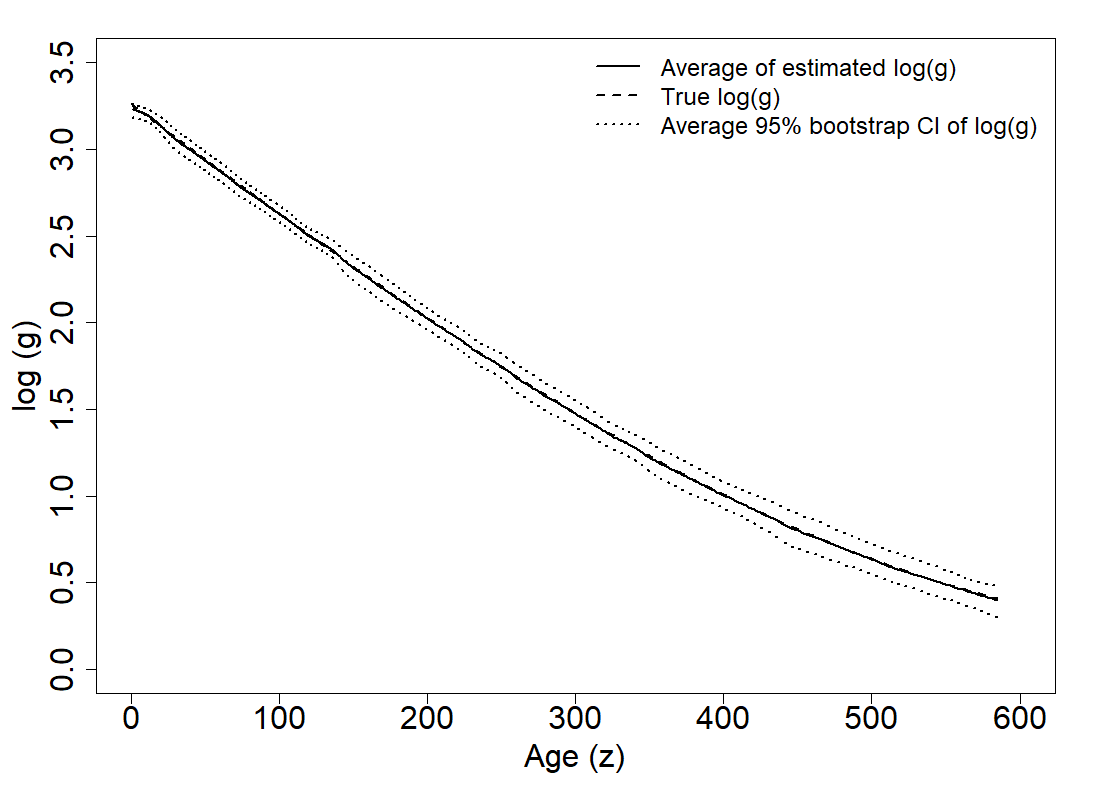} &
\includegraphics[width=2.5in,height=1.5in]{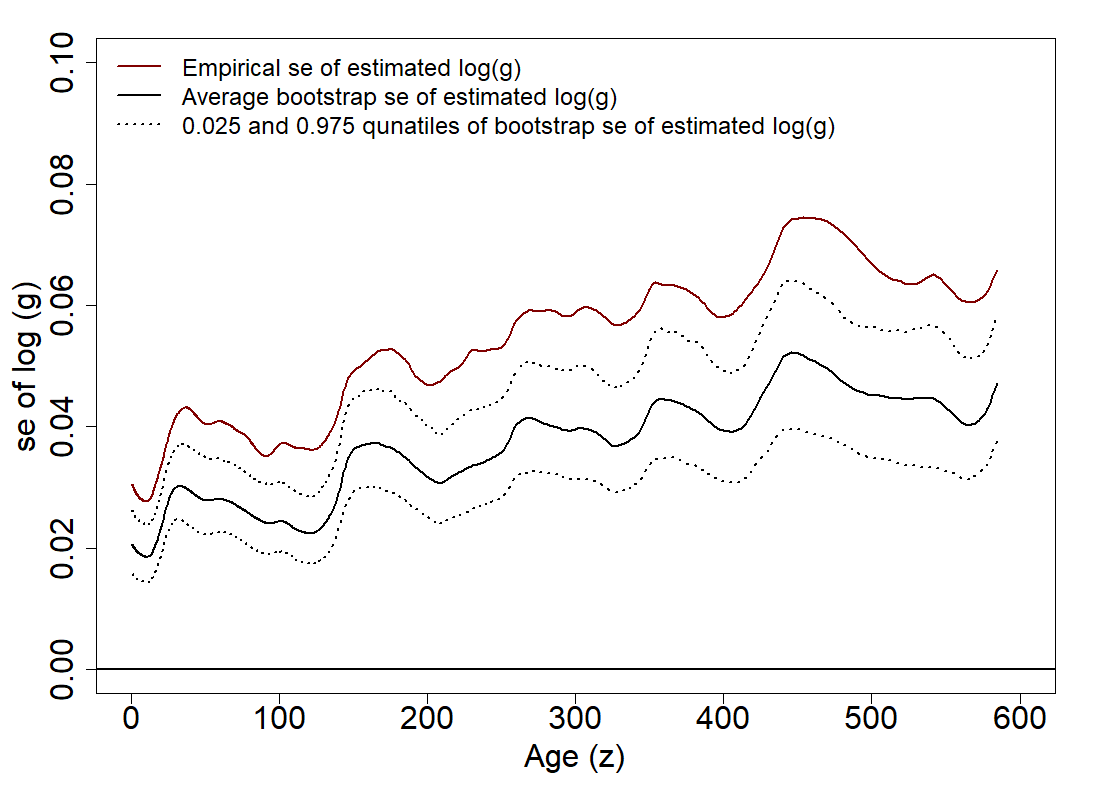}\\
\includegraphics[width=2.5in,height=1.5in]{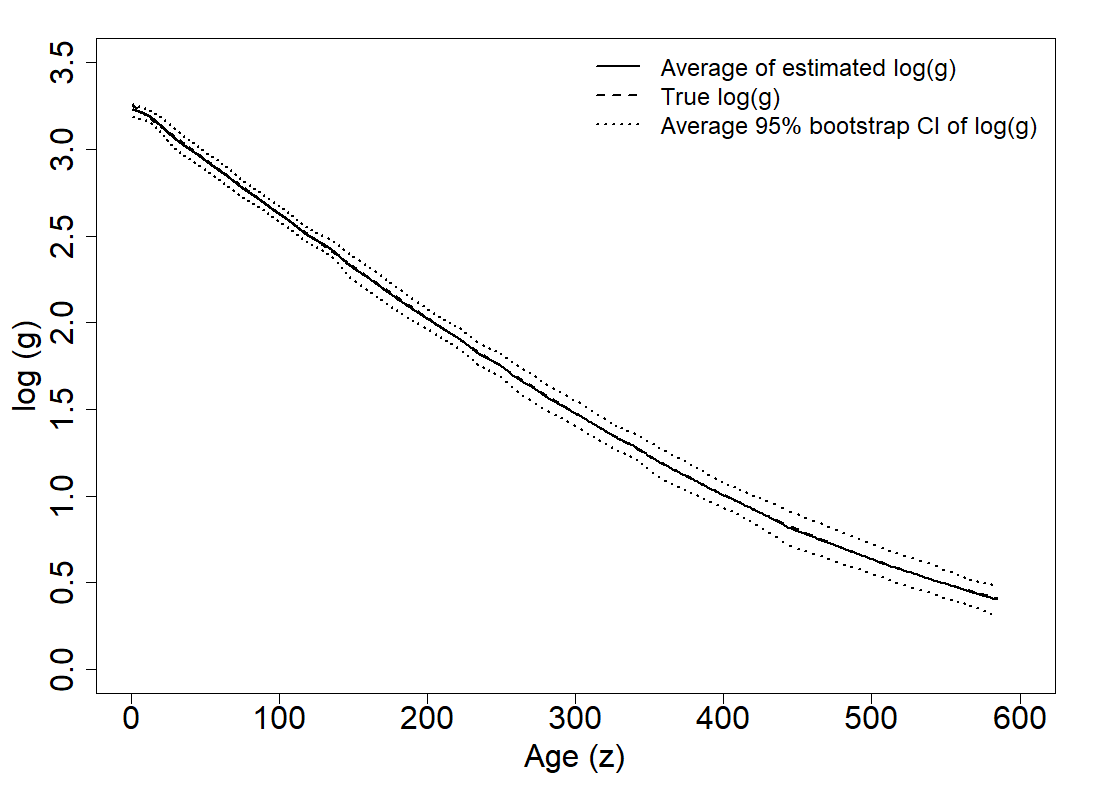} &
\includegraphics[width=2.5in,height=1.5in]
{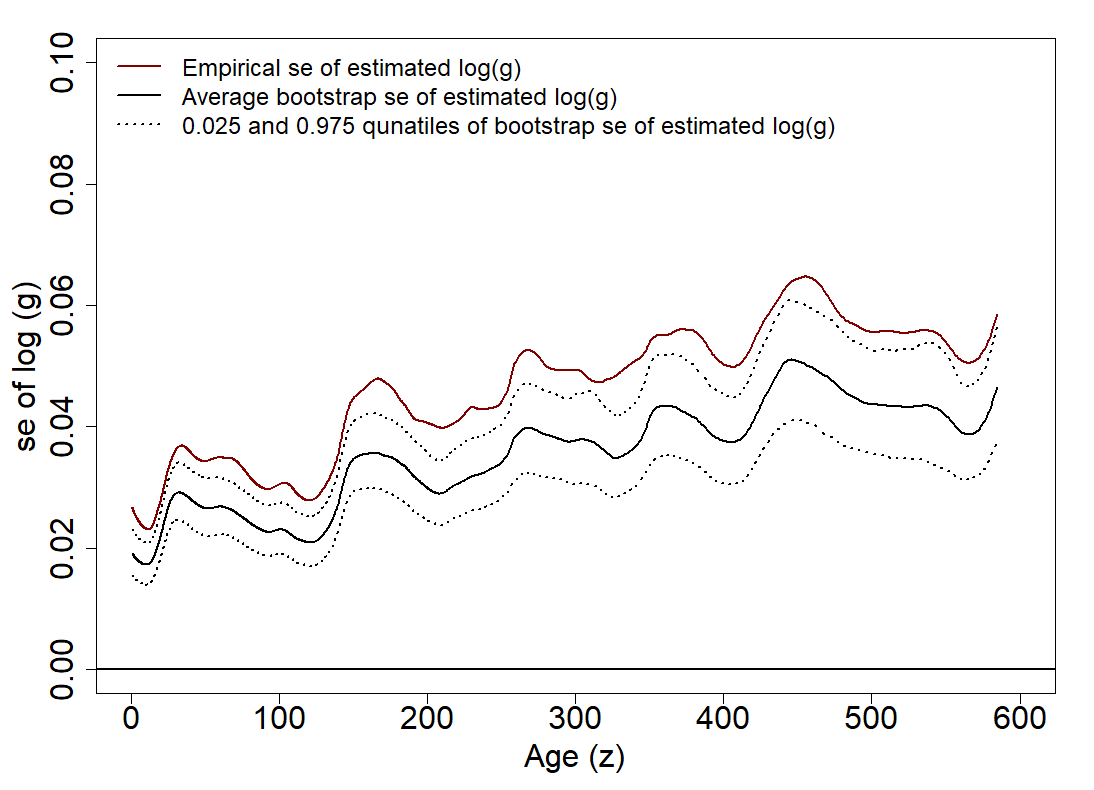}\\
\includegraphics[width=2.5in,height=1.5in]{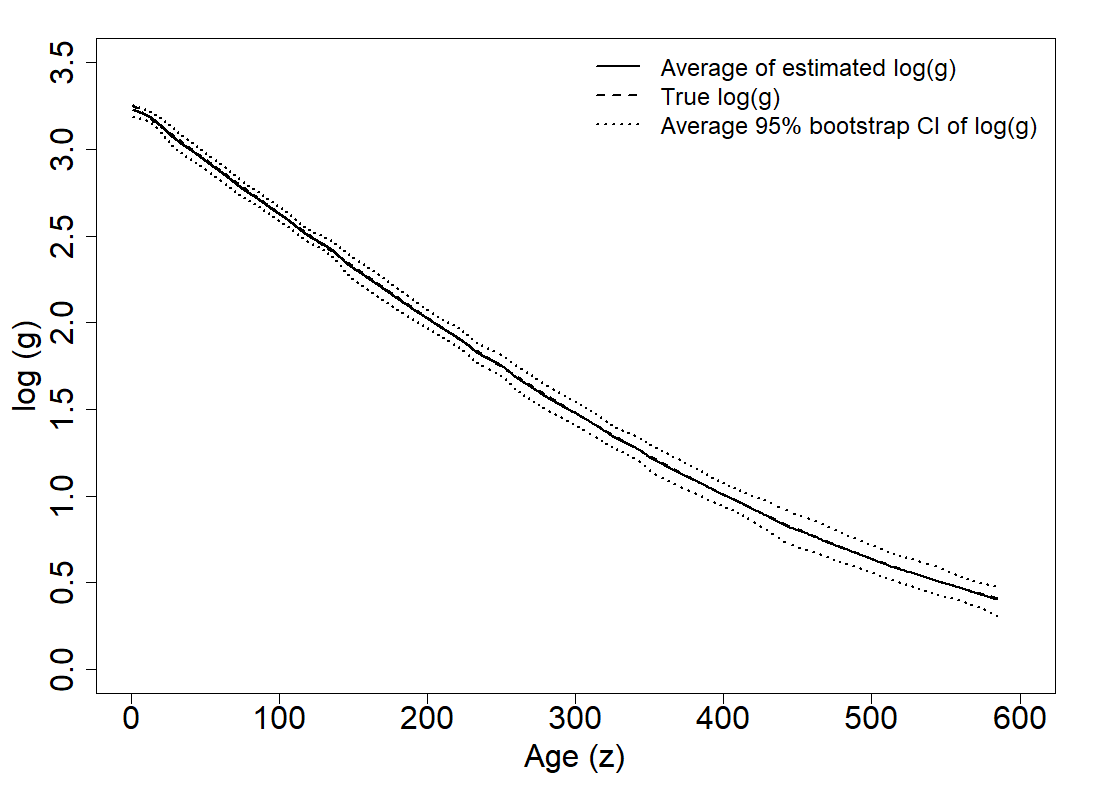} &
\includegraphics[width=2.5in,height=1.5in]
{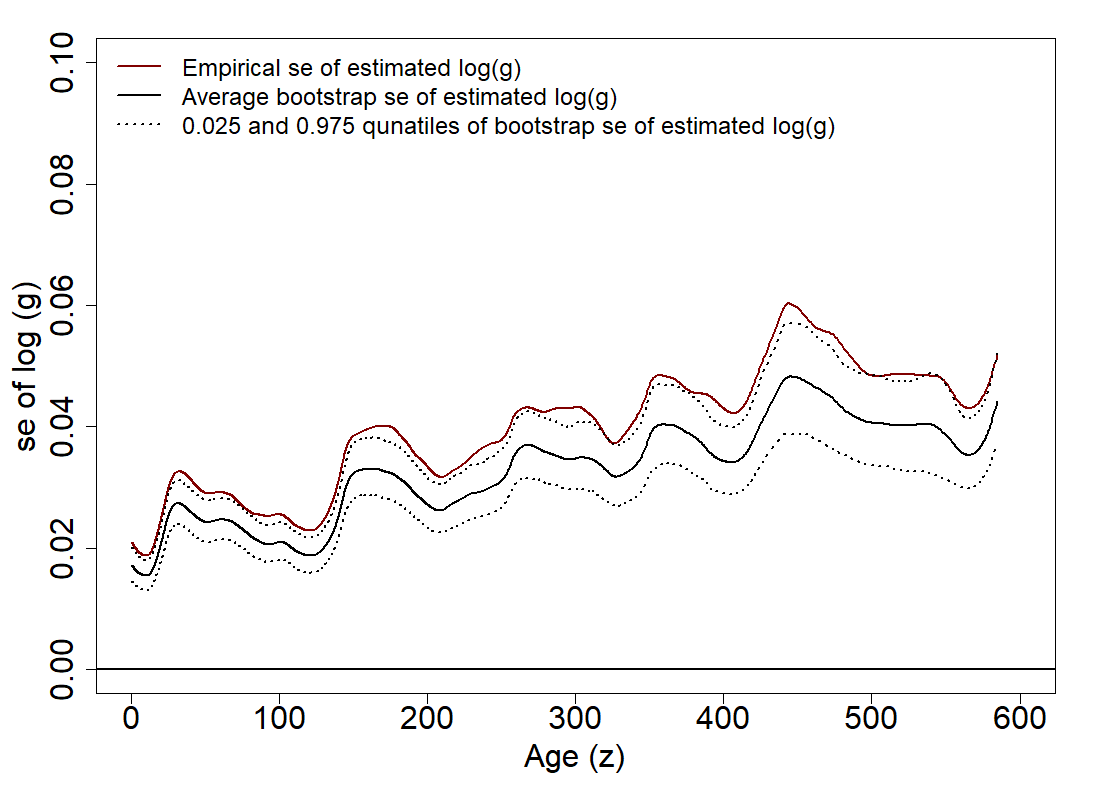}\\
\includegraphics[width=2.5in,height=1.5in]{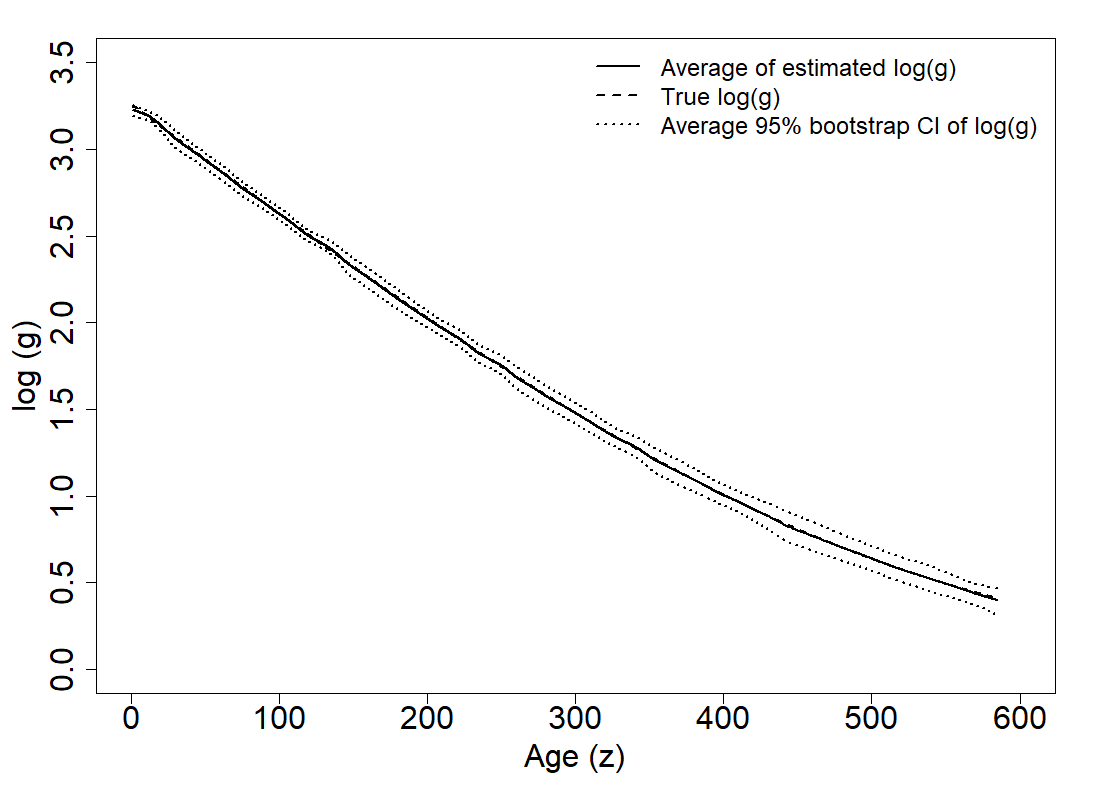} &
\includegraphics[width=2.5in,height=1.5in]
{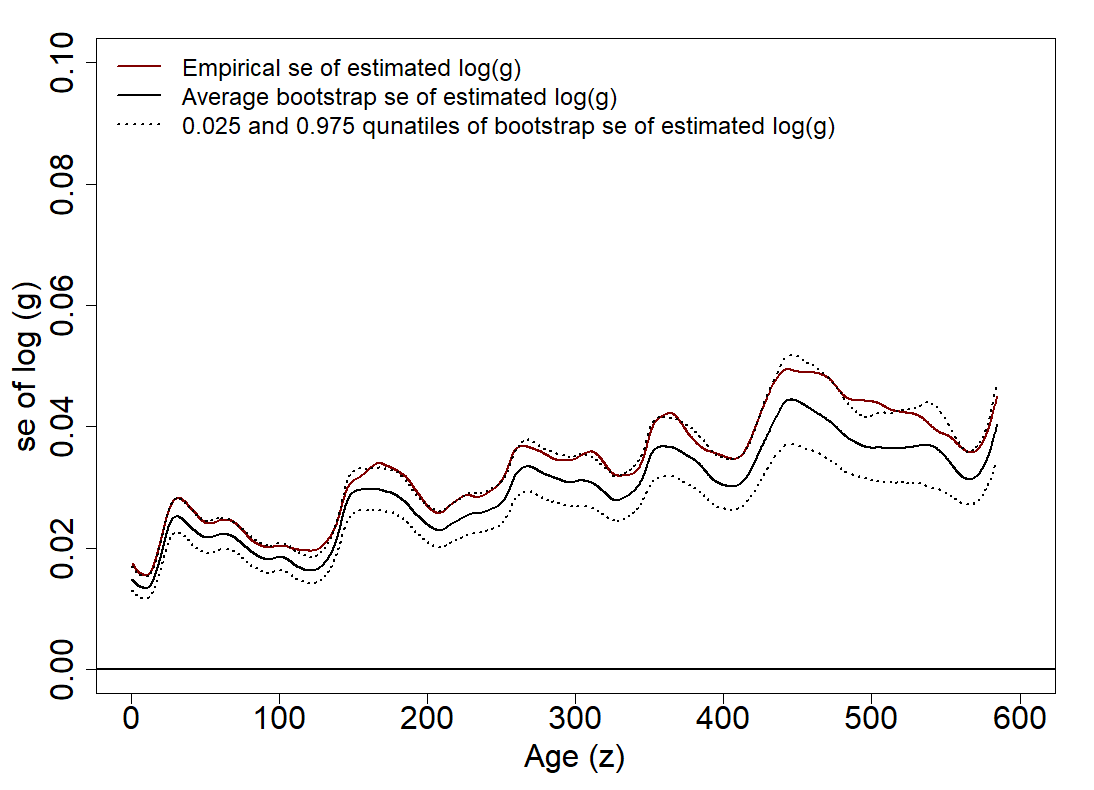}
\end{tabular}
\caption{Average of bootstrap confidence limits of $\log\tilde g_n$ (left) and the bootstrap standard errors of $\log\tilde` g_n$ (right). The rows from top to bottom correspond to sample size, $n= 750, 
\ 1500, \ 3000$ and $6000$, respectively}
\label{f4}
\end{figure}

The left panel of Figure~\ref{f4} shows the plots of $\log g_0$ as a function of age, the average estimate (across runs) of $\log\tilde g_n$, and the average of $95\%$ bootstrap confidence limits of $\log g_0$. The right panel of Figure~\ref{f4} shows the plots of the empirical pointwise standard errors (across runs) of $\log\tilde g_n$ in solid maroon line as a function of age, and the average, the $0.025$ quantile and the $0.975$ quantile (across runs) of the bootstrap pointwise standard errors of $\log\tilde g_n$. The rows of Figure~\ref{f4} correspond to sample sizes $n=750, \ 1500, \ 3000$ and $6000$, respectively. It is seen that $\log g_0$ lies in the average of bootstrap confidence limits for all the sample sizes. The average width of the bootstrap confidence limits also seen to be smaller for larger sample sizes. The lower and upper quantiles of standard errors of $\log\tilde g_n$ contains the empirical pointwise standard errors (across runs) for sample size 3000 and 6000.

\bibliographystyle{Chicago}

\bibliography{AAR_Arxiv}

\begin{thebibliography}{}

\bibitem[\protect\citeauthoryear{A{\"\i}t-Sahalia}{A{\"\i}t-Sahalia}{2002}]{ait2002}
A{\"\i}t-Sahalia, Y. (2002).
\newblock Maximum likelihood estimation of discretely sampled diffusions: {A}
  closed-form approximation approach.
\newblock {\em Econometrica\/}~{\em 70\/}(1), 223--262.

\bibitem[\protect\citeauthoryear{Algarra, Nieto, Ramos, Eiras-Barca, Trigo, and
  Gimeno}{Algarra et~al.}{2020}]{Algarra2020}
Algarra, I., R.~Nieto, A.~M. Ramos, J.~Eiras-Barca, R.~M. Trigo, and L.~Gimeno
  (2020).
\newblock Significant increase of global anomalous moisture uptake feeding
  landfalling atmospheric rivers.
\newblock {\em Nature {C}ommunications\/}~{\em 11\/}(1), 2041--1723.

\bibitem[\protect\citeauthoryear{Allan and Soden}{Allan and
  Soden}{2008}]{Allan2008}
Allan, R.~P. and B.~J. Soden (2008).
\newblock Atmospheric warming and the amplification of precipitation extremes.
\newblock {\em Science\/}~{\em 321\/}(5895), 1481--1484.

\bibitem[\protect\citeauthoryear{{Bazin}, {Landais}, {Lemieux-Dudon}, {Toy\'{e}
  Mahamadou Kele}, {Veres}, {Parrenin}, {Martinerie}, {Ritz}, {Capron},
  {Lipenkov}, {Loutre}, {Raynaud}, {Vinther}, {Svensson}, {Rasmussen},
  {Severi}, {Blunier}, {Leuenberger}, {Fischer}, {Masson-Delmotte},
  {Chappellaz}, and {Wolff}}{{Bazin} et~al.}{2013}]{bazin2013}
{Bazin}, L., A.~{Landais}, B.~{Lemieux-Dudon}, H.~{Toy\'{e} Mahamadou Kele},
  D.~{Veres}, F.~{Parrenin}, P.~{Martinerie}, C.~{Ritz}, E.~{Capron}, V.~Y.
  {Lipenkov}, M.-F. {Loutre}, D.~{Raynaud}, B.~M. {Vinther}, A.~M. {Svensson},
  S.~O. {Rasmussen}, M.~{Severi}, T.~{Blunier}, M.~C. {Leuenberger},
  H.~{Fischer}, V.~{Masson-Delmotte}, J.~A. {Chappellaz}, and E.~W. {Wolff}
  (2013).
\newblock An optimized multi-proxy, multi-site {A}ntarctic ice and gas orbital
  chronology ({AICC2012}): 120-800 ka.
\newblock {\em Climate of the Past\/}~{\em 9\/}(4), 1715--1731.

\bibitem[\protect\citeauthoryear{Bradley}{Bradley}{2015}]{bradley}
Bradley, R.~S. (2015).
\newblock {\em Paleoclimatology: Reconstructing Climates of the Quaternary\/}
  (Third ed.).
\newblock Academic Press, Oxford.

\bibitem[\protect\citeauthoryear{Casado, Landais, Picard, Muench, Laepple,
  Stenni, Dreossi, Ekaykin, Arnaud, Genthon, Touzeau, Masson-Delmotte, and
  Jouzel}{Casado et~al.}{2018}]{Casado}
Casado, M., A.~Landais, G.~Picard, T.~Muench, T.~Laepple, B.~Stenni,
  G.~Dreossi, A.~Ekaykin, L.~Arnaud, C.~Genthon, A.~Touzeau,
  V.~Masson-Delmotte, and J.~Jouzel (2018).
\newblock Archival processes of the water stable isotope signal in east
  {A}ntarctic ice cores.
\newblock {\em The Cryosphere\/}~{\em 12}, 1745–--1766.

\bibitem[\protect\citeauthoryear{Dai, Song, Barber, and Raskutti}{Dai
  et~al.}{2020}]{dai2020bias}
Dai, R., H.~Song, R.~F. Barber, and G.~Raskutti (2020).
\newblock The bias of isotonic regression.
\newblock {\em Electronic journal of statistics\/}~{\em 14\/}(1), 801--834.

\bibitem[\protect\citeauthoryear{DiCiccio and Efron}{DiCiccio and
  Efron}{1996}]{Efron1996}
DiCiccio, T.~J. and B.~Efron (1996).
\newblock Bootstarp confidence interval.
\newblock {\em Statistical Science\/}~{\em 11}, 189--228.

\bibitem[\protect\citeauthoryear{Efron and Tibshirani}{Efron and
  Tibshirani}{1993}]{Efron1993}
Efron, B. and R.~J. Tibshirani (1993).
\newblock {\em An Introduction to Bootstrap}.
\newblock Florida, USA: Chapman and Hall CRC.

\bibitem[\protect\citeauthoryear{Fudge, Markle, Cuffey, Buizert, Taylor, Steig,
  Waddington, Conway, and Koutnik}{Fudge et~al.}{2016}]{Fudge2016}
Fudge, T.~J., B.~R. Markle, K.~M. Cuffey, C.~Buizert, K.~C. Taylor, E.~J.
  Steig, E.~D. Waddington, H.~Conway, and M.~Koutnik (2016).
\newblock Variable relationship between accumulation and temperature in {W}est
  {A}ntarctica for the past 31,000 years.
\newblock {\em Geophysical Research Letters\/}~{\em 43\/}(8), 3795--3803.

\bibitem[\protect\citeauthoryear{Ghosh and Chaudhuri}{Ghosh and
  Chaudhuri}{2004}]{GhoshChau2004}
Ghosh, A.~K. and P.~Chaudhuri (2004).
\newblock Optimal smoothing in kernel discriminant analysis.
\newblock {\em Statistica Sinica\/}~{\em 14\/}(2), 457--483.

\bibitem[\protect\citeauthoryear{Greve and Blatter}{Greve and
  Blatter}{2009}]{greve2009}
Greve, R. and H.~Blatter (2009).
\newblock {\em Dynamics of Ice Sheets and Glaciers}.
\newblock Berlin: Springer.

\bibitem[\protect\citeauthoryear{Guntuboyina and Sen}{Guntuboyina and
  Sen}{2018}]{guntuboyina2018}
Guntuboyina, A. and B.~Sen (2018).
\newblock Nonparametric shape-restricted regression.
\newblock {\em Statistical Science\/}~{\em 33\/}(4), 568--594.

\bibitem[\protect\citeauthoryear{H{\"a}rdle}{H{\"a}rdle}{1990}]{Hardle1990}
H{\"a}rdle, W. (1990).
\newblock {\em Applied Nonparametric Regression}.
\newblock Cambridge, England: Cambridge University Press.

\bibitem[\protect\citeauthoryear{Jennrich}{Jennrich}{1969}]{jennrich1969}
Jennrich, R.~I. (1969).
\newblock Asymptotic properties of non-linear least squares estimators.
\newblock {\em The Annals of Mathematical Statistics\/}~{\em 40\/}(2),
  633--643.

\bibitem[\protect\citeauthoryear{{Jouzel} and {Masson-Delmotte}}{{Jouzel} and
  {Masson-Delmotte}}{2007}]{Jouzel2007}
{Jouzel}, J. and V.~{Masson-Delmotte} (2007).
\newblock {EPICA Dome C Ice Core 800KYr deuterium data and temperature
  estimates}.
\newblock Supplement to: Jouzel, Jean et al. (2007): Orbital and millennial
  Antarctic climate variability over the past 800,000 years. Science,
  317(5839), 793-797.

\bibitem[\protect\citeauthoryear{Kessler}{Kessler}{1997}]{kessler1997}
Kessler, M. (1997).
\newblock Estimation of an ergodic diffusion from discrete observations.
\newblock {\em Scandinavian Journal of Statistics\/}~{\em 24\/}(2), 211--229.

\bibitem[\protect\citeauthoryear{Khasminskii}{Khasminskii}{2012}]{Hasminskii1980}
Khasminskii, R. (2012).
\newblock {\em Stochastic Stability of Differential Equations\/} (Second ed.).
\newblock Berlin: Springer-Verlag.

\bibitem[\protect\citeauthoryear{Kr\"{a}mer}{Kr\"{a}mer}{1983}]{Kramer83}
Kr\"{a}mer, W. (1983).
\newblock High correlation among errors and the efficiency of ordinary least
  squares in linear models.
\newblock {\em Statistische Hefte\/}~{\em 25}, 135–--142.

\bibitem[\protect\citeauthoryear{Lahiri}{Lahiri}{2003}]{Lahiri2003}
Lahiri, S.~N. (2003).
\newblock {\em Resampling Methods for Dependent Data}.
\newblock New York: Springer.

\bibitem[\protect\citeauthoryear{Lu and Park}{Lu and Park}{2019}]{lu2019}
Lu, Y. and J.~Y. Park (2019).
\newblock Estimation of longrun variance of continuous time stochastic process
  using discrete sample.
\newblock {\em Journal of Econometrics\/}~{\em 210\/}(2), 236--267.

\bibitem[\protect\citeauthoryear{Luss, Rosset, and Shahar}{Luss
  et~al.}{2012}]{Luss2012}
Luss, R., S.~Rosset, and M.~Shahar (2012).
\newblock Efficient regularized isotonic regression with application to
  gene–gene interaction search.
\newblock {\em Annals of Applied Statistics\/}~{\em 6\/}(1), 253--–283.

\bibitem[\protect\citeauthoryear{Markle and Steig}{Markle and
  Steig}{2022}]{Markle}
Markle, B.~R. and E.~J. Steig (2022).
\newblock Improving temperature reconstructions from ice-core water-isotope
  records.
\newblock {\em Climate of the Past\/}~{\em 18}, 1321–--1368.

\bibitem[\protect\citeauthoryear{Masson-Delmotte, Zhai, Pirani, Connors,
  P\'{e}an, Berger, Caud, Chen, Goldfarb, Gomis, Huang, Leitzell, Lonnoy,
  Matthews, Maycock, Waterfield, Yelek\c{c}i, Yu, and Zhou}{Masson-Delmotte
  et~al.}{2021}]{IPCC}
Masson-Delmotte, V., P.~Zhai, A.~Pirani, S.~Connors, C.~P\'{e}an, S.~Berger,
  N.~Caud, Y.~Chen, L.~Goldfarb, M.~Gomis, M.~Huang, K.~Leitzell, E.~Lonnoy,
  J.~Matthews, T.~Maycock, T.~Waterfield, O.~Yelek\c{c}i, R.~Yu, and B.~Zhou
  (2021).
\newblock {\em Climate Change 2021: The Physical Science Basis}.
\newblock Cambridge: Cambridge University Press.
\newblock Working Group I Contribution to the IPCC Sixth Assessment Report.

\bibitem[\protect\citeauthoryear{M\"unch and Laepple}{M\"unch and
  Laepple}{2018}]{Munch}
M\"unch, T. and T.~Laepple (2018).
\newblock What climate signal is contained in decadal- to centennial-scale
  isotope variations from {A}ntarctic ice cores?
\newblock {\em Climate of the Past\/}~{\em 14}, 2053–--2070.

\bibitem[\protect\citeauthoryear{Niculescu-Mizil and Caruana}{Niculescu-Mizil
  and Caruana}{2005}]{Nicules2005}
Niculescu-Mizil, A. and R.~Caruana (2005).
\newblock Predicting good probabilities with supervised learning.
\newblock In {\em ICML'05: Proceedings of the 22nd International Conference on
  Machine learning}.

\bibitem[\protect\citeauthoryear{Parrenin, Barnola, Beer, Blunier, Castellano,
  Chappellaz, Dreyfus, Fischer, Fujita, Jouzel, Kawamura, Lemieux-Dudon,
  Loulergue, Masson-Delmotte, Narcisi, Petit, Raisbeck, Raynaud, Ruth,
  Schwander, Severi, Spahni, Steffensen, Svensson, Udisti, Waelbroeck, and
  Wolff}{Parrenin et~al.}{2007}]{Parrenin2007}
Parrenin, F., J.-M. Barnola, J.~Beer, T.~Blunier, E.~Castellano, J.~Chappellaz,
  G.~Dreyfus, H.~Fischer, S.~Fujita, J.~Jouzel, K.~Kawamura, B.~Lemieux-Dudon,
  L.~Loulergue, V.~Masson-Delmotte, B.~Narcisi, J.-R. Petit, G.~Raisbeck,
  D.~Raynaud, U.~Ruth, J.~Schwander, M.~Severi, R.~Spahni, J.~P. Steffensen,
  A.~Svensson, R.~Udisti, C.~Waelbroeck, and E.~Wolff (2007).
\newblock The {EDC3} chronology for the {EPICA} {D}ome {C} ice core.
\newblock {\em Climate of the Past\/}~{\em 3\/}(3), 485--497.

\bibitem[\protect\citeauthoryear{Petit, Jouzel, Raynaud, Barkov, Barnola,
  Basile, Bender, Chappellaz, Davis, Delaygue, Delmotte, Kotlyakov, Legrand,
  Lipenkov, Lorius, P\'{e}pin, Ritz, Saltzman, and Stievenard}{Petit
  et~al.}{1999}]{Petit1999}
Petit, J.~R., J.~Jouzel, D.~Raynaud, N.~I. Barkov, J.-M. Barnola, I.~Basile,
  M.~Bender, J.~Chappellaz, M.~Davis, G.~Delaygue, M.~Delmotte, V.~M.
  Kotlyakov, M.~Legrand, V.~Y. Lipenkov, C.~Lorius, L.~P\'{e}pin, C.~Ritz,
  E.~Saltzman, and M.~Stievenard (1999).
\newblock Climate and atmospheric history of the past 420,000 years from the
  {V}ostok ice core, {A}ntarctica.
\newblock {\em Nature\/}~{\em 399\/}(6735), 429--436.

\bibitem[\protect\citeauthoryear{Rapp}{Rapp}{2019}]{Rapp2019}
Rapp, D. (2019).
\newblock {\em Ice Core Data}, pp.\  83--118.
\newblock Cham: Springer International Publishing.

\bibitem[\protect\citeauthoryear{Robertson, Wright, and Dykstra}{Robertson
  et~al.}{1988}]{RWR}
Robertson, T., F.~T. Wright, and R.~Dykstra (1988).
\newblock {\em Order Restricted Statistical Inference}.
\newblock New York: Wiley.

\bibitem[\protect\citeauthoryear{Robinson}{Robinson}{1997}]{robinson1997large}
Robinson, P.~M. (1997).
\newblock Large-sample inference for nonparametric regression with dependent
  errors.
\newblock {\em The Annals of Statistics\/}~{\em 25\/}(5), 2054--2083.

\bibitem[\protect\citeauthoryear{Ruppert, Wand, and Carroll}{Ruppert
  et~al.}{2003}]{ruppert2003semiparametric}
Ruppert, D., M.~P. Wand, and R.~J. Carroll (2003).
\newblock {\em Semiparametric Regression}.
\newblock Number~12. Cambridge university press.

\bibitem[\protect\citeauthoryear{Scambos and Shuman}{Scambos and
  Shuman}{2016}]{scambos2016}
Scambos, T. and C.~Shuman (2016).
\newblock Comment on `{M}ass gains of the {A}ntarctic ice sheet exceed
  losses’ by {H. J. Z}wally and others.
\newblock {\em Journal of Glaciology\/}~{\em 62\/}(233), 599–603.

\bibitem[\protect\citeauthoryear{Sen, Banerjee, and Woodroofe}{Sen
  et~al.}{2010}]{Bodhi2010}
Sen, B., M.~Banerjee, and M.~Woodroofe (2010).
\newblock Inconsistency of bootstrap: The {G}renander estimator.
\newblock {\em Annals of Statistics\/}~{\em 38}, 1953--1977.

\bibitem[\protect\citeauthoryear{Siegert}{Siegert}{2003}]{Siegert2003}
Siegert, M.~J. (2003).
\newblock Glacial–interglacial variations in central {E}ast {A}ntarctic ice
  accumulation rates.
\newblock {\em Quaternary Science Reviews\/}~{\em 22}, 741--750.

\bibitem[\protect\citeauthoryear{Tang and Chen}{Tang and Chen}{2009}]{Tang2009}
Tang, C.~Y. and S.~X. Chen (2009).
\newblock Parameter estimation and bias correction for diffusion processes.
\newblock {\em Journal of Econometrics\/}~{\em 149\/}(1), 65--81.

\bibitem[\protect\citeauthoryear{Veres, Bazin, Landais, Toye Mahamadou~Kele,
  Lemieux-Dudon, Parrenin, Martinerie, Blayo, Blunier, Capron, Chappellaz,
  Rasmussen, Severi, Svensson, Vinther, and Wolff}{Veres et~al.}{2013}]{VERES}
Veres, D., L.~Bazin, A.~Landais, H.~Toye Mahamadou~Kele, B.~Lemieux-Dudon,
  F.~Parrenin, P.~Martinerie, E.~Blayo, T.~Blunier, E.~Capron, J.~Chappellaz,
  S.~O. Rasmussen, M.~Severi, A.~Svensson, B.~Vinther, and E.~W. Wolff (2013).
\newblock The antarctic ice core chronology (aicc2012): {A}n optimized multi-
  parameter and multi-site dating approach for the last 120 thousand years.
\newblock {\em Climate of the Past\/}~{\em 9}, 1733--1748.

\bibitem[\protect\citeauthoryear{Wang}{Wang}{2013}]{Wang2013}
Wang, P.~K. (2013).
\newblock {\em Physics and Dynamics of Clouds and Precipitation}.
\newblock Cambridge, England: Cambridge University Press.

\bibitem[\protect\citeauthoryear{Yang, Yao, Wang, and Gou}{Yang
  et~al.}{2006}]{Yang}
Yang, M., T.~Yao, H.~Wang, and X.~Gou (2006).
\newblock Correlation between precipitation and temperature variations in the
  past 300 years recorded in {G}uliya ice core, {C}hina.
\newblock {\em Annals of Glaciology\/}~{\em 43\/}(1), 137--141.

\bibitem[\protect\citeauthoryear{Zwally, Li, Robbins, Saba, Yi, and
  Brenner}{Zwally et~al.}{2015}]{zwally2015}
Zwally, H.~J., J.~Li, J.~W. Robbins, J.~L. Saba, D.~Yi, and A.~C. Brenner
  (2015).
\newblock Mass gains of the {A}ntarctic ice sheet exceed losses.
\newblock {\em Journal of Glaciology\/}~{\em 61\/}(230), 1019–1036.

\end{thebibliography}
\end{document}